\documentclass[11pt]{article}
\usepackage{geometry}
\usepackage{amsthm}
\usepackage{amsfonts} 
\usepackage{amsmath}
\usepackage{amssymb}
\usepackage{thmtools}
\usepackage{hyperref} 
\usepackage{changepage} 
\usepackage{enumerate}

\usepackage[ruled,vlined]{algorithm2e} 

\usepackage{subcaption}

\usepackage{pgfplots} 
\pgfplotsset{compat=1.6}

\newtheorem{theorem}{Theorem}[section]
\newtheorem{lemma}[theorem]{Lemma} 
 
\newtheorem{remark}[theorem]{Remark} 
\newtheorem{definition}[theorem]{Definition} 
 
\newtheorem{observation}[theorem]{Observation} 

\newcommand{\ceil}[1]{\lceil #1 \rceil}
\newcommand{\floor}[1]{\lfloor #1 \rfloor}

\newcommand{\minCarboricity}[0]{(2+\epsilon)\alpha-1} 
\newcommand{\minCarboricityExample}[0]{2\alpha} 
\newcommand{\minCarboricityB}[0]{\frac{C}{2\alpha-1}} 

\newcommand{\lowerBoundExplicitExpression}[0]{$\Omega \Big (
\frac{
    \log \frac{nC}{(\Delta-C)^2}
    }{
    \log \frac{\Delta+C}{\Delta-C}
    }
     \Big )$}
\newcommand{\separationExpression}[0]{\lowerBoundExplicitExpression}

\begin{document}

\title{Beyond Vizing Chains: Improved Recourse in Dynamic Edge Coloring}
\author{
Yaniv Sadeh\footnote{Tel Aviv University, Israel. yanivsadeh@tau.ac.il.}
\and
Haim Kaplan\footnote{Tel Aviv University, Israel. haimk@tau.ac.il.}
}

\maketitle

\begin{abstract}

We study the maintenance of a $(\Delta + C)$-edge-coloring ($C\ge 1$) in a fully dynamic graph $G$ with maximum degree $\Delta$. We focus on minimizing \emph{recourse} which equals the number of recolored edges per edge insertions and deletions.

\looseness=-1
We present a new recoloring technique based on an object which we call a \emph{shift-tree}.
This object tracks multiple possible recolorings of $G$ and enables us to maintain a proper coloring with small recourse in polynomial time. This technique relies on shifting colors over a path of edges, but unlike many other algorithms, we do not use \emph{fans} and \emph{alternating bicolored paths}.

\looseness=-1
We combine the shift-tree with additional techniques to obtain an algorithm
with a recourse of $O \big ( \frac{\log n}{\log \frac{\Delta+C}{\Delta-C}} \cdot \frac{1}{C/\Delta - 1/\phi} \big )$ for all
$C \ge \frac{\Delta}{\phi}+0.724$, which is \emph{tight} for $C - \frac{\Delta}{\phi} = \Omega(\Delta)$ (say, $C \ge 0.62 \Delta + 1$) and $\Delta-C = O(n^{1-\delta})$ for some $\delta \in (0,1]$.
Our algorithm is the first deterministic algorithm to establish tight bounds
for large palettes, and the first to do so when $\Delta - C = o(\Delta)$.
This result settles the theoretical complexity of the recourse for large palettes.
Furthermore, we believe that viewing the possible shifts as a tree can lead to similar tree-based techniques that extend to lower values of $C$, and to improved update times.

A second application of the shift-tree is to graphs with low arboricity $\alpha$. Previous works~\cite{ArboricitySWAT2024_sayan,ArboricitySWAT2024_Christiansen} achieve $O(\frac{1}{\epsilon} \log n)$ recourse per update with $C \ge (4+\epsilon) \alpha$, and we improve by achieving the same recourse while only requiring $C \ge (2+\epsilon)\alpha - 1$. We note that this result is \emph{$\Delta$-adaptive}, that is, at any time the coloring only uses $\Delta_t + C$ colors where $\Delta_t$ is the current maximum degree (only a bound on $\alpha$ is required).

Trying to understand the limitations of our technique, and shift-based algorithms in general, we show a separation between the recourse achievable by algorithms that only shift colors along a path, and more general algorithms such as those using the \emph{Nibbling Method} \cite{NibblingCycles2024,NibblingMethod2021}.
\end{abstract}

\pagebreak

\setcounter{tocdepth}{2}
\tableofcontents


\section{Introduction}
\label{section_introduction}

Edge coloring is a fundamental problem in combinatorics focused on assigning colors to the edges of a graph $G$.\footnote{All graphs in this paper are \emph{simple graphs}, without loops and without parallel edges.} The primary goal is to ensure that no two edges sharing a common vertex (i.e., adjacent edges, or neighbours) receive the same color.

\looseness=-1
Let $\Delta$ denote the maximum degree of a graph. It is intuitive that at least $\Delta$ colors are necessary, as all edges incident to a vertex of degree $\Delta$ must have distinct colors. A significant breakthrough came in the 1960s in a series of works by Vizing. \emph{Vizing's Theorem}~\cite{Vizing1964} establishes that $\Delta+1$ colors are always sufficient. Vizing also identified specific graph classes where $\Delta$ colors are indeed enough~\cite{Vizing1965_10up,Vizing1965_8up}. Since Vizing's pioneering work, extensive research has explored edge coloring in various computational models. A recurring theme across these models is a trade-off: algorithms that use fewer colors are often less efficient, and vice versa.

In the standard offline and centralized model, where we are given an uncolored graph and need to color it, Vizing's results imply an algorithm that runs in $O(mn)$ time given $\Delta+1$ colors (or more).\footnote{As is commonly denoted in the literature: $n$ is the number of vertices and $m$ is the number of edges in the graph $G$.} A long line of work improves that, with the most recent results giving an algorithm for $\Delta+1$ edge coloring that runs in $O(m \log \Delta)$ time \cite{STOC2025OfflineColoring}, and an 
algorithm for
$(1+\epsilon)\Delta$ edge coloring
that runs in $O(m \log(1/\epsilon) / \epsilon^4)$ time
\cite{2025LinearEdgeColoring}. See these papers for additional references to earlier work, including deterministic results, in particular Table 1 in \cite{2025LinearEdgeColoring}. Generally speaking, edge coloring algorithms color the edges in some order, and sometimes coloring an edge $e$ (or a batch of edges) requires changing the color of previously colored edges. We refer to changing the colors as \emph{recoloring}. The set of recolored edges is sometimes referred to in the literature as the \emph{augmenting component}.

The edge coloring problem was also studied intensively in the distributed model, where each node has to compute a local coloring for its own edges that is consistent with the rest of the graph and in particular with its neighbours. In this model, one is interested in the number of communication rounds required to achieve a proper coloring. See \cite{LowerBoundRecourse2018} for a table that summarizes notable distributed edge coloring algorithms. More recent results are due to \cite{SuVuTruncatedVizing2019,Bernshteyn2022SmallRecoloring,MultistepVizing2023,Davies2023DistributedEdgeColoring}. While this model is very different from the (offline) centralized model, many breakthrough results in this model guarantee that small augmenting components exist, and thus help improve results in the centralized model.

Another computation model with some relevance to our discussion is the \emph{online model}. In this model the graph is not given in advance, but is revealed incrementally
and edges have to be colored irrevocably when they are revealed. In other words, recoloring is not allowed. This study was initiated by \cite{BarNoyOnline1992}. See the work of \cite{blikstad2024onlineedgecoloringnearly, FOCS2025ONlineEdgeColoring} for recent results and citations of additional work. 

Our main interest is also to study this problem in an online model but when recoloring is allowed. This is sometimes referred to as the \emph{online with recourse} model, where the \emph{recourse} is the number of recolorings an algorithm makes, that is, changes of edge colors following an insertion or a deletion. This model is also referred to as the \emph{dynamic edge coloring} model. In this model, the graph is fully dynamic, and edges may come and go, while the maximum degree is guaranteed to not exceed $\Delta$ and the number of colors is fixed. One may consider two parameters of interest: the recourse and the running time per update, and we focus on studying the recourse. 

Assuming an explicit representation of edge colors, the recourse lower bounds the running time. Therefore,
understanding the required recourse is interesting since it is an inherent hurdle on the running time if explicit coloring is desired. Furthermore, different applications may be more sensitive to recourse than to running time. This may happen if color changes require an expensive external action such as expensive accesses to a database, re-scheduling meetings, or a physical update to links~\cite{KMatchings2022,sadeh2024CachingInMatchingsICALP,bMatchingOpticalSwitches2025}.

\looseness=-1
The work in this dynamic model is relatively more recent compared to the other models. \cite{Bernshteyn2022SmallRecoloring} was the first to achieve recourse that is logarithmic in $n$, concretely $O(\Delta^6 \log^2 n)$ recourse (worst-case per update), using a randomized algorithm based on multi-step Vizing chains, with $\Delta+1$ colors. \cite{MultistepVizing2023} improved the recourse to $O(\Delta^7 \log n)$, using a deterministic algorithm. This recourse is tight for $\Delta = O(1)$.
Given $(1+\epsilon)\Delta$ colors, \cite{2025LinearEdgeColoring} achieve $O(\epsilon^{-4} \log n)$ recourse using a randomized algorithm. This recourse is tight for $\epsilon = \Theta(1)$ provided that $(1-\epsilon) = \Theta(1)$ (in terms of extra colors $C$, and $\Delta$, this is equivalent to $\Delta-C = \Theta(\Delta)$).\footnote{We always assume that $C \le \Delta - 2$ because if $C \ge \Delta-1$ recourse is never necessary.} \cite{NibblingMethod2021,NibblingCycles2024} remove the dependence on $n$ and achieve $O(\mathrm{poly}(1/\epsilon))$ recourse but in expectation per update (rather than worst-case), and there is an additional constraint that requires $\Delta$ to be at least $\Omega(\mathrm{poly}(1/\epsilon) \cdot \log n)$. Other recent works~\cite{ArboricitySWAT2024_sayan,ArboricitySWAT2024_Christiansen} focus on special families of graphs, with low arboricity $\alpha$, to achieve  $O(\mathrm{poly}(\log n))$ recourse using $\Delta + O(\alpha)$ colors. See Table~\ref{table_recourse_results} in Appendix~\ref{section_appendix_table_of_results} for a summary of these and other main recourse results, along with our new results.

\looseness=-1
It is worthwhile to note that when $\Delta = \Omega(\log n)$, recourse that depends on $\Delta$ may be reduced, by reducing the effective degree of the graph. One can choose a subset of colors, a \emph{sub-palette} $P$, and only work with the colors of $P$, the subgraph induced by edges colored by $P$, and the new uncolored inserted edge. In order to successfully work on this subgraph, its induced maximum degree $\Delta' \le |P|$ should depend also on the algorithm in use. For instance, if the algorithm requires $1$ color more than the maximum degree then we require $\Delta' \le |P|-1$. \cite{DynamicEdgeColoring2019DuanEtal} show that by sampling $P$ such that $|P| = \ceil{B \log n / \epsilon}$, then $\Delta' \le |P|-1$ with probability of at least $1-n^{-B/2+1}$. \cite{MultistepVizing2023} shows that by sampling a few more colors, $|P| = \Theta(\epsilon^{-1} \log n \log \Delta)$, one gets that every node $v$ has an available color among the first $\ceil{(1+\epsilon) \cdot deg(v)}$ colors (assuming a total order) with high probability.

\medskip

\textbf{Our contributions are as follows.} 

\begin{enumerate}
    \item We introduce the \emph{shift-tree} technique, which powers a class of color-shifting algorithms for
    recoloring edges in order to extend a partial coloring. 
    The shift-tree is defined for a partially colored graph. It spreads from a vertex that is adjacent to an uncolored edge, such that a path in the tree corresponds to a path in the graph over which we can shift colors by moving them from one edge to the next along the path (formally defined in Definition~\ref{definition_chain_shift} and Definition~\ref{definition_shift_tree_and_terminology}). This technique is interesting both for the results we derive from it (next items), and for being an unconventional approach: Most shift-based algorithms (though not all) use Vizing's approach with fans and bicolored paths.\footnote{A path is \emph{bicolored} if its edges only have $2$ colors. The colors must alternate in a properly edge colored graph.} This new tree-based approach may spark additional variants.
    
    \item \label{item_contribute_large_C} 
    We suggest a new technique to determine when the process of shifting colors on a path can be stopped with a proper coloring that extends to an uncolored edge. This technique is based on defining good and bad neighbours for each edge (Definition~\ref{definition_dangerous}), and combined with the shift-tree technique, it allows us to derive algorithms that are tailored to a rather large color palette. A part of the novelty is that we also loosen the common requirement for paths to be non-self-intersecting, and allow the last edge to be a revisited edge (Remark~\ref{remark_path_with_possible_repeat_edge}).
    The large palette is required to make the idea work, and it is an interesting open question whether one can generalize the technique so it applies to smaller palettes. Concretely, our first algorithm (Algorithm~\ref{algorithm_generic_shift_tree} calling Algorithm~\ref{algorithm_handle_leaves_generic}, Theorem~\ref{theorem_tight_worst_case_recourse}) 
     applies for $\Delta+C$ colors, where $C \ge \frac{\Delta}{\phi} + 0.724$ and $\phi$ is the golden ratio ($\frac{1}{\phi} \approx 0.618$). The second algorithm (Algorithm~\ref{algorithm_generic_shift_tree} calling Algorithm~\ref{algorithm_handle_leaves_DeltaMinus2}, Theorem~\ref{theorem_DeltaMinus2_colors}) is tailored to $C = \Delta-2$ to cover the case of $C=\Delta-2$ when $\Delta \in \{4,5,6,7\}$ which is not covered by the first algorithm. Both algorithms achieve a tight worst-case recourse of $O(\frac{\log n}{\log \frac{\Delta+C}{\Delta-C}})$, with respect  to the lower bound of \cite{LowerBoundRecourse2018}, subject to $\Delta-C = O(n^{1-\delta})$ for a constant $\delta > 0$ and $\frac{C}{\Delta} - \frac{1}{\phi} = \Omega(1)$. We clarify that while existing results are tight in partial domains, such as \cite{MultistepVizing2023} (deterministic) is tight for $\Delta = \Theta(1)$ with $C=1$, and \cite{2025LinearEdgeColoring} (randomized) is tight for any $\Delta$ when $C = \Theta(\Delta)$ \emph{and} $\Delta - C = \Theta(\Delta)$, our result (deterministic) shows tightness
     for all $C \ge 0.62\Delta + 1$. It is the first deterministic algorithm to be tight for this range of $C$, and also the first among randomized algorithms to be tight for $\Delta - C = o(\Delta)$, thus settling the theoretical complexity of the recourse for large palettes when $\Delta - C = O(n^{1-\delta})$.\footnote{Even though in practice  $C = \Theta(\Delta)$ colors may be large, the theoretical differences are still of interest. Some applications in other models even consider super-linear palettes, for example: the classical distributed result of \cite{LinialDistributedDeltaSquared}, and recent streaming studies~\cite{StreamingChechikEtal2024ICALP,GhoshEtalWStreamingEdgeColoringICALP2024,StreamingEdgeColoringICALP2025}.}

    \item We study the shift-tree technique also in the context of graphs with low arboricity, denoted by $\alpha$. In this case, no additional idea is necessary other than the shift-tree, and we show that for $C \ge \minCarboricity{}$ we get recourse of $O(\frac{1}{\epsilon} \log n)$ (Theorem~\ref{theorem_adaptive_Delta_for_low_arboricity}). This improves over previous results~\cite{ArboricitySWAT2024_sayan,ArboricitySWAT2024_Christiansen} that achieve the same recourse but require $C \ge (4+\epsilon) \alpha$. We also note that the algorithm we get is $\Delta$-adaptive, that is, it uses $\Delta_t + C$ colors at time $t$ where the maximum degree in the graph is $\Delta_t$ without requiring a known bound $\Delta$ in advance. We complement this with an almost matching worst-case lower bound that shows recourse of $\Omega(\frac{\log n}{\log \Delta})$ even when $C = \Delta-2$ (Remark~\ref{remark_generalized_lower_bound}). This lower bound generalizes that of \cite{LowerBoundRecourse2018} to graphs with low arboricity.

    \item Most algorithms in the literature, including our own, are based on shifting colors (see Definition~\ref{definition_chain_shift}). We show that in a worst-case scenario, such algorithms might have a large recourse compared to more generic recoloring algorithms. We show this with a simple reduction that maintains a recourse lower bound, say of \cite{LowerBoundRecourse2018}, to hold against shift-based algorithms, while the actual achievable recourse is $O(1)$ (Theorem~\ref{theorem_gap_shift_recoloring_versus_general_separation_simple}).
\end{enumerate}

\textbf{Paper organization:}
Section~\ref{section_overview} motivates our contributions and presents our technique in more details, though due to space constraints formal proofs and some exact definitions are deferred. It is designed to convey our main idea without too many technical details. Section~\ref{section_summary_conclusion_open_questions} concludes with a discussion and interesting open questions. Section~\ref{section_appendix_table_of_results} summarizes in a table the main recourse results, including our new ones, and also mentions related work in a broader context of edge coloring that was not mentioned in Section~\ref{section_introduction}. Then follow the formal discussions and proofs: Section~\ref{appendix_section_shift_tree_algorithms_overview} covers the shift-tree related results and Section~\ref{appendix_section_lower_bounds_and_separation} focuses on lower bounds.

\section{Overview of Technique and Results}
\label{section_overview}

In this section we lay the intuitive ground for our new approach and the results we derive from it. We consider a dynamic graph setting, where the number of vertices is fixed, denoted by $n$, and we insert and delete edges. Formally, denote the graph at step $t$ by $G^t$ and its number of edges by $m_t$, initially $G^0$ has no edges ($m_0=0$). We assume that $G^t$ is always a simple graph. Note that $G^t$ and $G^{t+1}$ differ by exactly one edge, the one that was inserted (missing from $G^t$) or was deleted (missing from $G^{t+1}$). Like most edge coloring algorithms, we maintain an explicit proper edge coloring. Since our algorithms only depend on the current coloring, we omit the time superscript and simply refer to the (current) graph as $G$, with (fixed) $n$ vertices and (variable) $m$ edges. Most results in this work only update colors upon insertions so deletions become trivial.

We are guaranteed that the maximum degree of any vertex, at any time, does not exceed $\Delta$. We have $\Delta+C$ colors at our disposal, usually $C \ge 1$, and our objective is to maintain a proper $\Delta+C$ edge coloring of the graph. We are allowed to change the color of edges in order to maintain a proper coloring, and whenever we change the color of an edge we pay $1$ for it. The total number of color changes is called the \emph{recourse}, and our goal in this paper is to minimize this quantity as a function of $n$, $\Delta$ and $C$. Note that the recourse lower bounds the running time of an algorithm that maintains an explicit proper edge coloring.

We can properly color a new edge $e$ by either (1) finding a color in the palette not used by any neighbour of $e$, or by (2) introducing a completely new color to the palette and using it to color $e$.
Extending the palette for every edge is of course wasteful, but it is useful when carefully applied. For example, the folklore greedy algorithm can be thought of as always looking for an available color, with a fallback to extend its palette if none can be found. It is easy to see that the palette size does not exceed $2\Delta-1$ in this case. The whole line of research on the \emph{online} model (mentioned in Section~\ref{section_introduction}) is based on the idea that the palette must increase when no free color is available. Another prominent example is the \emph{Nibbling Method} used by \cite{NibblingMethod2021,NibblingCycles2024} in the \emph{dynamic} model. Without getting into the details, in these algorithms, edges that fail to be colored by the existing palette are put in a \emph{failed} set, which is eventually colored greedily by a disjoint palette.

When extending the palette is not an option, and there is no free color for a new uncolored edge, we should design a way to recolor edges that is guaranteed to \emph{terminate} with a proper coloring that includes the new edge. A concrete class of such methods is based on recoloring bicolored paths, as in Vizing's original theorem~\cite{Vizing1964}. The subgraph induced by a pair of colors is a collection of paths and cycles. If we choose this pair properly, possibly by first applying a \emph{Vizing fan}, then we can guarantee that our uncolored edge $e$ connects two paths in this subgraph rather than closing a cycle. Thus, by flipping colors along one of these bicolored paths we can vacate one of these two colors for $e$. This high-level argument is the basis of any recoloring method of a Vizing-based algorithm, including \cite{Bernshteyn2022SmallRecoloring,MultistepVizing2023,2025LinearEdgeColoring}. What differentiates these algorithms is the way they find the terminal bicolored paths, in order to argue that the recourse is small.

\looseness=-1
A different recoloring method appears in recent works on coloring graphs with low arboricity $\alpha$~\cite{ArboricitySWAT2024_sayan,ArboricitySWAT2024_Christiansen}, using $\Delta + \Theta(\alpha)$ colors. The argument in these works does not consider bicolored paths but instead relies on a dynamic decomposition of the graph into layers \cite{LayersDecomposition2015}. The layering process recursively peels the graph, where vertices at level $0$ are those of degree at most $d = \Theta(\alpha)$, then these vertices and their edges are removed to compute recursively the vertices at level $1$ and so on. Consider coloring an edge $(u,v)$ such that $\ell \equiv level(v) \le level(u)$. The vertex $u$ has degree at most $\Delta$, and $v$ has at most $d$ edges to vertices in level $\ell$ or higher (including $(u,v)$). Therefore with $\Delta + d$ colors there must be a color we can use, which at best finishes the coloring, and at worst shifts the problem by stealing the color of an edge $(v,w)$ such that $level(w) < level(v)$. Since the level of the lowest endpoint of the uncolored edge monotonically decreases, the recoloring is guaranteed to finish, and the number of layers upper bounds the recourse. Note that for this argument we required more than $\Delta+1$ colors, and that this algorithm is clearly not based on bicolored paths.

\medskip

We define an object which we refer to as a \emph{shift-tree} (Definition~\ref{definition_shift_tree_and_terminology} below) which encompasses multiple recoloring options of a partially colored graph $G$ in order to extend its coloring to an uncolored edge $e$. The shift-tree turns out to be a powerful tool for designing dynamic edge coloring algorithms with low recourse, and in particular lets us derive tight worst-case recourse for a certain regime of parameters (Theorem~\ref{theorem_tight_worst_case_recourse}). The recolorings encoded in the shift-tree are \emph{shift-based}.

\begin{definition}[Chain, Useful, Shift-based]
\label{definition_chain_shift}
A sequence of edges $S \equiv [e_1,e_2,e_3,\ldots,e_k]$ is a \emph{chain} if for every $i \in [1,k-1]$ the edges $e_i$ and $e_{i+1}$ share a single endpoint (vertex). Given a graph $G$, a chain $S$ is \emph{shiftable} if for every prefix of $S$, shifting the colors over the edges such that $e_i$ is colored by the color of $e_{i+1}$ and leaving the last edge uncolored, gives a new proper partial coloring of $G$. If $e_1$ is not colored and after shifting according to $S$ we can color $e_k$ by some color not used by any of its neighbour edges, we say that $S$ is a \emph{useful} chain. We say that an algorithm is \emph{shift-based} if it only extends edge colorings by using shiftable chains.
\end{definition}


\looseness=-1
Definition~\ref{definition_chain_shift} allows a chain to stagnate at a vertex. The \emph{fan} object commonly used in the literature is such a chain. In contrast, we can think of a chain that is a path, in which case every edge $e_i$ shares different endpoints with $e_{i-1}$ and $e_{i+1}$. The bicolored paths in the literature are of this type, and in our paper, we only use chains that are paths but not necessarily bicolored. Hence, we will commonly use the term \emph{path}.
See Figure~\ref{figure_basic_example_path_fan_non_shift} for a basic example of fan, path, and non-shift-based recolorings. Furthermore, shift-based algorithms in the literature usually do not repeat edges ($i \ne j \Rightarrow e_i \ne e_j$). In contrast, our approach slightly relaxes this constraint, and the last edge in our path may also appear earlier on the path.

\begin{figure*}[t]
	\centering
     
    \begin{subfigure}[t]{.45\textwidth}
        \centering
        \includegraphics[width=0.75\textwidth]{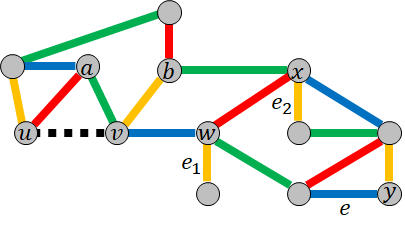}
        \subcaption{}
        \label{fig_basic_example1}
    \end{subfigure}
    \hfill 
    \begin{subfigure}[t]{.45\textwidth}
        \centering
        \includegraphics[width=0.75\textwidth]{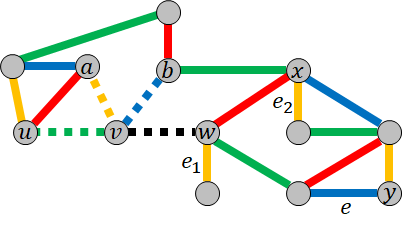}
        \subcaption{}
        \label{fig_basic_example2}
    \end{subfigure}
    \hfill 

    \bigskip
    
    \begin{subfigure}[t]{.45\textwidth}
        \centering
        \includegraphics[width=0.75\textwidth]{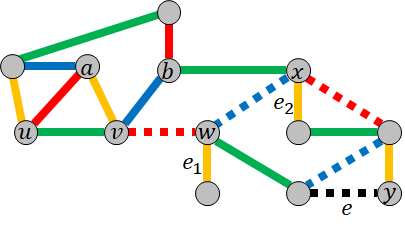}
        \subcaption{}
        \label{fig_basic_example3}
    \end{subfigure}
    \hfill 
    \begin{subfigure}[t]{.45\textwidth}
        \centering
        \includegraphics[width=0.75\textwidth]{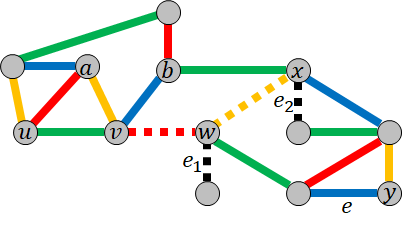}
        \subcaption{}
        \label{fig_basic_example4}
    \end{subfigure}
    \hfill 
    
    \caption{\small{
    A clarifying example for \emph{fans} (b), \emph{(bicolored) paths} (c), and \emph{non-shift-based} recolorings (d).
    (a) A new edge $(u,v)$ has been inserted.
    (b) We apply a fan around $v$ by shifting colors from one edge of $v$ to the next. The fan's order is $u,a,b,w$, now $(v,w)$ is uncolored.
    (c) We apply a bicolored path shift to color $(v,w)$ by shifting the red-blue path from $w$ to $y$ (possible since $v$ has no red edge). Now the final edge $e$ can be colored by red and recoloring is finished. The combined application of a fan and a bicolored path (i.e., (b)+(c)) is known as a \emph{Vizing chain}.
    (d) Instead of bicolored path shift, we could color $(v,w)$ by red and then $(w,x)$ by yellow. This does not shift the problem but duplicates it to $e_1$ \emph{and} $e_2$. But, we can still finish recoloring now: color $e_1$ by blue and $e_2$ by red.
    }}
	\label{figure_basic_example_path_fan_non_shift}
\end{figure*}

It is crucial to understand that every shift-based algorithm can be represented by a tree $T$ whose nodes correspond to edges of $G$, rooted at the uncolored edge, and the ability to shift color from $e$ to $e'$ implies that $e'$ is a parent of $e$ in $T$. All the possible paths which an algorithm can take implicitly induce $T$. In general, $T$ might be exponential in size even if a short useful chain exists because a single edge may correspond to many nodes (visited by multiple paths or multiple times on the same path). Therefore, a clever algorithm must handle this issue either by never explicitly exploring the tree as in \cite{2025LinearEdgeColoring}, or by strictly constraining the paths in a way that lets us compute the whole tree and find a useful chain as in the algorithm implied by the analysis in \cite{MultistepVizing2023}. Our shift-tree (definition below) encodes a collection of shiftable paths. Since it encodes only paths, and not fans, its nodes correspond to vertices of $G$ (rather than edges). The root of this tree is an endpoint of the uncolored edge, and each path to a leaf is a shiftable path.

\begin{definition}[Shift-tree]
\label{definition_shift_tree_and_terminology}
Let $G$ be a graph of maximum degree $\Delta$ that is $(\Delta+C)$ edge colored except for one uncolored edge $(u,v)$. We define the \emph{shift-tree} $T$ with respect to $G$ and its coloring  recursively level by level as follows. First, $v$ is the root of $T$, and we consider $u$ as the parent of $v$ (but $u$ is not part of $T$). Every leaf of $T$ is either \emph{active} or \emph{inactive}, initially the root $v$ is active.

Assume that we constructed levels $0$ to $i$ of $T$. We construct the level $i+1$ as follows. Inactive nodes of level $i$ remain leaves. Let $x$ be an active leaf with parent $p$. Loosely speaking, we define the children of $x$ to be its neighbours $y$ that can expand the shiftable path that ends in $(p,x)$. Formally, we make $y$ a child of $x$ if $(x,y) \in G$, and $color(x,y) \in A'(p)$ where $A'(p)$ are the colors available at $p$ after we shift the colors of the edges on the chain $P$ that is the path in $T$ from $(u,v)$ to $(p,x)$ such that $(p,x)$ becomes uncolored. A new child $y$ is set as active if this is the first time we see $y$ in $T$, otherwise it is set as inactive. The construction must stop at the first level that does not have an active node, but we can stop it earlier. Note that while we expand $T$ in a BFS manner, inactive copies of $y$ may appear much deeper than the active copy of $y$. We use the term \emph{vertex} to refer to $y \in G$, and \emph{node} to refer to specific copies of $y$ in $T$.

The \emph{depth} of a node $x \in T$ counts the edges on the path from the root of $T$ to $x$. $depth(T)$ is the maximal length of such a path, which is always well-defined. \emph{Level} $\ell$ is the set of nodes of $T$ at depth $\ell$. 

Given a shift-tree $T$ and a vertex $x \in G$ that appears in $T$, we define the \emph{skeleton} of $T$ with respect to $x$, denoted by $T^x$, to be the minimal connected subtree of $T$ that connects the root of $T$ and all the nodes that correspond to $x$ in $T$.
\end{definition}

\begin{figure*}[!htbp]
	\centering
    \begin{subfigure}[t]{.5\textwidth}
        \centering
        \includegraphics[width=0.85\textwidth]{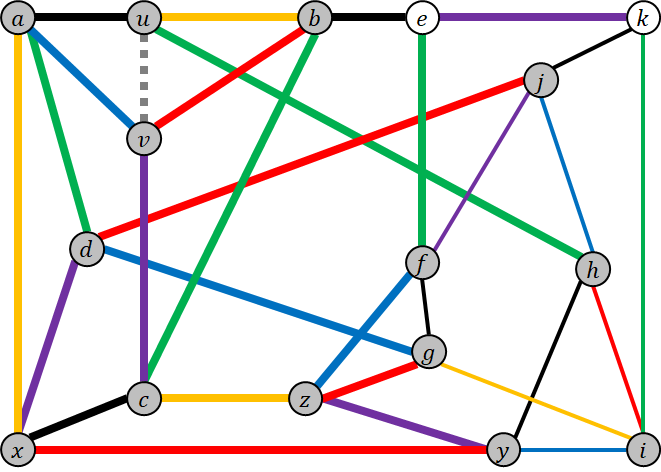}
        \subcaption{} 
        \label{fig_colored_G_example}
    \end{subfigure}
    
    \bigskip 
     
    \begin{subfigure}[t]{.3\textwidth}
        \centering
        \includegraphics[width=0.75\textwidth]{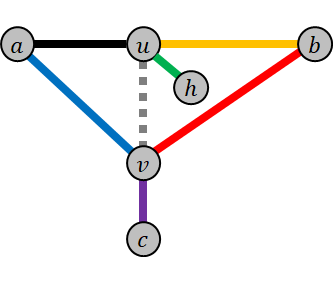}
        \subcaption{}
        \label{fig_colored_step1_graph}
    \end{subfigure}
    \hfill 
    \begin{subfigure}[t]{.15\textwidth}
        \centering
        \includegraphics[width=0.85\textwidth]{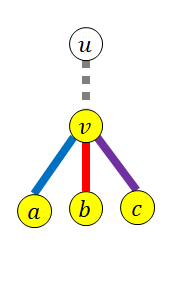}
        \subcaption{}
        \label{fig_colored_step1_tree}
    \end{subfigure}
    \hfill 
    \begin{subfigure}[t]{.35\textwidth}
        \centering
        \includegraphics[width=0.75\textwidth]{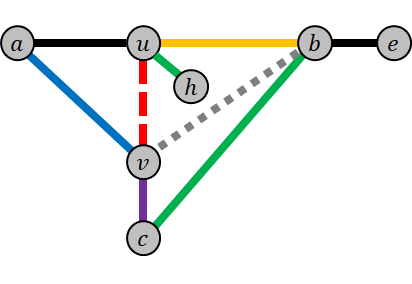}
        \subcaption{}
        \label{fig_colored_step2_graph}
    \end{subfigure}
    \hfill 
    \begin{subfigure}[t]{.15\textwidth}
        \centering
        \includegraphics[width=0.75\textwidth]{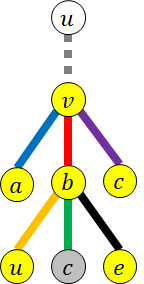}
        \subcaption{}
        \label{fig_colored_step2_tree}
    \end{subfigure}
    
    \bigskip
    
    \begin{subfigure}[t]{.38\textwidth}
        \centering
        \includegraphics[width=\textwidth]{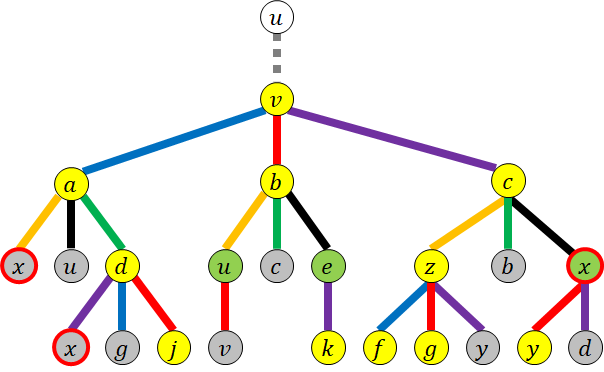}
        \subcaption{} 
        \label{fig_colored_shift_tree_T_example}
    \end{subfigure}
    \hfill 
    \begin{subfigure}[t]{.25\textwidth}
        \centering
        \includegraphics[width=.7\textwidth]{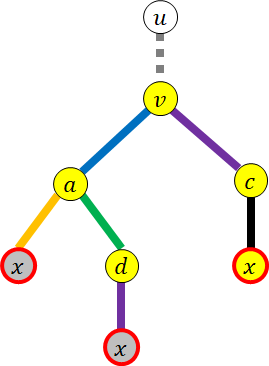}
        \subcaption{} 
        \label{fig_colored_shift_tree_T_skeleton_example}
    \end{subfigure}
    \hfill 
    \begin{subfigure}[t]{.25\textwidth}
        \centering
        \includegraphics[width=.65\textwidth]{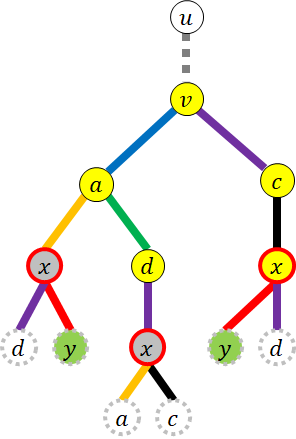}
        \subcaption{} 
        \label{fig_colored_shift_tree_T_expand_Tx}
    \end{subfigure}
    
    \caption{\small{An example of a colored graph $G$ and a shift-tree $T$ (Definition~\ref{definition_shift_tree_and_terminology}).
    (a) The graph $G$, with $\Delta=4$ and extra colors $C=2$ (total of $\Delta+C = 6$). White vertices have less than $\Delta$ edges, the edge $(u,v)$ is considered part of the graph but not yet colored. Fat edges (such as $(u,a)$) mark the ones that appear in the shift-tree in (f), thin edges (such as $(j,h)$) do not.
    (b) Neighbourhood $1$ of $(u,v)$ in the graph.
    (c) Depth $1$ shift-tree: the color of every edge of $v$ can be used for $(u,v)$.
    (d) \emph{If} we color $(u,v)$ by red, taking it from $(v,b)$, we now consider which colors available at $v$ can be used to color $(v,b)$.
    (e) The consideration in (d) implies the shift-tree subtree expansion of the node $b$. Note that $c$ is marked inactive (gray) because it will be expanded as a sibling of $b$.
    (f) A whole shift-tree $T$, expanded $3$ levels down. Every vertex of $G$ may correspond to multiple nodes of $T$ (for example $x$), but is only expanded with children once. A yellow/green node is an active copy of a vertex, gray nodes are inactive and remain leaves. The green nodes have less than $C+1 = 3$ children, which by Lemma~\ref{lemma_basic_shift_children} guarantees a useful path (Definition~\ref{definition_chain_shift}). For example, shifting over $(u,v,b,e)$ lets us color $(b,e)$ by either red or blue. Notice that $(u,v)$ appears as a red edge in $T$ because if we track the color shifting on the path to $u$, at the moment we expand children for $u$ the edge $(u,v)$ is red (stole the color from $(v,b)$).
    (g) The skeleton of $T$ when reduced only to paths that lead to $x$.
    (h) In the proof of Theorem~\ref{theorem_large_palette_general_parameterized} we expand each leaf of $T^x$ one step further regardless of being active or inactive. Then, some neighbour of $x$, say $y$, appears multiple times, so there are multiple possible paths to shift colors and reach the edge $(x,y)$.}}
	\label{figure_example_shift_tree_concepts}
\end{figure*}

Definition~\ref{definition_shift_tree_and_terminology} restricts the shiftable paths that we consider to be such that we do not expand a vertex more than once (a vertex appears at most once as an internal node). See also Figure~\ref{figure_example_shift_tree_concepts} for a visual example. This requirement is strict, but it simplifies the analysis by ensuring that the shift-tree does not include multiple copies of similar subtrees.

\begin{lemma}[Few Children Imply a Useful Shiftable Path]
\label{lemma_basic_shift_children}
A shift-tree $T$ with an active node $x$ that has at most $C$ children implies the existence of a useful shiftable path leading from the root of $T$ to the edge $(parent_T(x),x)$.
\end{lemma}

\begin{proof}
By definition, every path of edges in a shift-tree that begins at the root's edge $(u,v)$ (uncolored) and goes  down the tree to some edge $e = (p,x)$ 
is shiftable. Right after shifting the colors so that $(p,x)$ is uncolored, $|\overline{A}(p)| \le \deg_G(p) - 1$,\footnote{Denote for $v \in G$: $\deg_G(v)$ for its degree, $A(v)$/$\overline{A}(v)$ for its sets of available/used colors, respectively.} therefore $|A(p)| \ge \Delta + C + 1 - \deg_G(p) \ge C+1$, and thus if $x$ has less than $C+1$ children then it must be because $|A(x) \cap A(p)| \ge 1$, so the shiftable path to $(p,x)$ is useful.
\end{proof}

\looseness=-1
When the shift-tree contains a useful shiftable path, the recourse is bounded by its depth. A priori such paths are not guaranteed (in particular an active node with few children is not guaranteed), and we require an additional concept to find a useful shiftable path. The resulting recourse would be up to twice the depth of the shift-tree. When analyzing general graphs, our idea is completely different from the two methods described earlier of Vizing chains and strictly-decreasing layers. To best explain it, we first define good and bad neighbours.

\begin{definition}[Good/Bad Neighbours/Colors, Dangerous Edge]
\label{definition_dangerous}
Given a fully or partially colored graph $G$, fix an edge $e=(u,v)$, and a neighbour $e'=(u,x)$. If we shift the color of $e$ to $e'$ (ignoring possible color conflicts over the vertex $x$), and are able to color $e$ without any additional change to the colors of its neighbours, we say that $e'$ is a \emph{good neighbour} of $e$. If all the neighbours of $e$ are good we say that $e$ is \emph{not dangerous}. Otherwise, $e$ is \emph{dangerous} and every neighbour that is not good is a \emph{bad neighbour}. A \emph{bad color} (of $e$) is a color of one of the bad neighbours. Observe that for every bad color there are exactly two bad neighbours, one incident to $u$ and the other is incident to $v$. See Figure~\ref{figure_example_dangerous} for an example. Intuitively, a shiftable chain that ends in a non-dangerous edge is \emph{useful}. A shiftable chain that ends in a dangerous edge $e$ may also be useful, if it reaches $e$ from a good neighbour.

For $e$ to be dangerous, every color other than its own before the shift must appear in its neighbourhood, otherwise a missing color $c$ can be used to finish recoloring after we shift the color of $e$ to any neighbour $e'$. So if $e=(u,v)$ is dangerous, we have $(\deg_G(u)-1+\deg_G(v)-1) - (\Delta+C-1)$ \emph{bad colors} that are incident to both $u$ and $v$, excluding the color of $e$. This expression is upper bounded by $\Delta-1-C$, which we get when both $u$ and $v$ have degree $\Delta$.
\end{definition}

\begin{figure*}[t]
	\centering
    \begin{subfigure}[t]{.2\textwidth}
        \centering
        \includegraphics[width=0.8\textwidth]{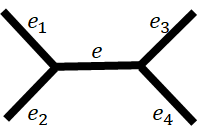}
        \subcaption{} 
        \label{fig_dangerous_no_colors}
    \end{subfigure}
    \hspace{4mm} 
    \begin{subfigure}[t]{.2\textwidth}
        \centering
        \includegraphics[width=0.8\textwidth]{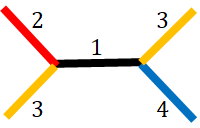}
        \subcaption{} 
        \label{fig_dangerous_colors_example1}
    \end{subfigure}
    \hspace{6mm} 
    \begin{subfigure}[t]{.2\textwidth}
        \centering
        \includegraphics[width=0.8\textwidth]{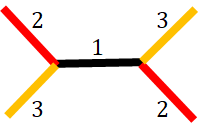}
        \subcaption{} 
        \label{fig_dangerous_colors_example2}
    \end{subfigure}
    
    \caption{
    \looseness=-1
    \small{An example for Definition~\ref{definition_dangerous}. (a) Edge $e$ and its neighbours. In (b) and (c) the numbers correspond to the colors, assume a total of $4$ colors ($\Delta=3$, $C=1$): black, red, orange, blue. (b) The edge $e$ is dangerous with \emph{bad color} orange whose \emph{bad neighbours} pair is $e_2,e_3$. The \emph{good neighbours} of $e$ are $e_1,e_4$. (c) The edge $e$ is not dangerous and all the neighbours are good (blue is always available).}}
	\label{figure_example_dangerous}
\end{figure*}

The new key idea is that with $\Delta+C$ colors for $C \ge 1$, every edge $e$ has at least $C$ good neighbours incident to each of its vertices. Therefore, if a shiftable path happens to arrive to $e$ via such a good neighbour, it is \emph{useful}.

As a first step to ensure reaching $e = (x,y)$ from a good neighbour, we would like to reach $e$ on multiple different paths each arriving from a different neighbour of a fixed endpoint of $e$. We count arrivals from $x$ and from $y$ separately because they might correspond to the same \emph{bad colors}. A combinatorial argument (Lemma~\ref{lemma_shift_tree_max_depth} and the note that follows it) shows that the maximum number of paths we can guarantee to arrive at \emph{some} vertex $x$ is $C$. Then, if we extend each of these paths by one more edge, we would reach many neighbours of $x$, in varying quantities (due to repetitions). The analysis shows that there has to be a neighbour $y$ that appears in approximately $N = \frac{C^2}{\Delta}$ of these extensions. In a worst-case scenario many of these paths might arrive to $(x,y)$ via a bad neighbour, but there are at most $\Delta-1-C$ bad neighbours of $(x,y)$ incident to $x$, so if $N \ge \Delta-C$ some shiftable path is guaranteed to be useful. The exact argument is a little more subtle, because a path that seems to arrive from a good neighbour initially might turn out to be bad if the path passes through $y$ before reaching $(x,y)$, and the technical nuances are given in Theorem~\ref{theorem_large_palette_general_parameterized} and its proof. However, at a high-level, we roughly get the condition that $\frac{C^2}{\Delta} \ge \Delta-C$, which when solved for $C$ gives $C \ge \frac{\Delta}{\phi}$ where $\phi \approx 1.618$ is the golden ratio. To clarify, when $C \ge \frac{\Delta}{\phi}$, we are guaranteed that the shift-tree constructively gives us a path over which to shift colors in order to extend the coloring to a new uncolored edge.
We can identify this path once we see a sufficient number of nodes that correspond to the same vertex, a parameter we denote by $b$ in Section~\ref{appendix_section_shift_tree_algorithms_overview} (such as in Algorithm~\ref{algorithm_generic_shift_tree}, Lemma~\ref{lemma_shift_tree_max_depth} and Theorem~\ref{theorem_large_palette_general_parameterized}). Recall that a large $b$ means that there are multiple shiftable paths from which to choose a useful path. We can upper bound the depth of the shift-tree (Lemma~\ref{lemma_shift_tree_max_depth}), and thus we get an upper bound on the recourse.

\begin{restatable}[Tight Recourse]{theorem}{theoremTightWorstCaseRecourse}
\label{theorem_tight_worst_case_recourse}
We can deterministically maintain a dynamic $\Delta+C$ edge coloring, $C \le \Delta-2$, in $O(m + \Delta^2)$ time with a worst-case recourse of $O(\frac{\log n}{\log {\frac{\Delta+C}{\Delta-C}}} \cdot \frac{1}{(C/\Delta) - 1/\phi})$ per insertion, for $C \ge \frac{\Delta}{\phi} + 0.724 \approx 0.618 \Delta$ ($\phi$ is the golden ratio). If $\Delta - C = O(n^{1-\delta})$ for a constant $\delta > 0$, and if $(C/\Delta) - 1/\phi = \Omega(1)$, the recourse is tight.
\end{restatable}

Extending our technique so it applies to a smaller color palette is a major challenge for future work, and might require additional tools or a generalization of the shift-tree (see Remark~\ref{remark_half_Delta_at_best}).

Theorem~\ref{theorem_tight_worst_case_recourse} claims a \emph{tight} recourse, which indeed up to a constant factor matches the lower bound due to \cite{LowerBoundRecourse2018}, stated as follows.

\begin{theorem}[Lower Bound by \cite{LowerBoundRecourse2018}]
\label{theorem_lower_bound_2018}
Let $1 \le C \le \Delta-2$. For any $n \ge \Delta + C + 2$ there exists a graph $G$ such that $|G|=n$, that is $(\Delta+C)$ edge colored except for one edge $e$, and any extension of the coloring to $e$ requires \lowerBoundExplicitExpression recourse. If $1 \le C \le \frac{\Delta}{3}$, it simplifies to $\Omega(\frac{\Delta}{C} \log \frac{Cn}{\Delta^2})$.\footnote{\cite{LowerBoundRecourse2018} and many of its citations give the expression $\Omega(\frac{\Delta}{C} \log \frac{Cn}{\Delta})$. This discrepancy in the denominator inside the logarithm is due to the following oversight: in the original notation, a recursive relation roughly of the form $n_{i+1} \approx \frac{k \cdot n_i}{k'}$ with $n_1 = k$ was used to conclude that $n_i = \Theta((k/k')^i)$, and get $\frac{\log{\frac{nC}{\Delta-C}}}{\log{\frac{\Delta+C}{\Delta-C}}}$ (before simplification). However, $n_i = \Theta((k/k')^i \cdot k')$. This extra $k' \equiv \frac{\Delta-C}{2}$ factor makes the difference in the numerator.}
\end{theorem}

Actually, \cite{LowerBoundRecourse2018} only states the simplified expression, and the general expression only appears in the proof. The simplification uses $\ln(1+x) \approx x$ for $x = \frac{2C}{\Delta-C} \le 1$ in the denominator. Because we consider large values of $C$, the tightness of Theorem~\ref{theorem_tight_worst_case_recourse} is with respect to the more general expression. The explicit comparison and the reason why we require $\Delta-C = O(n^{1-\delta})$ is in the proof of Theorem~\ref{theorem_tightness_domain}.\footnote{
Ironically, correcting the lower bound (footnote of Theorem~\ref{theorem_lower_bound_2018}) causes our imperfect tightness. Forgetting the factor $\frac{1}{\Delta-C}$ inside the logarithm ``makes'' the tightness apply without this restriction.} The condition on $\Delta-C$ can be informally interpreted as follows: The recourse is tight if $\Delta$ is not extremely large compared to $n$, that is, $\Delta = O(n^{1-\delta})$, or if $\Delta$ is extremely large but we have an almost matching extreme number of extra colors.

Note that Theorem~\ref{theorem_lower_bound_2018} is stated for a specific colored graph and is a worst-case claim. It does not say anything about amortized or expected recourse, which could be much lower. In addition, it might not apply to certain algorithms if they cannot reach the coloring required by the construction due to invariants that they maintain. As an example for an additional invariant, \cite{MultistepVizing2023,ArboricitySWAT2024_Christiansen} discuss scenarios where the palette available to an edge depends on the degrees of its endpoints (fewer neighbours, smaller palette). See also Section~\ref{subsection_adaptive_palettes}.

The shift-tree object can also be used to study graphs with low arboricity $\alpha$. Interestingly this analysis uses the same object (shift-tree), but its arguments are completely different and much simpler. We show that if $C = \Omega(\alpha)$ the shift-tree simply finds a short useful path, relying on Lemma~\ref{lemma_basic_shift_children}. Our shift-tree requires roughly half the number of extra colors compared to \cite{ArboricitySWAT2024_sayan,ArboricitySWAT2024_Christiansen} in order to achieve the same recourse, and does not require the decomposition into layers of \cite{LayersDecomposition2015}. Our simplicity comes at the price of a slower running time, polynomial rather than poly-logarithmic. The following theorem formalizes our result.

\begin{restatable}[Low Arboricity Graphs]{theorem}{theoremLowArboricity}
\label{theorem_low_arboricity}
Let $\alpha$ be an upper bound on the arboricity of a dynamic graph $G$ and let $\epsilon > 0$. If $C \ge \minCarboricity{}$, Algorithm~\ref{algorithm_generic_shift_tree} deterministically maintains a dynamic $\Delta + C$ edge coloring in $O(m)$ time with a worst-case recourse
$O(\log_{1 + \epsilon} n)$ per insertion, which is $O(\frac{1}{\epsilon} \log n)$ for a small $\epsilon$. In particular, $C = \minCarboricityExample{}$ ($\epsilon = \frac{1}{\alpha}$) yields $O(\alpha \log n)$ recourse.
\end{restatable}

In fact, we can adapt the algorithm slightly to be $\Delta$-adaptive. That is, using $\Delta_t + C$ colors where $\Delta_t$ is the maximum degree of $G^t$ and $C$ is only a function of $\alpha$. The adaptation requires us to recolor also due to deletions. Theorem~\ref{theorem_adaptive_Delta_for_low_arboricity}, proven in Section~\ref{subsection_adaptive_palettes}, states an even slightly stronger claim.

\begin{restatable}[$\Delta$-adaptive for Low Arboricity]{theorem}{theoremAdaptiveDeltaForLowArboricity}
\label{theorem_adaptive_Delta_for_low_arboricity}
Let $\alpha$ be an upper bound on the arboricity of a dynamic graph $G$ and let $\epsilon > 0$. If $C \ge \minCarboricity{}$, we can deterministically maintain a dynamic edge coloring in $O(m)$ time as follows. At any time $t$, for each edge $e=(u,v)$: $color(e) \le (\max\{\deg_t(v),\deg_t(u)\}+C)$ (in particular, $color(e) \le \Delta_t + C$), with worst-case recourse $O(\log_{1 + \epsilon} n)$ per update, which is $O(\frac{1}{\epsilon} \log n)$ for a small $\epsilon$.
\end{restatable}

\looseness=-1
In the case of low arboricity, we have an almost-matching lower bound. Theorem~\ref{theorem_lower_bound_2018} does not apply to graphs with low arboricity because the graph used in its proof has a relatively large arboricity ($\Theta(\Delta-C)$, where $C = O(\alpha) \ll \Delta$), so we generalize the construction in Theorem~\ref{theorem_lower_bound_generalized}.\footnote{The proof in \cite{LowerBoundRecourse2018} assumes that $\Delta+C$ is even. Our generalization also explicitly completes the proof for the case of odd $\Delta+C$.} Taking $\alpha=1$ (forests) as a concrete example, the lower bound turns out to be $\Omega(\frac{\log n}{\log \Delta})$.

\begin{restatable}[Lower Bound with Arboricity]{theorem}{theoremLowerBoundGeneralized}
\label{theorem_lower_bound_generalized}
Let $\Delta \ge 3$, $0 \le C \le \Delta-2$ and integer $1 \le \alpha \le \frac{\Delta-C}{2}$. For any $n \ge \Delta + C + 2$ there exists a graph $G$ with $n$ vertices and arboricity in $[\frac{\alpha+1}{2},\alpha]$, that is $(\Delta+C)$ edge colored except for one edge $e$, and satisfies the following property. If $\alpha < \frac{\Delta}{2}$ then any extension of the coloring to $e$ requires a recourse of
$\Omega \Big ( \log_{(\Delta/\alpha - 1)} \big ( \frac{n \cdot (\beta - 1)}{\Delta} \big ) \Big )$ where $\beta = \frac{(\Delta - \alpha)(C+\alpha)}{\alpha^2}$ and if
$\alpha = \frac{\Delta}{2}$ (then $C=0$ since $\alpha \le \frac{\Delta-C}{2}$), the recourse is $\Omega(n/\Delta)$.
\end{restatable}

The case of $C = 0$ was previously addressed by \cite{sadeh2024CachingInMatchingsICALP}, who studied the
\emph{competitive-ratio} of the recourse (rather than absolute recourse)
for
$\Delta$ edge coloring of bipartite graphs. The constructions there are different than those we use to prove Theorem~\ref{theorem_lower_bound_generalized} in Section~\ref{subsection_generalized_lower_bound}, and can be adapted to show $\Omega(n/\Delta)$ recourse for $\alpha \approx \frac{\Delta}{2}$ and $C=0$, also for odd $\Delta$.

\begin{remark}
\label{remark_generalized_lower_bound}
We simplify the expression of Theorem~\ref{theorem_lower_bound_generalized} in two special interesting cases:
\begin{enumerate}
    \item For $C = \Delta-2$ (the maximal non-trivial case): Then $\alpha = 1$, $\beta = (\Delta-1)^2$, and $\frac{n(\beta-1)}{\Delta} = \Theta(\Delta \cdot n)$. Therefore the lower bound is $\Omega(\log_\Delta n)$. This shows that for a fixed $\Delta$, we cannot guarantee $o(\log n)$ worst-case recourse with fewer than $\Delta-1$ extra colors.

    \item For $C = (2+\epsilon) \alpha$, as in the setting of Theorem~\ref{theorem_low_arboricity}: Then $\beta-1 = \Theta(\frac{\Delta}{\alpha})$. Indeed, $\beta-1 < \beta = (3+\epsilon)(\frac{\Delta}{\alpha}-1) = O(\frac{\Delta}{\alpha})$, and also $\alpha \le \frac{\Delta}{2}$ therefore $\beta - 1 > (3+\epsilon) \cdot \frac{\Delta}{2\alpha} - 1 > \frac{\Delta}{\alpha}$. Therefore the recourse is $\Omega(\log_{\Delta/\alpha}{(n/\alpha)})$. For example, for $\alpha = \Theta(\Delta)$ we get $\Omega(\log (n/\Delta))$.
\end{enumerate}
\end{remark}

As we mentioned, our shift-tree technique as its name suggests is based on shifting colors from one edge to another. Color shifting is a widely used approach, by most recoloring techniques. There is no apriori reason why one should only consider \emph{shift-based} algorithms (recall Definition~\ref{definition_chain_shift}), and it is natural to wonder whether this conceptually simplifying restriction comes at the cost of performance, in our case, in terms of recourse. We show that when considering a worst-case scenario, there is indeed a separation.

\begin{restatable}[Shift Recourse Gap]{theorem}{theoremRecourseGapShiftVersusGeneralSIMPLE}
\label{theorem_gap_shift_recoloring_versus_general_separation_simple}
\looseness=-1
Let $\Delta \ge 3$ and $0 \le C \le \Delta-3$. For $n = \Omega(\Delta)$ there exists a graph $G$ with maximum degree $\Delta$ that is $(\Delta+C)$ edge colored except for one edge $e$, that satisfies the following property. We can extend the coloring of $G$ to $e$ with only $O(1)$ recourse, but any shift-based algorithm will have recourse of $\Omega(n/\Delta)$ if $C=0$ and \separationExpression if $C \ge 1$. The latter expression can be simplified as follows: If $C \le \frac{\Delta}{3}$ then it is $\Omega \Big (\frac{\Delta}{C} \log \frac{nC}{\Delta^2} \Big )$, if $C = \epsilon \Delta$ for a constant $\epsilon > 0$ then it is $\Omega(\log (n/\Delta))$ and if $C = \Delta - O(1)$ then it is $\Omega(\frac{\log n}{\log \Delta})$.
\end{restatable}

Observe that the expression in Theorem~\ref{theorem_gap_shift_recoloring_versus_general_separation_simple} for $C \ge 1$ is identical to the lower bound due to \cite{LowerBoundRecourse2018}, and indeed our construction uses that of \cite{LowerBoundRecourse2018}.
The proof is very simple and goes as follows: take a worst-case graph $H$ and its coloring as constructed for the recourse lower bound in Theorem~\ref{theorem_lower_bound_2018}, for maximum degree $\Delta-1$ and $\Delta + C - 1$ colors. The color we saved, say red, is used to color a matching between $H$ and additional $|H|$ vertices. To color the uncolored edge of $H$, shift-based algorithms are restricted to shift colors over edges of $H$, while general recoloring algorithms can simply color the uncolored edge by red and recolor its two red neighbours. To prove for $C=0$ we reserve $\Delta-2$ colors for $\Delta-2$ matchings (instead of only the red color), and force the shift-based problem to pay the recourse of flipping a $2$-edge-coloring of a path, as demonstrated in Figure~\ref{figure_recourse_gap_on_a_path}.

\begin{figure*}[t]
    \centering
    \begin{subfigure}[t]{.3\textwidth}
        \centering
        \includegraphics[width=\textwidth]{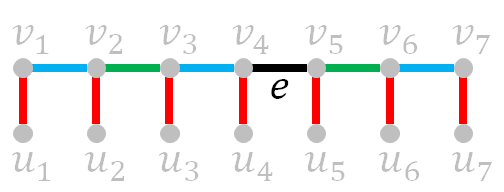}
        \subcaption{Initial Coloring}
        \label{figure_recourse_gap0}
    \end{subfigure}
    \hfill
    \begin{subfigure}[t]{.3\textwidth}
        \centering
        \includegraphics[width=\textwidth]{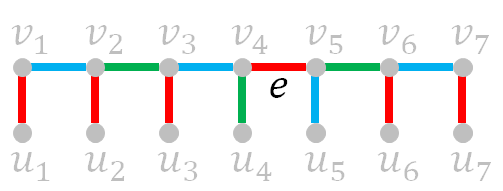}
        \subcaption{$O(1)$ general recourse}
        \label{figure_recourse_gap1}
    \end{subfigure}
    \hfill
    \begin{subfigure}[t]{.3\textwidth}
        \centering
        \includegraphics[width=\textwidth]{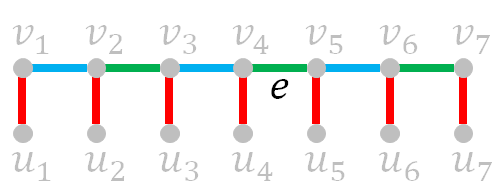}
        \subcaption{$\Omega(n)$ shift-based recourse}
        \label{figure_recourse_gap2}
    \end{subfigure}
    
    \caption{\small{A visual example for the proof of Theorem~\ref{theorem_gap_shift_recoloring_versus_general_separation_simple} for $\Delta=3$ and $C=0$. (a) The graph is a union of a red matching (vertical edges, $(v_i,u_i)$) and a subgraph $H$ constructed over the vertices $v_i$ by Theorem~\ref{theorem_lower_bound_generalized} for $\Delta'=2$, which in this case is a path (horizontal edges $(v_i,v_{i+1})$) colored as two bicolored paths and separated by an uncolored edge $e$ in the middle. (b) Extending the coloring to $e$ requires $O(1)$ recourse if we simply color it red. (c) If we insist on only shifting colors over a chain of edges, we must shift colors horizontally towards one of the ends, resulting in $\Omega(n)$ recourse.}
    }
	\label{figure_recourse_gap_on_a_path}
\end{figure*}

The proof outlined above is oblivious to the particular lower bound construction. It just modifies this bound to only apply to shift-based algorithms. Theorem~\ref{theorem_gap_shift_recoloring_versus_general_separation} generalizes Theorem~\ref{theorem_gap_shift_recoloring_versus_general_separation_simple} to formally state and prove the separation as a parameterized expression that depends on a given lower bound for general worst-case recourse.

Known recoloring algorithms that break the $\log n$ barrier are \cite{NibblingMethod2021,NibblingCycles2024} (the Nibbling Method) with $O(\mathrm{poly}(1/\epsilon))$ recourse \emph{in expectation}, using $(1+\epsilon)\Delta$ colors with high probability and provided that $\Delta$ is sufficiently large. They are not shift-based, so 
this intuitively supports the hypothesis that shift-based algorithms may be inherently weaker.

\section{Conclusions and Open Questions}
\label{section_summary_conclusion_open_questions}

We developed the \emph{shift-tree} technique and used it to derive tight recourse bounds. The main disadvantage of this technique is the large number of extra colors, $C$, that it requires. Thus, the most natural next question is: \emph{Can we augment the technique so it can be used with a smaller palette?} To clarify the core issue, observe that the shift-tree (Definition~\ref{definition_shift_tree_and_terminology}) expands potential recoloring paths exactly once from each vertex. This is fine when $C$ is large, because we utilize ``most'' of the edges of each vertex. However, when $C$ is small, say $C = O(1) \ll \Delta$ we utilize only a small number of paths through each vertex.

To give an alternative perspective, while $\Theta(\Delta^2)$ bicolored paths pass through or originate from a vertex of degree $\Delta$, our under-utilization results in only considering $O(C)$ bicolored paths through it, and $O(C^2)$ bicolored paths that originate from it. As a contrasting example, if we translate any of the \emph{multi-step Vizing chains} algorithms to a ``tree language'', then their tree is not expanded in a BFS manner, and it does expand a vertex multiple times when it is part of multiple possible bicolored paths. Therefore, a generalization of our technique should probably allow a vertex to be expanded multiple times, in some clever way so we can still control the size of the tree. We believe that such generalizations should be able to apply the technique with lower values of $C$.

\looseness = -1
Working beyond the worst-case is also of interest, such as achieving low amortized or expected recourse. The \emph{Nibbling Method} shows expected recourse that is independent of $n$, and offline coloring results also show that the average time per edge (and thus recourse) is independent of $n$. \emph{Can the shift-tree be used to design algorithms with good amortized recourse?} 

In parallel to further understanding the best recourse, one may also consider improving the running time while keeping the recourse small. The shift-tree by itself takes polynomial time because we might essentially BFS-scan the whole graph to determine a useful path, unless there is a better way to trim undesired sub-trees. A randomized alternative is to expand a random path, that will be useful with probability $p > 0$. We could amplify $p$ by repetitions, but to succeed with high-probability in poly-logarithmic time, $p$ must be sufficiently large, and the shift-tree only guarantees a single useful path in the whole tree. One approach to increase $p$ is to expand paths farther. Since we anyway expand a single path and not the whole shift-tree, we do not know or care if a vertex should be inactive due to another branch, as long as the path itself does not revisit vertices. Another helpful approach is to use more colors, as in the case of low arboricity where $C = (2+\epsilon)\alpha - 1 \to (4+\epsilon)\alpha$ reduces the update time from linear to poly-logarithmic (in this case the algorithms are even deterministic). Also, perhaps other clever ideas and data structures could help.

\appendix
\section{Related Work and Recourse-results Table}
\label{section_appendix_table_of_results}

This section has two purposes. The first is to summarize the results that were mentioned throughout Section~\ref{section_introduction}, in Table~\ref{table_recourse_results}. These are the main results for dynamic edge coloring. The table is quite heavy due to the multiple aspects to consider: The maximum degree, the number of colors, arboricity, running time and of course, our main focus, recourse.

\begin{table}[!p]
    \scriptsize
            \hspace*{-2cm}
            \begin{tabular}{|l|c|c|c|c|l|}
    
            \hline
            Reference & Colors & Recourse (per update) & \multicolumn{2}{|c|}{Runtime} & Notes \\
            \hline
            \hline

            \hline
            \cite{Vizing1964} & $\Delta + 1$ & $O(n)$ & $O(n)$ & D & Implicit running time. \\

            \hline
            \cite{GP2020TwoStepVizing} & $\Delta + 1$ & $O(\mathrm{poly}(\Delta) \cdot \sqrt{n})$ & $O(\mathrm{poly}(\Delta) \cdot n)$ & D & 
            \begin{tabular}{@{}l@{}}
                \cite{Bernshteyn2022SmallRecoloring} discusses how \\
                \cite{GP2020TwoStepVizing} implies this result.
            \end{tabular} \\
            
            \hline
            \cite{2025SODASayanOfflineVizingSqrtNDelta} & $\Delta + 1$ & $O(runtime)$ & $O((\Delta^2 + \sqrt{\Delta n}) \cdot \log^2 n)$ & R & Theorem 8.1 \\
    
            \hline
            \cite{Bernshteyn2022SmallRecoloring} & $\Delta + 1$ & $ O(\Delta^6 \cdot \log^2 n)$ & $O(\Delta^6 \cdot n \log ^2 n)$ & R & 
            \begin{tabular}{@{}l@{}}
                Runtime implied by an \\
                $O(\Delta^6 \cdot \log^2 n)$-rounds \\
                distributed algorithm.
            \end{tabular} \\
    
            \hline
            \cite{MultistepVizing2023} & $\Delta + 1$ & $O(\Delta^7 \cdot \log n)$ & $O(\mathrm{poly}(\Delta) \cdot n)$ & D & Implicit running time. \\ 

            \hline
            \cite{DynamicEdgeColoring2019DuanEtal} & $(1+\epsilon)\Delta$ &  $O(runtime)$ & 
            $O(\epsilon^{-2} \log^7 n)$ worst-case & R & 
            $\Delta = \Omega(\epsilon^{-2} \log^2 n)$ \\

            & $(1+\epsilon)\Delta_t$ & $O(runtime)$ &
            $O(\epsilon^{-4} \log^8 n)$ amortized
             & R & $\Delta_t = \Omega(\epsilon^{-2} \log^2 n)$
            \\
            
            \hline
            \cite{MultistepVizing2023} & $(1+\epsilon)\Delta_t$ &
            $O(\Delta_t ^6 \log^2 n)$ &
            $O(\epsilon^{-6} \log^9 n \log^6 \Delta_t )$ & R & . \\
            
            \hline
            \cite{NibblingMethod2021} & $(1+\epsilon)\Delta \ w.h.p.$ & $O(\epsilon^{-9})$ in expectation & $O(\mathrm{poly}(n))$ & R & 
            \begin{tabular}{@{}c@{}}
                Implicit runtime; \\
                $\Delta = \Omega(\epsilon^{-6} \log n)$ \\
            \end{tabular} \\
            
            
            \hline
            \cite{NibblingCycles2024} & $(1+61 \epsilon)\Delta  \ w.h.p.$ & $O(\epsilon^{-4})$ in expectation &
            $O(\epsilon^{-9} \cdot \log^4(1/\epsilon))$ & R &
            \begin{tabular}{@{}c@{}}
                $\Delta \ge (100 \log n / \epsilon)^{x}$ \\
                $x \equiv 30 \log (1/\epsilon) / \epsilon$
            \end{tabular} \\

            \hline
            \cite{2025LinearEdgeColoring} & $(1+\epsilon)\Delta$ & $O(\epsilon^{-4} \cdot \log n)$ & $O(m \epsilon^{-4} \log(1/\epsilon))$ & R & 
            \begin{tabular}{@{}l@{}}
                Can color a whole graph \\
                within this runtime \\
                (offline algorithm).
            \end{tabular} \\
            
            \hline
            \cite{bhattacharya2018dynamic} & $2\Delta - 1$ & 0 & $O(\log \Delta)$ & D & Folklore recourse. \\

            \hline
            \begin{tabular}{@{}l@{}}
                \textbf{New}, \\ Thm.~\ref{theorem_tight_worst_case_recourse} 
            \end{tabular} &
            \begin{tabular}{@{}c@{}}
                $\Delta+C$ \\
                $C \ge \frac{\Delta}{\phi} + 0.724$ \\ ($\approx 0.618 \Delta$)
            \end{tabular} &
            $O \left (\frac{\log n}{{\log \frac{\Delta+C}{\Delta-C}}} \cdot \frac{1}{\frac{C}{\Delta} - \frac{1}{\phi}} \right )$ & $O(m + \Delta^2)$ & D & 
            \begin{tabular}{@{}l@{}}
                If $\Delta - C = O(n^{1-\delta})$ \\
                and $(C/\Delta) - 1/\phi = \Omega(1)$: \\
                \textbf{Tight} worst-case recourse.
            \end{tabular} \\

            \hline
            \begin{tabular}{@{}l@{}}
                \textbf{New}, \\ Thm.~\ref{theorem_DeltaMinus2_colors} 
            \end{tabular} &
            $2\Delta-2$ & $O(\log_{\Delta} n)$ & $O(m)$ & D & 
            \begin{tabular}{@{}l@{}}
                $\Delta \ge 4$. \\
                \textbf{Tight} worst-case recourse. \\
            \end{tabular} \\

            \hline
            \hline
            \cite{ArboricitySWAT2024_sayan} & $\Delta + (4+\epsilon)\alpha$ & $O(\epsilon^{-1} \log n)$ & $O(\epsilon^{-2} \log^2 n \log \Delta)$ & A & Theorem 8 \\

            & $\Delta_t + (4+\epsilon)\alpha_t$ & $O(\epsilon^{-5} \log^4 n)$ amortized & $O(\epsilon^{-6} \log^6 n)$ & A & Theorem 1 \\
            
            \hline

            \cite{ArboricitySWAT2024_Christiansen} & $\Delta + (4+\epsilon)\alpha$ &  & $O(\epsilon^{-2} \log n \log \Delta)$ & A & Theorem 17 \\

            & $\Delta_e + (8+\epsilon)\alpha$ & $O(runtime)$ amortized & $\epsilon^{-3} \cdot \mathrm{polylog}(n,\Delta, \alpha)$ & A & Theorem 18 \\

            & $\Delta_e + O(\alpha_t)$ & & $\mathrm{polylog}(n, \Delta, \alpha)$ & A & Theorem 14 \\

            \hline
            \begin{tabular}{@{}l@{}}
                \textbf{New}, \\
                Thm.~\ref{theorem_adaptive_Delta_for_low_arboricity}
            \end{tabular} &
            $\Delta_e + \minCarboricity{}$ & $O(\epsilon^{-1} \log n)$ & $O(m)$ & D & $\frac{1}{\alpha} \le \epsilon \le 1$ \\
            
            \hline
            \hline
            \cite{LowerBoundRecourse2018} & 
            \begin{tabular}{@{}c@{}}
                $\Delta+C$ \\
                $C \in [1,\Delta-2]$
            \end{tabular} &
            $\Omega \left (\frac{\log \frac{nC}{(\Delta-C)^2}}{\log \frac{\Delta+C}{\Delta-C}} \right )$ &
            N/A & . & 
            \begin{tabular}{@{}l@{}}
                Specific graph and coloring. \\
                $C \le \frac{\Delta}{3} \Rightarrow \Omega( \frac{\Delta}{C} \cdot \log \frac{Cn}{\Delta^2} )$
            \end{tabular} \\


            \hline
            \begin{tabular}{@{}l@{}}
                \textbf{New}, \\
                Thm.~\ref{theorem_lower_bound_generalized}
            \end{tabular} & 
            \begin{tabular}{@{}c@{}}
                $\Delta + C$ \\
                $C \in [1,\Delta-2]$
            \end{tabular} &
            \begin{tabular}{@{}c@{}}
                $\Omega \Big ( \log_{(\frac{\Delta}{\alpha'} - 1)} \big ( \frac{n \cdot (\beta - 1)}{\Delta} \big ) \Big )$ \\ 
                $\beta \equiv (\frac{\Delta}{\alpha'}-1)(\frac{C}{\alpha'}+1)$
            \end{tabular} &
            N/A & . & 
            \begin{tabular}{@{}l@{}}
                Generalization of \cite{LowerBoundRecourse2018}, \\
                to arboricity $\alpha \in [\frac{\alpha'+1}{2},\alpha']$ \\
                where $\alpha' \in [1, \frac{\Delta-C}{2}] \cap \mathbb{N}$.
            \end{tabular} \\
            \hline
    \end{tabular}

    \caption{\scriptsize{Summary of recourse results. The table begins with upper bounds (most rows), and ends with few lower bounds. \emph{w.h.p.} means with-high probability (could exceed the number of colors). Recourse is worst-case unless specified otherwise. While we only focus on recourse, we also mention running times, and $D$, $A$ and $R$ stand for deterministic, amortized (deterministic) and randomized (oblivious adversary, success w.h.p.). Parameters in the table are: $n$ (vertices), $m$ (edges), $\Delta$ (max degree), $\alpha$ (arboricity), and $\epsilon > 0$ or $C$ for extra colors, which we always assume to be at most $\Delta-2$ otherwise recourse is never necessary. Subscript implies adaptivity, where $\Delta_t$ and $\alpha_t$ are the maximum degree and arboricity (respectively) at time $t$, and $\Delta_e \equiv \max \{ \deg_t(u),\deg_t(v) \}$ for edge $e=(u,v)$. Adaptivity means that the algorithm does not require prior knowledge of the parameter, but performs with respect to it.}}
    \label{table_recourse_results}
        
\end{table}

The second purpose of this section is to mention additional edge coloring models and results which are not directly relevant to the discussion in Section~\ref{section_introduction}. We preface this part with the fact that edge coloring is an interesting combinatorial problem, even before considering it as an algorithmic problem. See books~\cite{GraphEdgeColoring2012Book, BookGraphColoring2024Book} for references and results, both combinatorial and algorithmic.

Focusing on algorithmic areas, Section~\ref{section_introduction} did not cover all the models in which edge coloring has been studied, and we list a few more. Some recent works in the \emph{streaming model} are~\cite{StreamingChechikEtal2024ICALP,GhoshEtalWStreamingEdgeColoringICALP2024,StreamingEdgeColoringICALP2025}. The work~\cite{sadeh2024CachingInMatchingsICALP} is in a \emph{caching model} (similar to the \emph{dynamic model} but not quite), and studies the recourse in terms of \emph{competitive-analysis} rather than in absolute terms. Edge coloring has also been studied in the \emph{PRAM model} (parallel computing)~\cite{CHROBAK1989,CHROBAK1990,HE1990,ApproxEdgeColorPRAM2001}. Another adjacent model is when the coloring is implicit~\cite{DynamicColoringImplicitExplicit2020,christiansen2022fullydynamicArboricity}.

In addition to categorizing work by models, some works focus within a specific model on special sub-classes of graphs. Planar graphs were shown by Vizing~\cite{Vizing1965_10up,Vizing1965_8up} to be $\Delta$ edge colorable for $\Delta \ge 8$ with instances that show the necessity of $\Delta+1$ colors if $\Delta \le 5$. \cite{SANDERS2001} showed that $\Delta=7$ planar graphs are also $\Delta$ edge colorable ($\Delta=6$ remains open), and there are other non-algorithmic results that show $\Delta$ edge colorability of certain ``almost planar graphs'', for example \cite{ZHANG20111}. Therefore much work has been put to $\Delta$ edge coloring of planar graphs in the offline centralized and parallel models~\cite{CHROBAK1989,CHROBAK1990,HE1990,Cole2008}. Bipartite graphs are also an interesting family of graphs, and are also $\Delta$ colorable. Some recent works on bipartite graphs are \cite{sadeh2024CachingInMatchingsICALP, onlineMatchingSODA2025}. An even more special family of graphs are trees (both bipartite and planar), which is of particular interest in the \emph{distributed model} setting~\cite{BookDistributedAlgorithms2013,LowerBoundRecourse2018}.
\section{Analysis of Shift-Tree Algorithms (Formal Section)}
\label{appendix_section_shift_tree_algorithms_overview}

In this section we cover the formal analysis of shift-trees and the results that follow from this technique. We start with formal definitions and observations, and then proceed to the results.

\subsection{Definitions and Observations}

For completeness and exact notations, we formally cover some definitions that were already mentioned.

\begin{definition}[Reserved Notations]
$G = (V,E)$ is a simple (dynamic) graph, where $|V|=n$ is fixed in time and $|E|=m$ changes over time based on insertions and deletions. $\Delta$ is an upper bound on the maximum degree of $G$, and similarly $\alpha$ is an upper bound on the arboricity. In adaptive contexts, $\Delta_t$ and $\alpha_t$ are the actual maximum degree and arboricity at time $t$. $\deg_H(v)$ is the degree of node $v$ in a graph $H$. $C$ is the number of ``extra'' colors at our disposal, out of a total of $\Delta+C$ colors.
\end{definition}

\begin{definition}[Recourse]
Let $G$ be a fully dynamic graph. Given an update to $G$, the amount of recourse per update is the number of edges being \emph{recolored}. Coloring a new edge is not per se recourse. If we did consider it as recourse, it would simply add $+1$ per insertion.
\end{definition}

\begin{definition}[Neighbourhood]
Given a graph $G=(V,E)$, the neighbourhood of a vertex $v \in V$ is the set of all vertices such that there is an edge connecting $u$ and $v$ in $G$. Formally, $N(v) \equiv \{ u \in V : \exists (u,v) \in E \}$. The neighbourhood of an edge $e=(u,v) \in E$ is the set of all edges that share an endpoint with $e$. Formally, $N(e) \equiv \{ e' \in E \setminus \{e\} : e' \cap e \ne \emptyset\}$. (One can define neighbours and neighbourhood of an edge in terms of vertices of the line-graph of $G$.)
\end{definition}

\begin{definition}[Available, Used]
\label{definition_available_colors}
Given a graph $G=(V,E)$ and $v \in V$, let $A(v)$ be the set of available colors, that are not used by any edge incident to $v$. Note that $|A(v)| \ge \Delta + C - \deg_G(v) \ge C$. For an edge $e=(u,v)$, let $A(e) = A(v) \cap A(u)$ be the colors available for $e$. The complementary set of colors, denoted $\bar{A}(x)$ (whether $x$ is a vertex or an edge) is \emph{unavailable}, or \emph{used} near $x$.
\end{definition}

We proceed from basic definitions and notations to shift-trees. The basic definition was given in Definition~\ref{definition_shift_tree_and_terminology}. We make the following simple observations on shift-trees.

\begin{observation}[Unique Internal Appearance]
\label{observation_simple_path_until_leaf}
Since we do not expand inactive copies of a vertex, any path $P$ from the root of a shift-tree $T$ to a leaf contains only one internal (non-leaf) node corresponding to each vertex. The vertex that corresponds to the leaf at the end of $P$ may also have an additional internal active node on $P$. For example, in Figure~\ref{fig_colored_shift_tree_T_example}, $v$ is both the root and a leaf of $T$.
\end{observation}

\begin{observation}[Parent-Child]
\label{observation_parent_child_in_shift_tree}
Two neighbours $x,y \in G$ may be parents of each other in different locations in $T$. This can happen if we expand $x$ and $y$ in separate branches of $T$, and each happens to have the other as a child. For example, in Figure~\ref{fig_colored_shift_tree_T_example}, the pairs $x$, $d$ and $b$, $c$ satisfy this. However, $x$ cannot be a grandchild of itself in $T$. In such a case, let $y$ be the node in between. Then by definition when we expand the children of $y$, we look for edges of $y$ that are colored by colors currently available in $x$. In particular, since currently $(x,y) \in G$ is considered uncolored (since we shifted colors along the path), no available color in $x$ corresponds to $(x,y)$.
\end{observation}

\begin{observation}[Degree Implications]
\label{observation_degree_in_T_versus_G}
If $x$ corresponds to an internal node in $T$ then the degree of $x$ in $G$ is at least the number of its children in $T$ plus $1$ for its parent (the root also has a ``virtual'' parent). We also know that the degree of $x$ in $G$ is at least the number of nodes that are parents of $x$ in $T$ since every vertex can only correspond to one internal node in $T$. However, we cannot simply sum these two quantities due to Observation~\ref{observation_parent_child_in_shift_tree}.
\end{observation}

\begin{observation}[Unique Edges of Paths]
\label{observation_unique_edges_path_in_shift_tree}
Let $P$ be a shiftable path contained in $T$ with $root(T)=v$, that begins at $e_1 = (u,v)$ and ends at $e_k = (u',v')$ for some $k \ge 2$. Every edge in $P$ is unique, 
except possibly $e_1$ and $e_k$ that may be the same, with
$u'=u$ and $v'=v$.
Indeed, by Observation~\ref{observation_parent_child_in_shift_tree} no vertex appears internally twice on the path or is a grandchild of itself. Therefore an edge $e'=(x,y)$ cannot appear before $e''=(y,x)$ because either $e''$ is the successor of $e'$ and $x$ is a grandchild of itself, or $y$ is internal twice. So, the only possibility for a repeating edge is if $u'=u$ and $v'=v$, which is possible (see an example in Figure~\ref{fig_colored_shift_tree_T_example}).
\end{observation}

\subsection{Bounding the Depth of a Shift-tree}
\label{appendix_proof_parameterized_lemma}

In order to analyze the recourse, we require a bound on the depth of a shift-tree $T$. Lemma~\ref{lemma_shift_tree_max_depth} gives this bound in terms of a multiplicity parameter $b$. Recall that the intuitive interpretation of $b$ is how many shiftable paths in $T$ reach a specific vertex $x$ from different edges of $x$. The more we have, the more likely we are to find a useful shiftable path.

\begin{restatable}[Depth Bound]{lemma}{lemmaBoundShiftTreeDepth}
\label{lemma_shift_tree_max_depth}
Let $T$ be a shift-tree that we extend until either: (1) we find a useful shiftable path, or (2) some vertex has $b \le C$ inactive copies. Let $\beta \equiv \frac{C+1}{b}$ and $H = \floor{\log_\beta n}$. Then $depth(T) \le H + 1$.
\end{restatable}

Lemma~\ref{lemma_shift_tree_max_depth} is the reason for the bound $b \le C$ in Theorem~\ref{theorem_large_palette_general_parameterized} that appears later. $b = C+1$ is also achievable, but then the depth of the tree might become linear due to the rate of encountering new vertices. By a simple counting argument, one cannot guarantee having $b \ge C+2$ inactive copies of some vertex.\footnote{A tree with $N$ inner nodes each having exactly $C+1$ children has a total of $1+N(C+1)$ nodes. Then on average each vertex appears $C+1 + \frac{1}{N}$ times, one copy is active, so it could be that no vertex has more than $C+1$ inactive copies.}

\begin{remark}
\label{remark_nitpick_full_level}
\looseness=-1
The pedantic reader may wonder when \emph{exactly} we stop the expansion of $T$. It does not truly matter. Our discussions take the arguably more natural approach to consider $T$ in whole levels. Implementation-wise, one may wish to immediately stop even within a level.
\end{remark}

\begin{proof}[Proof of Lemma~\ref{lemma_shift_tree_max_depth}]
Consider the process of extending the shift-tree level by level. By Lemma~\ref{lemma_basic_shift_children} every inner node of $T$ must have at least $C+1$ children because by the first stopping condition we stop successfully just before expanding a node with fewer children. Next, we show that we must encounter $b$ inactive copies of some vertex at level $\le H + 1$. To simplify the proof, we first assume that every active node has \emph{exactly} $C+1$ children. We remove this assumption at the end of the proof.

\textbf{We define a single-player game} that is based on the way we extend $T$. It is played in \emph{steps}, not to be confused with depth/levels in $T$. We think of the graph $G$ as a player in this game who plays exactly according to the way we defined the shift-tree. 
The \emph{state} of the game at \emph{step} $k$ is  a triplet $(A_k,I_k,L_k)$ whose components correspond to the total active nodes in $T$ at depth $\le k$, the total inactive nodes in $T$ at depth $\le k$, and leaves of $T$ at level $k+1$ prior to being assigned their active/inactive status, respectively. The initial state of the game, played from the children of the root of $T$, is $(A_0,I_0,L_0) = (1,0,C+1)$ because the root of $T$ is the only node at depth $\le 0$, it is active, and we need to assign active/inactive status to its children.\footnote{A more faithful version would start the game from the grandchildren in $T$, with the state $(C+2,0,(C+1)^2)$ since we know that every child of the root corresponds to a different vertex in $G$. However, this extra care for details does not improve the bound by much, so we simplify.}

A \emph{step} of the game proceeds by choosing a quantity $0 \le a_k \le L_k$ and setting $a_k$ of the leaves to be active. The rest, $L_k-a_k$ are set as inactive. Then, each active leaf has $C+1$ children, so the new state of the game is: $(A_{k+1},I_{k+1},L_{k+1}) = (A_k + a_k , I_k + (L_k-a_k) , (C+1) \cdot a_k)$. We define the game to end consistently with the termination of the construction of the shift-tree. So, the game ends either when we see $b$ inactive copies of some vertex (think for now that we keep track of the identity of the inactive nodes counted in $I_k$ but this would not be needed in the next paragraph where we relax the game) or when $A_k \ge n$. In the latter case the leaves of the last state belong to the tree (and all must be inactive) whereas in the former they do not, therefore the number of steps in the game plus $1$ upper bounds the depth of $T$. The objective of the game is to play as many steps as possible. Our goal is to show that the game stops at step $\le H$.

Critically, even though $G$ always chooses $a_k \in \mathbb{N}$, and its states are triplets of integers, we define the game over the real numbers. That is, alternative strategies may populate these triplets with fractional values, in which case the correspondence to actual nodes in the shift-tree is just conceptual (but, the strategy $G$ still produces $T$ in this relaxation). Because we study the game over the reals, we relax the first stopping condition to end the game if $(b-1) \cdot A_k < I_k$. In the integer case, this indeed guarantees that we have $b$ inactive nodes of some vertex by the pigeonhole principle. Instead of referring to nodes in a tree and expanding each active node by its $C+1$ children, we adopt the continuous language of \emph{segments}. We group all nodes in a level into a single segment. This segment is split into two disjoint parts, one represents the active nodes (blue in Figure~\ref{figure_continuous_tree_segments}), and its length is $a_k$. The other represents the inactive nodes (gray in Figure~\ref{figure_continuous_tree_segments}) and its length is $L_k-a_k$. The active segment then expands by a factor of $C+1$ to generate the entire segment (union of blue and gray) of the next level. In this continuous setting we can consider an arbitrary segment $[x,y]$ as a ``node'' in the ``continuous tree''. It has a parent-segment $[\frac{x}{C+1},\frac{y}{C+1}]$, and if it is active then it has a child-segment $[(C+1)x,(C+1)y]$ stretched out of it.

\begin{figure*}[t]
	\centering
    \begin{subfigure}[t]{.35\textwidth}
        \centering
        \includegraphics[width=\textwidth]{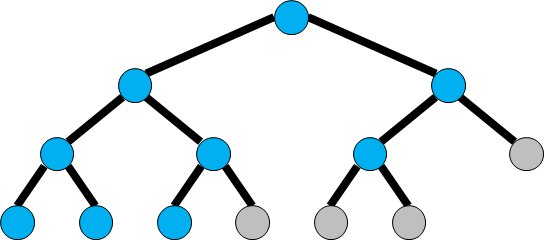}
        \subcaption{A tree $T$}
        \label{fig_game_tree}
    \end{subfigure}
    \hfill 
    \begin{subfigure}[t]{.3\textwidth}
        \centering
        \includegraphics[width=0.6\textwidth]{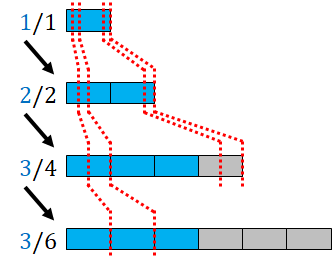}
        \subcaption{$T$ as integer segments}
        \label{fig_game_segments1}
    \end{subfigure}
    \hfill 
    \begin{subfigure}[t]{.3\textwidth}
        \centering
        \includegraphics[width=0.6\textwidth]{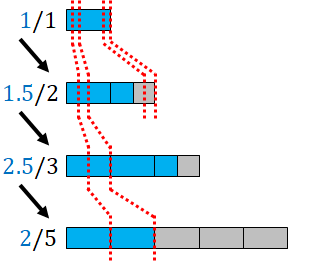}
        \subcaption{Continuous Segments}
        \label{fig_game_segments2}
    \end{subfigure}

	\caption{\small{An example of the expansion game in the proof of Lemma~\ref{lemma_shift_tree_max_depth}. (a) A standard tree $T$, with expansion of $2$ ($C=1$). Blue marks active nodes and gray marks inactive nodes. (b) A presentation of $T$ as integer segments, same color-coding. To the left of each segment we have two numbers, black for the length of the entire segment and blue for the length of its active subsegment. The length of the whole segment is twice the length of the active subsegment of the previous level. The red funnels illustrate the fractional ancestry hierarchy of these segments, where a segment $[x,y]$ of level $k$ is a child of the segment $[\frac{x}{2},\frac{y}{2}]$ of level $k-1$ and, if active, the parent of segment $[2x,2y]$ of level $k+1$. (c) Illustration of a game with fractional inactivations. The first two fractional inactivations are $\frac{1}{2}$. Observe that the total inactive amount is $4$ in both (b) and (c), yet the total blue area (active) in (c) is smaller. This is natural since an early inactivation reduces the size of future segments.}}
	\label{figure_continuous_tree_segments}
\end{figure*}

The continuous game allows us to simplify the analysis and avoid rounding issues. We try to understand how deep this tree can grow without having $b$ inactive copies of the same vertex. Intuitively, we show that this constraint, coupled with the fact that each non-leaf node has $C+1$ children, forces the player to introduce many new vertices into the game at each level and thereby quickly add all available vertices into the tree. We show that to maximize the depth we would like to have as many deactivations at each level as possible, so our growth is as small as possible. We can introduce $b-1$ inactive nodes per active node in the tree. But for example, when $b$ does not divide $C+1$ and we insist on an integral number of nodes we may not be able to do this exactly, and will have to keep track of the leftovers. The continuous game allows us to use our ``deactivation budget'' immediately, without having to track leftovers to be used in future steps. Figure~\ref{figure_continuous_tree_segments} exemplifies this where in the first expansion we have a budget to inactivate $\frac{1}{2}$. In the integer case we would have waited for the next level to accumulate another $\frac{1}{2}$ for a whole inactivation unit, while the continuous version does not have to wait.

We define the greedy strategy (called Greedy, whose values are annotated with superscript $g$) to play the game by choosing the minimal $a^{g}_k$ possible such that $(b-1) \cdot A^{g}_{k+1} \ge I^{g}_{k+1}$. We prove three properties of Greedy: (1) Greedy can never choose $a^{g}_k=0$, (2) Greedy is optimal, and (3) Greedy plays at most $H$ steps. Recall that the number of steps played by $G$ is tied to the depth of $T$. The optimality of Greedy provides us with an upper bound to the number of steps played by $G$, and proving these three claims will give us that $depth(T) \le steps(G) + 1 \le steps(Greedy) + 1 \le H+1$.

\textbf{Greedy can never choose $a^g_k = 0$.} We prove by induction that $\forall k \ge 1: a^g_{k-1} >0$ and $(b-1) \cdot A^{g}_{k} = I^{g}_{k}$. The base $(i=1)$ is the first step: Choosing $a^{g}_0=0$ results in $A^{g}_1=1,I^{g}_1=C+1$, but $\frac{I^{g}_1}{A^{g}_1} = C+1 > b$ is a violation of the definition of Greedy. By definition Greedy chooses $a^{g}_0 = \frac{C+2-b}{b}$, and gets $A^{g}_1 = 1+a_0=\frac{C+2}{b}$ and $I^{g}_1 = C+1-\frac{C+2-b}{b} = \frac{(C+2)(b-1)}{b}$ so $I^{g}_1 = (b-1) A^{g}_1$. Henceforth, for $i \ge 1$, by the induction hypothesis $a^g_{i-1} > 0$ and $I^g_{i} = (b-1) A^g_{i}$. Since $a^g_{i-1} > 0$ we get that $L^g_{i} > 0$. By definition Greedy maintains a ratio of (at most) $b-1$ and thus chooses $a^g_{i} = \frac{L^{g}_{i}}{b} > 0$, because then:
$\frac{I^g_{i} + (L^{g}_{i} - a^g_{i})}{A^g_{i} + a^g_{i}} = \frac{(b-1)A^g_{i} + (b-1) \cdot a^g_{i}}{A^g_{i} + a^g_{i}} = (b-1)$. Choosing $a^g_{i} < \frac{L^{g}_{i}}{b}$ would increase the ratio, violating the definition of Greedy.

\textbf{Greedy is optimal.} This is shown by a simple but technical substitution argument. Recall that an active segment $[x,y]$ has a child-segment $[(C+1)x,(C+1)y]$ and a parent-segment $[\frac{x}{C+1},\frac{y}{C+1}]$. Recall also the intuition that inactivating segments early results in smaller segments in the future, which serves to prolong the game, and this is why Greedy is optimal. Revisit Figure~\ref{figure_continuous_tree_segments} for an example. The formal technical argument follows.

Let $S$ be some non-greedy algorithm that plays the game $N$ steps. Use the superscript $s$ to annotate values related to it. Let $t^s \ge 0$ be the first step at which $S$ deviates from Greedy, so $a_t^s > a_t^g$. Let $j^s > t^s$ be the earliest future step where $S$ inactivates some quantity, that is $a^{s}_j < L^{s}_j$. If no such step exists, $j^s = \infty$. We show that we can modify $S$ to $S'$ (annotate values of $S'$ with prime) that also plays the game $N$ steps and either $t' \ge t^s + 1$ (same as greedy for one more step) or $j' \ge j^s + 1$ (another kind of progress). In particular, if $j^s = \infty$ then it must be that $t' \ge t^s + 1$.

We can repeat this argument to improve $S'$ to $S''$, etc. Because $N$ is finite, the increments of $t^s$ and $j^s$ eventually yield a strategy $S^*$ that is a prefix of Greedy. Then we may simply continue to play as Greedy, which would show that the number of steps played by Greedy is the maximum possible, being at least as large as that of any other strategy $S$. It may be strictly larger, because each improvement $S \to S'$ reduces the total accumulated active amount, and this reduction may let us play additional steps while $A^{g}_k < n$. We proceed to detail the modification $S \to S'$, see Figure~\ref{figure_modified_strategy}. In the following paragraphs, due to heavy use of $t^s$ and $j^s$ we simply write $t \equiv t^s$ and $j \equiv j^s$.

\begin{figure*}[t]
	\centering
    \begin{subfigure}[t]{.3\textwidth}
        \centering
        \includegraphics[width=0.9\textwidth]{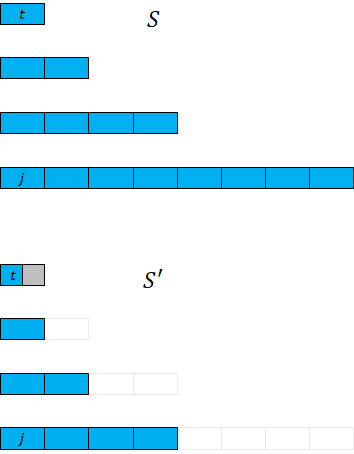}
        \subcaption{No future inactivation}
        \label{figure_strategy_modified_j_infinity}
    \end{subfigure}
    \hfill 
    \begin{subfigure}[t]{.3\textwidth}
        \centering
        \includegraphics[width=0.9\textwidth]{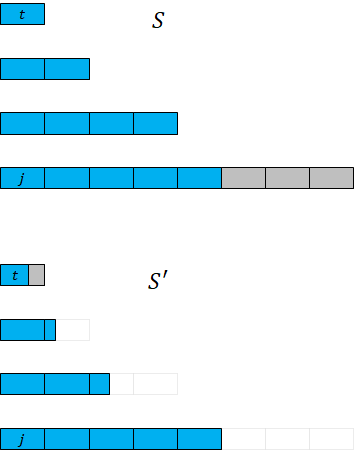}
        \subcaption{Small future inactivation}
        \label{figure_strategy_modified_small_quantity}
    \end{subfigure}
    \hfill 
    \begin{subfigure}[t]{.3\textwidth}
        \centering
        \includegraphics[width=0.9\textwidth]{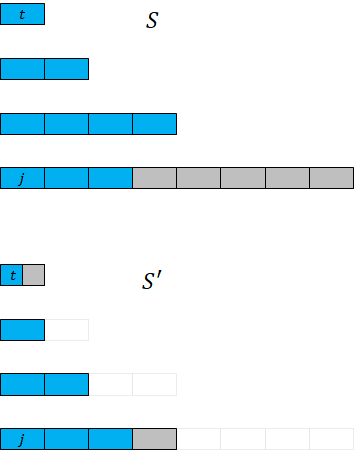}
        \subcaption{Large future inactivation}
        \label{figure_strategy_modified_large_quantity}
    \end{subfigure}

	\caption{\small{Three cases for the modification from strategy $S$ (top) to $S'$ (bottom) as detailed in the proof of Lemma~\ref{lemma_shift_tree_max_depth}. We focus on the full active (blue) segment at step $t$, and the next inactivation (grey) is at step $j$, if such exists. The growth factor of blue segments in this example is $2$ ($C=1$). The segments are divided for visual purposes to ease comparing sizes, these are not necessarily integer units. Phantom dashed segments in $S'$ mark quantities that are saved (do not exist) by early inactivation. (a) $S$ makes no inactivation after $t$ ($j = \infty$), in this case we just inactivate in step $t$ an extra segment such that $S'$ conforms with Greedy at step $t$. (b) Past step $t$, $S$ makes a little inactivation at step $j$, insufficient for $S'$ to conform to Greedy in step $t$. Then we inactivate an extra $\epsilon$ such that $\epsilon \cdot (C+1)^{j-t}$ exactly voids the inactivation of $S$ at step $j$. (c) Past step $t$, $S$ makes a large inactivation at step $j$ that lets us conform $S'$ to Greedy in step $t$. Some leftover inactivation in $S'$ may remain at step $j$. Observe that in both (b) and (c) $S'$ accumulates less activity in steps $t+1$ to $j-1$, but the current active segment in step $j$ is the same for $S$ and $S'$ by construction.}}
	\label{figure_modified_strategy}
\end{figure*}

If $j = \infty$, let $a'_t = a^{g}_t < a^{s}_t$ and $a'_i = L'_i$ for $t < i \le N$ (See Figure~\ref{figure_strategy_modified_j_infinity}). We remind the reader that $a_i$ affects state $i+1$, so the choice at time $t$ only makes a difference between $S'$ and $S$ from state $t+1$. In words: $S'$ agrees with Greedy one additional step, in state $t+1$, and then makes no additional inactivations, like $S$. Observe that $A'_{t+1} < \ldots < A'_{N}$
(expansion) while $I'_{t+1} = \ldots = I'_{N}$ (no more inactivations). Then $\forall i \in [t+1,N]: \frac{I'_{i}}{A'_{i}} \le \frac{I'_{t+1}}{A'_{t+1}} = \frac{I^{g}_{t+1}}{A^{g}_{t+1}} = b-1$. Thus the game does not stop for $S'$ due to the ratio rule, and since $\forall i \ge t: A'_{i} < A^{s}_i$, the game also does not stop for $S'$ due to accumulating $n$ active quantity. Thus $S'$ plays at least $N$ steps.

If $j < \infty$, let $\epsilon = \min(a^{s}_t - a^{g}_t , (L^{s}_j-a^{s}_j) \cdot \frac{1}{(C+1)^{j-t}})$ and set $a'_t = a^{s}_t - \epsilon$. In words, $S'$ inactivates $\epsilon$ more than $S$ in step $t$, where
the minimum in the definition of $\epsilon$ ensures two things: (1) Not to inactivate more than Greedy because then the game ends due to the ratio rule. (2) Not to inactivate an extra segment that is too long. The segment is too long if eventually its expansion at state $j$ would be larger than the whole inactive amount of $S$ at step $j$. During steps $i \in [t+1,j-1]$, $S$ has no inactivations so we choose the same for $S'$ and define $a'_i = L'_i$ (then $a'_i = (C+1)^{i-t} a'_{t}$). When $i \in [j,N]$ we resume to play like $S$ and choose $a'_i = a^{s}_i$. We can choose $a'_i = a^{s}_i$ because by definition of $\epsilon$, the second condition guarantees that $L'_j \ge a^s_j$ (we promote inactivation that \emph{does not} cause $L'_j < a^s_j$).

Now we show $S'$ keeps playing for $N$ steps. We remark that $A'_i \le A^{s}_i$ for $i \ge 0$ because the early inactivation at step $t$ reduces the accumulated active quantity, so if $S'$ stops before $N$ steps, it is only because of violating the ratio rule, and we show that it does not. Consider the following periods:
\begin{enumerate}
    \item $i \in [0,t]$: Recall that we made the change at step $t$, so the first affected state is $t+1$. Thus the game of $S'$ is identical to the game of $S$ during this period, and it does not stop.
    
    \item $i \in [t+1,j]$: $A'_{t+1} = A^{s}_{t+1} - \epsilon$ and $I'_{t+1} = I^{s}_{t+1} + \epsilon$. The choice of $\epsilon \le a^{s}_t - a^{g}_t$ guarantees that $a'_t \ge a^{g}_t$, which means that $A'_{t+1} \ge A^{g}_{t+1}$ and $I'_{t+1} \le I^{g}_{t+1}$. Then $\frac{I'_{t+1}}{A'_{t+1}} \le \frac{I^{g}_{t+1}}{A^{g}_{t+1}} = b-1$. For $i>t+1$ observe that $S'$ accumulates no additional inactive quantity so $I'_{i} = I'_{t+1}$, while $A'_{i} > A'_{t+1}$, therefore $\frac{I'_{i}}{A'_{i}} \le \frac{I'_{t+1}}{A'_{t+1}} \le b-1$. Thus the game of $S'$ does not stop during this period.
    
    \item $i \in [j+1,N]$: Inactivating $\epsilon$ more in step $t$, reduces the active amount in state $k$, $t+1 \le k \le j$ by $\epsilon \cdot (C+1)^{k-(t+1)}$, so the total (accumulated) reduction until step $j$ is the sum of a geometric series: $\epsilon \cdot \frac{(C+1)^{j-t} - 1}{C}$. The change from $S$ to $S'$ also increased the inactive amount in step $t$ by $\epsilon$ and reduced $\epsilon \cdot (C+1)^{j-t}$ inactive amount at step $j$. So concisely, denote $\delta \equiv \epsilon \cdot \frac{(C+1)^{j-t} - 1}{C}$, we have $A'_{j+1} = A^{s}_{j+1} - \delta$ and $I'_{j+1} = I^{s}_{j+1} - C \delta$. Observe that we subtract from $I^{s}_{j+1}$ a term that is $C$ times larger than what we subtract from $A^{s}_{j+1}$. Together with $b \le C$, and $\frac{I^{s}_{j+1}}{A^{s}_{j+1}} \le (b-1)$ (if $N \ge j+1$, the game of $S$ does not stop at step $j+1$), we get:
    $
    \frac{I'_{j+1}}{A'_{j+1}} =
    \frac{I^s_{j+1} - C\delta}{A^s_{j+1} - \delta} \le
    \frac{(b-1)A^s_{j+1} - (b-1)\delta}{A^s_{j+1} - \delta} =
    (b-1)
    $.
    For $i > j+1$ the difference between quantities of $S$ and $S'$ remains the same because as we explained, our choice of $\epsilon$ ensures that $L'_j \ge a^s_j$ so that we can choose $a'_i = a^{s}_i$ for $i \ge j$. That is, $A'_{i} = A^{s}_i - \delta$ and $I'_{i} = I^{s}_{i} - C\delta$ also for $i>j+1$ so by the same argument, $\frac{I'_i}{A'_i} \le b-1$ as long as $\frac{I^s_i}{A^s_i} \le b-1$. Thus the game of $S'$ does not stop during this period.
\end{enumerate}

\textbf{Greedy plays the game at most $H$ steps.} Here we omit the ``$g$'' superscript for brevity. Earlier we computed that greedy chooses $a_0 = \frac{C+2-b}{b}$ and that $a_{i+1} = \frac{L_{i+1}}{b} = \frac{C+1}{b} \cdot a_i = \beta \cdot a_i$. The last playable step $k$ satisfies $A_k \le n$. Then in total:
\begin{equation}
\label{eq_max_step}
 n \ge A_{k} = A_0 + a_0 \cdot \sum_{i=1}^{k}{\beta^{i-1}} = A_0 + a_0 \cdot \frac{\beta^k - 1}{\beta - 1}
\Rightarrow
k \le \log_{\beta}{\Big ( \frac{(\beta-1)(n-A_0)}{a_0} + 1 \Big )}    
\end{equation}
Since $A_0 = 1$ and $\frac{\beta - 1}{a_0} = \frac{C+1-b}{C+2-b} \in [0.5,1)$, and $k$ is integer, we get that $k \le \floor{\log_{\beta}{n}} = H$.
Finally, to remove the assumption 
that each node in $T$ has \emph{exactly} 
$C+1$ children, we revise the definition of the game.\footnote{A lazier solution would be to \emph{define} the shift-tree such that every node has at most $C+1$ children and pick the $C+1$ arbitrarily if there are more than $C+1$ options.} We construct the shift-tree $T$ (according to $G$) and  define the \emph{stretch} $\rho_k$ between states $k$ and $k+1$ to be the average number of children of a node at level $k$ of $T$. This means that there are $s$ active leaves at level $k$ of $T$ with a total of $\rho_k \cdot s$ children in level $k+1$. The game changes to use the $\rho_k$'s as follows: The player chooses $0\le a_k \le L_k$ and then state $k+1$ is $(A_{k+1},I_{k+1},L_{k+1}) = (A_k + a_k , I_k + (L_k-a_k) , \rho_{k+1} \cdot a_k)$.\footnote{Note that $a_k$ and $\rho_{k+1}$ are off-by-one with their respective index because the game begins already with the children of the root as unassigned leaves, so while this is assignment $a_0$, their expansion is $\rho_1$.} Since every active node has at least $C+1$ children, $\forall k: \rho_k \ge C+1$, and we show that we encounter $b$ inactive copies of some vertex at level $\le H+1$. We also define the initial state to be $(1,0,\rho_0)$ instead of $(1,0,C+1)$.

The analysis does not change much, and we still think of $G$ as a player that plays according to the definition of the shift-tree $T$. The analysis of Greedy becomes even more cumbersome, mainly its optimality, because we do not get a fixed growth rate of $C+1$. The proof of optimality needs to use $\epsilon = \min(a^{s}_t - a^{g}_t , (L^{s}_j-a^{s}_j) \cdot \Pi_{i=t}^{j-1}{\frac{1}{\rho_i}})$ to account for the different growth rate. Also, when we analyze the third period $i \in [j+1,N]$, we do not have a neat geometric series but rather: $A'_{j+1} = A^{s}_{j+1} - \epsilon \cdot (1 + \rho_{t+1} + \rho_{t+1} \cdot \rho_{t+2} + \ldots + \Pi_{i=t+1}^{j-1}{\rho_i})$ and $I'_{j+1} = I^{s}_{j+1} - \epsilon \cdot (\Pi_{i=t+1}^{j}{\rho_i} - 1)$.
Recall that $\rho_i \ge C+1$, we can show that $A^{s}_{j+1} - A'_{j+1} \le \frac{I^{s}_{j+1} - I'_{j+1}}{C}$ so that the rest of the optimality proof is the same. Indeed, denote $X \equiv \Pi_{i=t+1}^{j}{\rho_i}$, then
$A^{s}_{j+1} - A'_{j+1} = \epsilon \cdot X \cdot \sum_{k=1}^{j-t}{\Big (1 / \Pi_{i=t+k}^{j}{\rho_i} \Big )}
\le
\epsilon \cdot X \cdot \sum_{k=1}^{j-t}{\frac{1}{(C+1)^k}}
=
\epsilon \cdot X \cdot \frac{1}{C+1} \cdot \frac{1 - \frac{1}{(C+1)^{j-t}}}{1 - \frac{1}{C+1}}
=
\frac{1}{C} \cdot \epsilon (X - \frac{X}{(C+1)^{j-t}})
\le 
\frac{1}{C} \cdot \epsilon (X - 1) = \frac{I^s_{j+1} - I'_{j+1}}{C}$.

Finally, Equation~(\ref{eq_max_step}) still applies with the modification of: $A_{k} \ge A_0 + a_0 \cdot \sum_{i=1}^{k}{\beta^{i-1}}$ (equality to inequality).
\end{proof}

\subsection{Shift-tree Results}

As a first warm-up application of the shift-tree, we analyze the recourse on graphs with large girth, which in particular applies to dynamic trees. Theorem~\ref{theorem_large_girth_small_recourse} corresponds to Algorithm~\ref{algorithm_generic_shift_tree} (without implementing a leaves handler).

\begin{algorithm}[t]
    
    \DontPrintSemicolon
    \KwIn{
        Integer $C \ge 0$, a graph $G$ with maximum degree $\Delta$ and an edge $(u,v) \in G$ such that $G \setminus \{ e \}$ is properly $(\Delta+C)$ edge colored. Also, a parameter $2 \le b \le C$ when a leaves handler is used, otherwise $b = \bot$. (If $C \in \{ 0,1 \}$ then the use of a leaves handler is undefined; In our actual specific uses, $C \ge 2$.)
    }

    \KwOut{
        A proper $(\Delta+C)$ edge coloring of $G$.
    }
    
    \SetKwFunction{funcExtendColoring}{ExtendColoring}
    \SetKwFunction{funcLeavesHandler}{LeavesHandler}
    \SetKwProg{Fn}{Function}{:}{}
    
    \Fn{\funcExtendColoring{$G$, $(u,v)$, $C$, $b$}}{
        Compute a shift-tree $T$ with respect to $G$ and $(u,v)$, with stopping conditions:

        \If{a node $x \in T$ shares a free color with its parent in $T$}{
            Let path $P$ be defined as the path in $T$ from $(u,v)$ to $(parent_T(x),x)$.
        }
        
        \If{$b \ne \bot$ and a vertex $x \in G$ corresponds to at least $b$ leaves $x_i \in T$ for $i \in [b]$.}{
            Let path $P$ = \funcLeavesHandler{}.
        }
        Shift colors over the path $P$, then color the last edge of $P$ with an available color.
    }

    \caption{The generic shift-tree algorithm for extending coloring of a dynamic graph to a newly inserted edge $e$, assuming that a leaves handler is implemented, see Algorithm~\ref{algorithm_handle_leaves_generic} and Algorithm~\ref{algorithm_handle_leaves_DeltaMinus2} for implementations of two such handlers.
    }
    \label{algorithm_generic_shift_tree}
\end{algorithm}

\begin{restatable}[Large Girth]{theorem}{theoremUpperBoundWorstCaseLargeGirth}
\label{theorem_large_girth_small_recourse}
Let $C \ge 1$ and denote $\beta = \log_{(C+1)} \frac{n}{\Delta + C^2} + 2$. If $girth(G) \equiv g \ge 2 \beta + 1$, Algorithm~\ref{algorithm_generic_shift_tree} with $b = \bot$ maintains a $\Delta+C$ edge coloring of $G$ upon insertion of a new edge $e$ with recourse $R < \beta$. In particular, for $C=O(1)$, a girth of $g = \Omega(\log \frac{n}{\Delta})$ guarantees recourse of $R = O(\log \frac{n}{\Delta})$, and for $C = \Omega(\sqrt{\Delta})$, a girth of $g = \Omega(\frac{\log n}{\log \Delta})$ guarantees recourse of $R = O(\frac{\log n}{\log \Delta})$.
\end{restatable}

\begin{proof}
We start by proving a simpler bound of $R \le \beta'$ when the girth is $g \ge 2\beta' + 1$, for $\beta' \equiv \log_{(C+1)} n$. Let $T$ be a shift-tree with respect to $e$ and $G$, expanded where the only stopping condition is that we found a useful shiftable path (it will follow from the proof that we do not get stuck at a level such that no new vertices are revealed). Let $d$ denote the depth of $T$. We show that $d < \beta'$ and therefore $R < \beta'$.

To prove our claim, let $e=(u,v)$ and assume that $root(T) = v$. Note that every node in $T$ down to depth $d-1$ has (at least) $C+1$ children by Lemma~\ref{lemma_basic_shift_children}. There are exactly $n$ vertices in $G$, and because the girth is at least $2\beta'+1$, every node within the first $\hat{d} \equiv \min \{d,\beta' \}$ levels of $T$, plus $u$, corresponds to a distinct vertex of $G$, so we get: $n \ge 1 + \sum_{i=0}^{\hat{d}}{(C+1)^i} > (C+1)^{\hat{d}}$. Therefore $\hat{d} < \log_{(C+1)}{n} = \beta'$. Since $\hat{d} < \beta'$ it must be that $\hat{d} = d$, and we deduce that $d < \beta'$.

To prove the theorem with respect to the actual value $\beta = \log_{(C+1)} \frac{n}{\Delta+C^2} + 2$,\footnote{$\beta$ is marginally tighter than $\beta'$ when $\log (C+1) = o(\log \Delta)$ and $\log \frac{n}{\Delta} = o(\log n)$, e.g., $C=1$ and $\Delta = \frac{n}{\log n}$.} we take into account additional vertices that appear in $G$ but not in $T$. To start, assume that $T$ is regular, that is, every node down to depth $d-1$ has \emph{exactly} $C+1$ children. (We address the irregular case later.) Consider an edge $(p,x)$ in $T$ such that $p$ has grandchildren in $T$. Then it has \emph{exactly} $2(C+1)$ neighbour edges in $T$: $C+1$ edges to children of $x$, $C$ edges from $p$ to siblings of $x$ and $1$ edge from $p$ to its parent. 
Since $x$ has children we know that
$(p,x)$ has no available color (otherwise the depth of $T$ should have been smaller than $d$.)
There are $\Delta+C$ colors in total, but only $2(C+1)$ are used by the edges adjacent to $(p,x)$ in $T$, so there must be at least $(\Delta+C) - 2(C+1) = \Delta - C - 2$ additional edges adjacent to $(p,x)$ that belong to $G$ but do not appear in $T$. Furthermore, these edges must be incident to $p$ and not to $x$ because every color available in $p$ by definition reveals an edge of that same color in $x$ (to a child of $x$).\footnote{Note that these neighbours of $p$ were revealed by its child $x$, but different children of $p$ do not reveal additional neighbours, it would be double counting.} We can now revisit the counting analysis for $\hat{d} \equiv \min \{d, \beta \}$. Again, $\hat{d} \le \beta$ implies that every node within these levels corresponds to a unique vertex. In particular, the neighbours of $p$ in $G \setminus T$ are not neighbours of any other $p'$ in $T$ (this too will close a short cycle). Then we get:
$n \ge 1 + \sum_{i=0}^{\hat{d}}{(C+1)^i} + \sum_{i=0}^{\hat{d}-2}{(\Delta-C-2) \cdot (C+1)^i}
> (C+1)^{\hat{d}-2} \cdot (C^2 + \Delta)
$. Therefore, $\hat{d} < \log_{(C+1)}{\frac{n}{\Delta + C^2}} + 2 = \beta$. We conclude that $\hat{d} = d$ and therefore $d < \beta$.

To conclude the proof it remains to handle irregular trees. For this we describe a sequence of graphs $G_i$ with girth $g_i$ and size $n_i$, each with its corresponding shift-tree $T_i$, that satisfy the following properties: (1) All the trees have $depth(T_i) = d$; (2) Monotonicity: $n_{i+1} \le n_i$, $g_{i+1} \ge g_i$; (3) Trimming condition (towards regularity): pick a node $x \in T_i$ with more than $C+1$ children, if such exists. 
Then $T_{i+1}$ is obtained from $T_i$ by removing a child-subtree from $x$, and possibly additional child-subtrees from siblings of $x$; no new subtree appears,
and a deepest useful path is preserved. We begin the sequence with $G_0 \equiv G$. The sequence is defined until at some point we get $G^*$ with a corresponding regular shift-tree $T^*$ such that every inner node in $T^*$ has exactly $C+1$ children. This chain must be finite because by the trimming condition the number of subtrees that remain to be trimmed is strictly reduced. Since $d < \beta^*$ (because $T^*$ is regular), and by monotonicity of $n_i$ and $g_i$, we get that $d < \beta^* = \log_{(C+1)}{\frac{n^*}{\Delta+C^2}} + 2 \le \log_{(C+1)}{\frac{n}{\Delta+C^2}} + 2 = \beta$. In loose terms, this shows that the depth-maximizing shift-trees are regular.\footnote{The condition $g_{i+1} \ge g_i$ is stronger than what we need. It suffices to have $\forall i: g_i \ge 2\beta_i + 1$, but we get it by induction from the girth relation. Indeed $g_i$ grows as a function of $i$ while $\beta_i$ decreases, because $n_i$ decreases.}

We modify $G_i$ to $G_{i+1}$ and consequently its shift-tree $T_i$ to $T_{i+1}$ as follows. Let $z \in T_i$ be a leaf such that the path on $T_i$ to it is useful. We ensure not to trim it thus $depth(T_{i+1}) = depth(T_i)$. To clarify, naive trimming yields $depth(T_{i+1}) \le depth(T_i)$, while keeping a deepest leaf and preventing the scenario where all the useful paths get trimmed which triggers a shift-tree expansion to deeper levels gives $depth(T_{i+1}) \ge depth(T_i)$. Let $x \in T_i$ be a node with more than $C+1$ children. Let $p$ be the parent of $x$, and let $y$ be a child of $x$ such that $c \equiv color(x,y) \in A(p)$ in the coloring of $G_i$ (before shifting any colors). There are at least $C+1 \ge 2$ candidates for $y$, and we pick one that is as follows: If $x$ is an ancestor of $z$, make sure $y$ is not an ancestor of $z$. If a sibling $x'$ of $x$ is an ancestor of $z$, pick $y$ such that $c$ is \emph{not} the color from $x'$ towards the subtree that contains $z$. Define $G'$ as $G_i$ except that we remove all the edges incident to every vertex that appears in the subtree of $y$ in $T_i$ (including $y$) to make these vertices isolated. To get $G_{i+1}$, add to $G'$ the edge $(p,y)$ with color $c$, and if $y$ is not a leaf in $T_i$, also add the edge $(x,y')$ with color $c$ where $y'$ is some child of $y$ in $T_i$. Because we added an edge to $p$ colored by $c$, when the shift-tree expands siblings of $x$ we may have to also trim subtrees of those siblings that hang from an edge colored by $c$: this is why we chose $y$ carefully, to not trim $z$ by mistake. If such subtrees need to be trimmed, we update $G'$ to remove all the edges of all of these vertices (like we did to trim the subtree of $y$). The coloring of $G'$ is valid because we deleted edges from $G_i$, and the coloring of $G_{i+1}$ is valid because $c \in A(p)$ in $G_i$ and hence in $G'$, and $(x,y')$ replaces $(x,y)$ in using the color $c$ near $x$. We consider any isolated vertex of $G_{i+1}$ as if being removed, therefore $n_{i+1} \le n_i$. Also, $g_{i+1} \ge g_i$ because we removed edges to get $G'$, and then only added at most two edges to isolated vertices in $G'$. Finally, $deg_{G_{i+1}}(p) = deg_{G_i}(p) + 1 \le \Delta$ as necessary because $x$ has at least $C+2$ children in $T_i$, which implies at least $C+2$ available colors in $p$ in $G_i$, so $deg_{G_i}(p) \le \Delta-1$. By the choice of $y$ with respect to $z$, the useful path to $z$ remains, so $depth(T_{i+1}) = depth(T_i)$.  It remains to prove that no new subtree is spawned and that the indicated subtrees disappear.

Before proving that $T_{i+1}$ is a trimmed version of $T_i$ without the subtree of $y$, we clarify some intuition. For clarity, below, we add a subscript to the availability sets $A_t(x)$ to clarify the graph $G_t$ to which they refer. We cannot simply remove the edge $(x,y)$ from $G_i$ to remove it from $T_i$, because then $c \in A_{i+1}(p) \cap A_{i+1}(x)$ and so without it the expansion of $T_{i+1}$ stops prematurely. Therefore, to prevent such premature termination, we add a neighbour of $p$ with this color. We add the new neighbour of $p$ isolated to ensure that the girth is not affected. If there is more than one isolated vertex, then we could keep the edge $(x,y)$ and add the edge $(p,y')$, but in the special case where $y$ is a leaf in $T_i$, we have to prioritize adding $(p,y)$ to $G_{i+1}$ over keeping $(x,y)$. We still need the edge $(x,y')$, if we can have it, to prevent the creation of new subtrees of $x$ in $T_{i+1}$ compared to $T_i$ as a result of having $c \in A_{i+1}(x)$. Let us now resume the formal proof.

$T_{i+1}$ is clearly identical to $T_i$ until we reach $depth(x)-1$ because only then the difference between $G_i$ and $G_{i+1}$ affects anything. There are three possible cases:
\begin{enumerate}
    \item $x$ is the root of $T_i$ and $p$ is its virtual parent (in other words, $(p,x)$ is the newly inserted edge). In this case, $p$ is not part of $T_{i+1}$ and therefore $y$ cannot be its child in $T_{i+1}$. So $x$ lost a child ($y$), and the rest of the expansion of $T_{i+1}$ is identical to $T_i$. Note that $c \notin A_{i+1}(p)$ therefore $y'$ is not a child of $x$ in $T_{i+1}$. Previously $T_i$ did not reach any vertex that was in the subtree of $y$ except via $y$, and in $G_{i+1}$ they won't be reached because they are isolated. $y$ and $y'$ can be reached via $p$ and $x$ respectively, but due to the large girth we do not revisit $p$ and $x$ within depth $d$.\footnote{Unlike $x$ which was already expanded and won't be expanded again by definition, if we did reach $p$ somehow, we would have expanded it because it would be the first expansion of $p$, possibly visiting $y$ next.} Therefore $T_{i+1}$ is the trimmed $T_i$ as argued.

    \item $y$ is a leaf in $T_i$. Then $y'$ does not exist, so in $G_{i+1}$ the color $c$ is available at $x$. However, it is not available for $(p,x)$ due to the edge $(p,y)$. The new color available in $x$ has no effect because $y$ being a leaf means that $x$ has no grandchildren anyway. Therefore $T_{i+1}$ is the trimmed $T_i$ as argued.
        
    \item The remaining case is for $x \ne root(T_i)$ and $y$ being a non-leaf in $T_i$. Then $y'$ exists, and $p$ has a parent, denote it $p'$. Observe that $y$ is not a child of $p$ in $T_{i+1}$, otherwise (by definition) $c \in A_i(p')$ and $c \in A_{i+1}(p')$ (in both $G_i$ and $G_{i+1}$), but then $c \in A_i(p) \cap A_i(p')$ in $G_i$ in contradiction to expanding $T_i$ down to $y$ (and beyond). Similarly, $y'$ is not a child of $x$ in $T_{i+1}$ because $c \notin A_{i+1}(p)$ due to the edge $(p,y) \in G_{i+1}$. Finally, like in the first case, any vertex that we made isolated cannot appear in $T_{i+1}$, and $y$ and $y'$ do not appear in $T_{i+1}$ elsewhere because they are only reachable from $p$ and $x$ respectively, and we won't revisit those in $T_{i+1}$ due to the large girth of $G_{i+1}$. Therefore $T_{i+1}$ is the trimmed $T_i$ as argued. \qedhere
\end{enumerate}
\end{proof}


\begin{lemma}[Running Time]
\label{lemma_running_time_without_handler}
Algorithm~\ref{algorithm_generic_shift_tree} takes $O(m)$ time excluding the leaves handler call.
\end{lemma}

\begin{proof}
A shift-tree $T$ expands each vertex of $G$ at most once, therefore it has at most $2m$ edges (an edge $(x,y)$ may appear twice, $x$ parent of $y$ and $y$ parent of $x$). Assume that each vertex $v \in G$ maintains its set of available colors $A(v)$ as a list, and also has an array $B_v$ of length $\Delta+C = \Theta(\Delta)$ where $B_v[c]$ either points to a neighbour $u$ such that $color(v,u) = c$ or to the free color $c$ in $A(v)$. In order to determine the children of a node $x$ in $T$, we scan all $c \in A(p) \cup \{ c'\}$ where $p=parent_T(x)$, and $c' = color(p,parent_T(p))$ is the color that gets freed in $p$ when we shift colors on the path in $T$ down to $p$. 
For each such $c$ we check $B_x[c]$ to find the children of $x$ in $O(1)$ time per edge of $T$. If $B_x[c]$ indicates that $c$ is free then we terminate with the first \emph{if} of Algorithm~\ref{algorithm_generic_shift_tree}. Otherwise, make the neighbour $B_x[c]$ of $x$ a child of $x$.
To handle the second \emph{if} we maintain a counter $count(x)$ for each vertex $x$, which we increment each time a node that corresponds to $x$ is added as a leaf of $T$. At most one such leaf can be expanded to become non-leaf, so $count(x)$ is decremented at most once, and in total the maintenance of the counters takes $O(m)$ time.

The final line of Algorithm~\ref{algorithm_generic_shift_tree} shifts colors on a path $P$, $|P| = O(m)$, that we either find directly (first \emph{if}) or get from the leaves handler (second \emph{if}). When we shift a color from an edge $(u,v)$ to its neighbour $(v,w)$, it takes $O(1)$ time to update the lists ($A(i)$) and arrays ($B_i$) of $u$, $v$ and $w$ and no other vertex is affected. To color the uncolored edge $(u,v)$ post-shifting, we recall that by construction there is now a free color for it. We can find this free color in $O(\Delta)$ time by scanning $B_v$ and $B_u$ by index (color). If $m = o(\Delta)$,\footnote{Recall: $\Delta$ is a fixed bound on the maximum degree, while $m$ is the current number of edges in the graph.} then actually within the first $m+1$ entries of $B_v$ and $B_u$ there must be an available color. So overall, determining the free color of the final edge takes $O(m)$ time, and the total running time of the algorithm is indeed $O(m)$ time.
\end{proof}

Theorem~\ref{theorem_large_palette_general_parameterized} is our main tool for large palettes, first stated as a parameterized claim with parameter $b$ that represents how many leaf copies we aim to see of some vertex.

\begin{algorithm}[t]
    
    \SetKwFunction{funcLeavesHandler}{LeavesHandler}
    \SetKwProg{Fn}{Function}{:}{}
    
    \Fn{\funcLeavesHandler{caller's context: computed shift-tree $T$ with $b$ leaves $x_1,\ldots,x_b$ that correspond to the same vertex $x \in G$}}{
        Reduce $T$ to the connected subtree that contains $r \equiv root(T)$ and only the leaves $x_i$ ($i=1,\ldots,b$). Expand each $x_i$ by one additional shift-tree step. If any $x_i$ has less than $C+1$ children: return the path from $e$ to $(parent_T(x_i),x_i)$, otherwise denote the new tree by $T^{x+}$.
        
        Set $y$ as the neighbour of $x$ in $G$ that maximizes in $T^{x+}$ the number of leaves that correspond to $y$, minus the number of grandchildren that the active node of $y$ has. If the virtual parent of $r$ is a copy of $y$ (newly inserted edge was $(y,r)$), add to the grandchildren-count of $y$ the children of $r$.

        Check each leaf $y_i$ of $T^{x+}$ that corresponds to $y$: If after shifting colors over the path in $T^{x+}$ from $e$ to $e' \equiv (parent_{T^{x+}}(y_i),y_i)$, $e'$ has a free color: return this path.
    }

    \caption{A leaves-handler used with Algorithm~\ref{algorithm_generic_shift_tree}. Corresponds to  Theorem~\ref{theorem_large_palette_general_parameterized}.}
    \label{algorithm_handle_leaves_generic}
\end{algorithm}

\begin{restatable}[Parameterized Feasibility]{theorem}{theoremUpperBoundWorstCaseParameterized}
\label{theorem_large_palette_general_parameterized}
Let $G$ be a dynamic graph with maximal degree $\Delta$, that is edge colored with $(\Delta+C)$ colors for $C \in [1,\Delta-2]$. Let $2 \le b \le C$ be a parameter such that the following inequality is satisfied: $\frac{b(C-1)+1}{\Delta} > \Delta-C + 1$. Then upon insertion of a new edge $e$ to $G$, Algorithm~\ref{algorithm_generic_shift_tree} with Algorithm~\ref{algorithm_handle_leaves_generic} as the leaf handler extends the coloring to $e$ while recoloring $O(\log_\beta n)$ edges where $\beta \equiv \frac{C+1}{b}$, in $O(m + \Delta^2)$ time.
\end{restatable}

To prove Theorem~\ref{theorem_large_palette_general_parameterized}, we need the following simple combinatorial claim about trees.

\begin{lemma}
\label{lemma_descendants_sum}
Let $T$ be a rooted tree with $L \ge 1$ leaves. Denote the number of descendants of $v \in T$ at distance $d$  by $N_d(v)$. Then for every subset $S \subset T$: $\sum_{v \in S}{N_d(v)} \le d \cdot (L-1) + |S|$.
\end{lemma}

\begin{proof}
We begin with the case of children, $d=1$. Let $T' \subset T$ be the subset of all non-leaf nodes. For a non-empty tree, $L = 1 + \sum_{v \in T'}{(N_1(v)-1)}$, by induction on the size of the tree. It is trivial for a singleton tree ($T' = \emptyset$). The claim continues to hold when we expand the tree top-down and add $N_1(u)$ children to some leaf $u$, since $u$ stops being a leaf and we gain $N_1(u)$ new leaves. Since leaves have no children:
$\sum_{v \in S}{(N_1(v) - 1)} \le \sum_{v \in T'}{(N_1(v) - 1)} = L-1 \Rightarrow \sum_{v \in S}{N_1(v)} \le L-1 + |S|$. 

For $d > 1$, we define the trees $T_0,\ldots,T_{d-1}$ as follows. First, we define the tree $T_0$ recursively as $root(T_0)=root(T)$, and the children of a node in $T_0$ are its descendants at distance $d$ in $T$ (if none exist, the node is a leaf of $T_0$). Note that $T_0$ essentially skips levels of $T$ that are not a multiple of $d$ (root has level $0$). Similarly, the forest $F_i$ for $i=1,\ldots,d-1$ corresponds to the levels that are $i$ modulo $d$, and we make each $F_i$ into a tree $T_i$ by hanging all the trees of $F_i$ under a common root that does not correspond to any node in $T$. See Figure~\ref{figure_descendants_sum_example} for an example with $d=2$ (grandchildren).

$T$ has at most $L$ leaves therefore each $T_i$ has at most $L$ leaves. Then, using superscript on $N(v)$ to clarify the tree of reference: $\sum_{v \in S}{(N^T_d(v) - 1)} = \sum_{i=0}^{d-1}{\sum_{v \in S \cap T_i}{(N^{T_i}_1(v) - 1)}} \le d \cdot (L-1)$. Therefore indeed $\sum_{v \in S}{N^T_d(v)} \le d \cdot (L-1) + |S|$. 

Finally, this bound is tight. Consider a node $v_0$ with $L$ children $u^1_0,\ldots,u^L_0$. Now extend this tree upwards with nodes $v_1,\ldots,v_{d-1}$ such that $v_i$ is a single child of $v_{i+1}$ and downwards with nodes $u^i_j$ for $i \in \{1,\ldots,L\}, j \in \{0,\ldots,d-1\}$ such that $u^i_{j+1}$ is a single child of $u^i_j$. Let $S = \{ v_0,v_1,\ldots,v_{d-1} \}$, then $|S| = d$ and each of the nodes in $S$ has $L$ descendants at distance $d$, see Figure~\ref{figure_descendants_sum_example_tightness} for an example.
\end{proof}

\begin{figure*}[t]
	\centering
    \begin{subfigure}[t]{.22\textwidth}
        \centering
        \includegraphics[width=\textwidth]{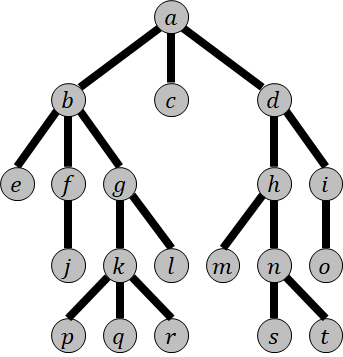}
        \subcaption{A tree $T$}
        \label{figure_descendants_sum_example_T}
    \end{subfigure}
    \hfill 
    \begin{subfigure}[t]{.22\textwidth}
        \centering
        \includegraphics[width=\textwidth]{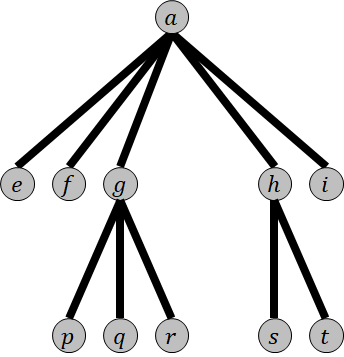}
        \subcaption{$T_0$ of even levels}
        \label{figure_descendants_sum_example_T0}
    \end{subfigure}
    \hfill 
    \begin{subfigure}[t]{.22\textwidth}
        \centering
        \includegraphics[width=\textwidth]{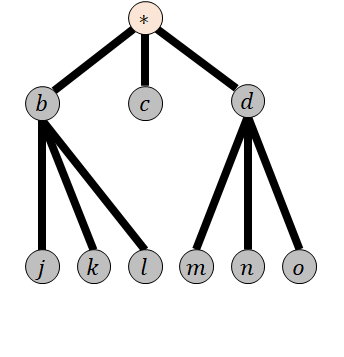}
        \subcaption{$T_1$ of odd levels}
        \label{figure_descendants_sum_example_T1}
    \end{subfigure}
    \hfill 
    \begin{subfigure}[t]{.25\textwidth}
        \centering
        \includegraphics[width=0.8\textwidth]{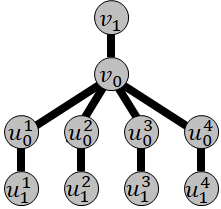}
        \subcaption{Tightness example}
        \label{figure_descendants_sum_example_tightness}
    \end{subfigure}

    \caption{\small{Illustrating Lemma~\ref{lemma_descendants_sum} and its proof, for $d=2$. (a) A given tree $T$. (b) A tree $T_0$ over the even layers, linking nodes to their grandchildren. (c) A tree $T_1$ over the odd layers, its root ensures that $T_1$ is a tree rather than a forest, but this root does not correspond to any node of $T$. (d) Tightness example: In this given tree with $L=4$ leaves, for $S = \{v_0,v_1\}$ we get $\sum_{x \in S}{N_2(x)} = 2(L-1) + |S| = 2L = 8$.}}
	\label{figure_descendants_sum_example}
\end{figure*}

\begin{proof}[Proof of Theorem~\ref{theorem_large_palette_general_parameterized}]
The high-level idea is as follows. We expand a shift-tree $T$ until we either find an edge with an available color or see many nodes that correspond to the same vertex, say $x$. Seeing many such nodes means that we found many different shift paths to reach $x$. However, we are truly interested in edges rather than vertices, so we will determine a suitable neighbour of $x$, say $y$, such that we reach $(x,y)$ over multiple shift paths. Then, we show that there are so many paths that one of them must be useful. We proceed to formal arguments.

Let $e=(u,v)$ be a newly inserted edge, and let $T$ be a shift-tree (Definition~\ref{definition_shift_tree_and_terminology}) with respect to $G \cup \{e\}$. We expand $T$ until we find a useful path, or see $b$ leaf copies of some vertex $x$. We will prove that if we did not find a useful path, then we can shift colors from $e$ to $e'  = (parent_T(x),x)$ for some copy of $x$, then shift one more color from an edge $(x,y)$ where $y$ is a neighbour of $x$, to $e'$, and can then finish and finally color the (now) uncolored edge $(x,y)$. By Lemma~\ref{lemma_shift_tree_max_depth}, $depth(T) \le O(\log_\beta n)$ where $\beta \equiv \frac{C+1}{b}$, and the number of \emph{re}-colorings would be at most $1 + depth(T)$. If we count colorings including the newly inserted edge, then add $1$ extra to the bound. In total, this expression is $O(\log_\beta n)$.

We focus on the skeleton $T^x$ that is formed by these nodes of $x$ (Definition~\ref{definition_shift_tree_and_terminology}), which has $b$ leaves.\footnote{Technically, $T^x$ could have $b+1$ leaves, if the shift tree has a non-leaf copy of $x$ such that none of the $b$ leaf copies of $x$ are among its descendants. This non-leaf copy of $x$ can even be useful to take the role of the $b$-th leaf in $T^x$ so we can stop the expansion of $T$ faster. That being said, for simplicity of presentation and because in the worst-case we won't get this extra leaf, we ignore this possibility.} Given $T^x$ we expand each of its leaves (a node corresponding to $x$) one step further, as in Figure~\ref{fig_colored_shift_tree_T_expand_Tx}, and denote this expanded skeleton by $T^{x+}$. The expansion is by the same logic of the shift-tree expansion: Consider each leaf $x_i$ of $T^x$, for $i=1,\ldots,b$. Let $p_i$ be the parent of $x_i$ in $T^x$. We define the children of $x_i$ as the neighbours of $x$ via colors that are free at $p_i$ after shifting the colors of the edges on the path in $T$ from $(u,v)$ to $(p_i,x_i)$ such that the only uncolored edge is $(p_i,x_i)$. Recall that by Lemma~\ref{lemma_basic_shift_children} stopping due to $b$ copies implies that every inner node in $T$ has at least $C+1$ children. So henceforth, assume that each $x_i$ has (at least) $C+1$ children. There could be more children if $deg(p_i) < \Delta$, but in the worst-case $deg(p_i) = \Delta$ and then $x_i$ has at most $C+1$ children.

The set of leaves of $T^{x+}$ is a multiset of vertices of $G$ that are neighbours of $x$, with repetitions. The size of this set is at least $b \cdot (C+1)$, so by the pigeonhole principle there is a neighbour of $x$ that corresponds to at least $\frac{b \cdot (C+1)}{deg(x)} \ge \frac{b \cdot (C+1)}{\Delta}$ leaves in $T^{x+}$. Denote this neighbour by $y$, and let $Y \subseteq \{1,\ldots,b\}$ be the set of indices $i$ such that $y$ is a child of $x_i$ in $T^{x+}$.\footnote{We abuse notation and refer to each of these nodes in $T$ as $y$ when there is no confusion, rather than giving each a distinct name.} Observe that for each $i\in Y$, $T^{x+}$ gives us a path in $G$ that begins with $e$ and ends with $(p_i,x)$ followed by $(x,y)$. We can shift colors over each such path such that after shifting, $(x,y)$ is the only uncolored edge. Denote the set of these paths by $\mathcal{P}$. If any path of $\mathcal{P}$ is such that $(p_i,x)$ is a good neighbour of $(x,y)$ (recall Definition~\ref{definition_dangerous}: this means that when we shift the color from $(x,y)$ to $(p_i,x)$, $(x,y)$ has an available color) then we can use this path to finish the coloring.

By the definition of the shift-tree, every inner node of $T^x$ corresponds to a different vertex, therefore the penultimate edge $(p_i,x)$  is different (as an edge in $G$) in each path in $\mathcal{P}$. Furthermore, by Observation~\ref{observation_unique_edges_path_in_shift_tree} $(p_i,x)$ appears once in its path and thus maintains its original color throughout the whole shifting until we shift the color of $(x,y)$ to it, unless $p_i=u$ and $x=v$ in which case $(p_i,x) = e$ has no original color.  Overall, we have at least $|Y| - 1 \ge \frac{b \cdot (C+1)}{\Delta} - 1$ paths in $G$ such that we can shift colors on each of them to eventually end up with the edge $(x,y)$ uncolored, where the penultimate edge $(p_i,x)$ in each path is colored by its initial color before shifting (and in particular $(p_i,x) \ne (u,v)$).

To conclude the proof naively, assume that none of the $|Y|-1$ paths in $\mathcal{P}$ passes through $y$ before reaching the last edge $(x,y)$. This assumption is equivalent to $y$ not being an inner node of $T^x$, or $u$ (being the parent of $root(T^x)$). (This assumption is demanding, and later we overcome it.) Provided that this assumption is true, then we know that the colors used around $y$ do not change by any shifting path that we choose. All that happens is that eventually $(x,y)$ becomes uncolored (shifting its color to $(p_i,x)$). Recall that there can be at most $\Delta-1-C$ colors associated with bad neighbours (Definition~\ref{definition_dangerous}) of the edge $(x,y)$ at any given time. Since the colors around $y$ do not change, the set of bad colors is pre-determined by the initial coloring. If $|Y| - 1 \ge \Delta-C$, we are guaranteed that some path in $\mathcal{P}$ reaches $(x,y)$ from a good neighbour of $x$ so we can shift colors on this path and finish the coloring of $(x,y)$ safely. Note that it could technically be possible that some path is fine for shifting even if $|Y| - 1 < \Delta-C$, yet this is what we require for a guarantee.

Alas, paths might go through $y$. In this case, the set of bad colors may change so that in some path $P \in \mathcal{P}$ through $p_i$ the edge $(p_i,x)$ is a good neighbour of $(x,y)$ based on the initial coloring, but turns out to be a bad neighbour when we shift colors along $P$. See Figure~\ref{figure_shift_color_makes_bad} for an example. If every path in $\mathcal{P}$ could ``invalidate itself'' by making $(p_i,x)$ bad after we recolor along it, we would be in trouble. Observe that the number of colors that may be shifted to edges incident to $y$ as bad colors depends on the number of grandchildren that $y$ has in $T^x$. Indeed, revisit Figure~\ref{figure_shift_color_makes_bad}: Shifting colors on a path where the child and grandchild of $y$ are $z$ and $w$ respectively results in $(y,z)$ receiving the color of $(z,w)$ (formerly free at $y$). 

So, the solution is simple: Let $N_\ell(y)$ be the number of paths in $\mathcal{P}$ that lead to a leaf $y$, and let $N_g(y)$ be the number of grandchildren of $y$ as an inner node in $T^x \cup \{ parent(root(T^x)) \}$ ($0$ if $y \ne u$ and is not an inner node; If $y=u$ we may count grandchildren for two instances of $u$, one in $T^x$ and another as the virtual parent of $root(T^x)$). Then define the \emph{effectiveness} of $y$ to be $N_\ell(y) - N_g(y)$. We can lower bound the maximum effectiveness by the average effectiveness, which is: $\frac{1}{deg(x)} \cdot \sum_{y : (x,y) \in G}{(N_\ell(y) - N_g(y))} = \frac{1}{deg(x)} \cdot \Big ( \sum_{y : (x,y) \in G}{N_\ell(y)} - \sum_{y : (x,y) \in G}{N_g(y)} \Big ) \ge \frac{b(C+1) - (2(b-1)+deg(x) + 1)}{deg(x)}$.
The last inequality follows since there are at least $b(C+1)$ leaves in total in $T^{x+}$, and by applying Lemma~\ref{lemma_descendants_sum} 
to $T^{x} \cup \{parent(root(T^x))\}$ (it has $L=b$ leaves and 
our set $S$ consists of $\deg(x)+1$ neighbours of $x$ where the $+1$ is due to possibly double-counting a neighbour $y=u$.)

It follows that there is a neighbour $y$ of $x$ with \emph{integer} effectiveness of at least $\frac{b(C-1)+1}{deg(x)} - 1$. If this effectiveness is strictly larger than the number of bad neighbours of $(x,y)$ on the side of $x$, which is at most $(deg(x)+deg(y)-1) - (\Delta+C) \le (deg(x)-1-C)$, then even though $N_g(y)$ paths in $\mathcal{P}$ might have been spoiled, there are still enough unspoiled paths to guarantee that at least one reaches $(x,y)$ from a good neighbour. Therefore, in conclusion, we require the following condition:
$\frac{b(C-1)+1}{deg(x)} - 2 > (deg(x)-1-C)$. (The $-2$ instead of $-1$ in the left-hand side is due to the possible case of $(p_i,x) = (u,v)$, and the strict inequality guarantees having at least one more path in $\mathcal{P}$ than the number of bad neighbours, because both right hand side, and the effectiveness of $y$ which is at least the left-hand side, are integers.) The condition is hardest to satisfy when $deg(x) = \Delta$, so in total we get that if $\frac{b(C-1)+1}{\Delta} > \Delta-C+1$ we can extend the coloring to $(u,v)$.

\looseness=-1
By Lemma~\ref{lemma_running_time_without_handler},
it takes $O(m)$ time to run Algorithm~\ref{algorithm_generic_shift_tree} excluding the call to the leaves handler. The leaves handler, Algorithm~\ref{algorithm_handle_leaves_generic}, has to determine the correct neighbour $y$ and edge $(x,y)$ among the multiple copies. Computing the node $y$ with maximum \emph{effectiveness} takes $O(\Delta^2)$ time, by scanning the $deg(x) = O(\Delta)$ neighbours of $x$ and for each counting its appearances in $T^{x+}$, and subtracting the number of its grandchildren in $T^x$ (together with the contribution of the virtual parent of $root(T^x)$, if relevant). To determine a copy of the edge $(x,y)$ that ends a useful path in $T^{x+}$, we compare the colors available at $x$ and at $y$ as follows. The colors available in $y$ are those originally available in $y$, $A(y)$, minus colors that could get spoiled due to grandchildren-edges, denote this list $A'(y)$ (can be computed in $O(\Delta)$ time). Next, go over the $O(\Delta)$ leaves that correspond to $x$ and recall that the color that gets freed in $x$ is $color(p_i,x)$, which is the original color of this edge. So if this color is free in $A'(y)$, we found a useful path. The total time is $O(\Delta^2)$, dominated by the time required to determine $y$.
\end{proof}

\begin{remark}[Self-intersection]
\label{remark_path_with_possible_repeat_edge}
As discussed following Definition~\ref{definition_chain_shift}, most shift-based algorithms in the literature refrain from shifting over a self-intersecting path. In contrast, our algorithm implied by Theorem~\ref{theorem_large_palette_general_parameterized} (Algorithm~\ref{algorithm_generic_shift_tree} with Algorithm~\ref{algorithm_handle_leaves_generic}) is mostly non-self-intersecting in edges \emph{except} for the last edge. This possibly gives us a slight advantage, and lets us tighten the recourse in Theorem~\ref{theorem_tight_worst_case_recourse}. To clarify the self-intersection: By the definition of the shift-tree $T$ every inner node on a path is unique, but the leaf might be repetitive. Then, when we expand the leaf $x$ one step further to a new leaf $y$, the final edge $(x,y)$ might in fact also appear earlier on the shifted path.
\end{remark}

\begin{figure*}[t]
	\centering
    \begin{subfigure}[t]{.16\textwidth}
        \centering
        \vspace{-0.9in} 
        \includegraphics[width=\textwidth]{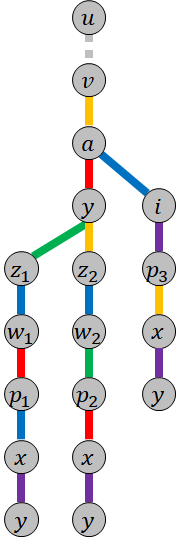}
        \subcaption{Shift-tree $T^x$}
        \label{fig_bad_shift_tree}
    \end{subfigure}
    \hspace{3mm} 
    \centering
    \begin{subfigure}[t]{.8\textwidth}
            \begin{subfigure}[t]{.4\textwidth}
                \centering
                \includegraphics[width=\textwidth]{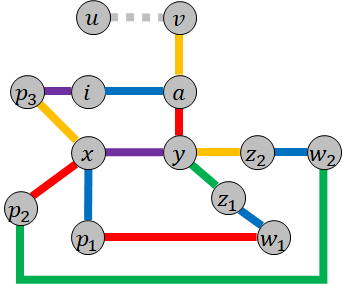}
                \subcaption{Initial Coloring}
                \label{fig_bad_shift_graph}
            \end{subfigure}
            \hspace{10mm} 
            \begin{subfigure}[t]{.4\textwidth}
                \centering
                \includegraphics[width=\textwidth]{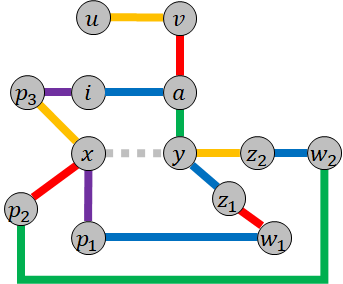}
                \subcaption{Path 1 shifted}
                \label{fig_bad_shift_graph1}
            \end{subfigure}
            
            \bigskip 
            \hspace{3mm} 
            \begin{subfigure}[t]{.4\textwidth}
                \centering
                \includegraphics[width=\textwidth]{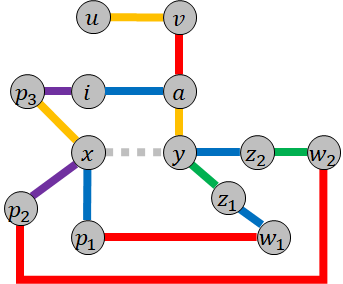}
                \subcaption{Path 2 shifted}
                \label{fig_bad_shift_graph2}
            \end{subfigure}
            \hspace{6mm} 
            \begin{subfigure}[t]{.4\textwidth}
                \centering
                \includegraphics[width=\textwidth]{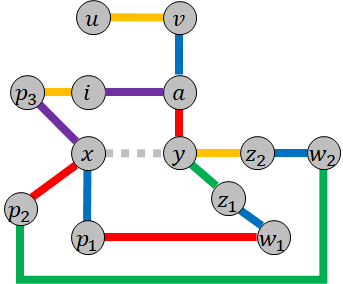}
                \subcaption{Path 3 shifted}
                \label{fig_bad_shift_grap3}
            \end{subfigure}
    \end{subfigure}
    \caption{\small{Illustration for the proof of Theorem~\ref{theorem_large_palette_general_parameterized}, of how shifting colors through $y$ might change the good/bad neighbours of an edge $(x,y)$, with $\Delta=4$, $C=1$. (a) A shift-tree (expanded skeleton) $T^{x+}$, showing multiple shift paths from $(u,v)$ (uncolored) to $(x,y)$ (we only show the children $y$ of $x$). (b) A sub-graph of $G$ induced only by the edges of $T^{x+}$, other edges of $G$ are omitted. Based on the initial coloring, the edge $(p_1,x)$ is a good neighbour of $(x,y)$, while $(p_2,x)$ and $(p_3,x)$ are bad neighbours due to having the same color as $(y,a)$ and $(y,z_2)$ respectively. However, (c) shows color shifting on the leftmost path (through $(y,z_1)$) that turns $(p_1,x)$ into a bad neighbour because $(y,z_1)$ has been colored blue during the shift. Similarly, but conversely, (d) shows that color shifting on the middle path (through $(y,z_2)$) makes $(p_2,x)$ a good neighbour because red is freed at $y$ during the shift. (e) When we shift colors over the rightmost path (through $(a,i)$) the initial status of $(p_3,x)$ is not changed because the path does not pass through $y$.}}
	\label{figure_shift_color_makes_bad}
\end{figure*}

We refer to all the pairs $(C,\Delta)$  for which Theorem~\ref{theorem_large_palette_general_parameterized} guarantees successful recoloring as the \emph{feasible domain}. This domain is analyzed by Theorem~\ref{theorem_feasible_domain}.

\begin{restatable}[Feasible Domain]{theorem}{theoremUpperBoundFeasibleDomain}
\label{theorem_feasible_domain}
For a fixed $\Delta \ge 3$, let $C^*(\Delta) = \frac{\Delta}{\phi} + f(\Delta)$ where $\phi \approx 1.618$ is the golden-ratio and $f(\Delta) \in (0.609,0.724)$ is defined precisely in the proof. Then for every $C > C^*(\Delta)$, we can deterministically maintain a dynamic $(\Delta + C)$ edge coloring using the algorithm implied by Theorem~\ref{theorem_large_palette_general_parameterized} (Algorithm~\ref{algorithm_generic_shift_tree} with Algorithm~\ref{algorithm_handle_leaves_generic}).
\end{restatable}

\begin{proof}
Recall the constraint of Theorem~\ref{theorem_large_palette_general_parameterized}: $\frac{b \cdot (C-1) + 1}{\Delta} > \Delta - C + 1$. Since $b$ only affects the left-hand side, the constraint is easiest to satisfy if we maximize $b$ and choose $b=C$. We get a simple quadratic inequality to solve:
$
C(C-1) + 1 > \Delta(\Delta-C+1)
\Rightarrow
C^2 + (\Delta-1) C - (\Delta^2 + \Delta - 1) > 0
$. Then:
$C > \frac{-(\Delta-1) + \sqrt{(\Delta-1)^2 + 4(\Delta^2 + \Delta - 1)}}{2} = 
\frac{\sqrt{5\Delta^2 + 2\Delta - 3} - (\Delta-1)}{2}
\equiv C^*(\Delta)
$.
By definition of $C^*(\Delta)$, choosing $C > C^*(\Delta)$ ensures feasibility. Since $C^*(\Delta) \approx \frac{(\sqrt{5}-1)\Delta}{2} = \frac{\Delta}{\phi}$, define $f(\Delta) \equiv C^*(\Delta) - \frac{\Delta}{\phi}$.
One can verify that $f$ is monotonically increasing (take the derivative), with $f(1) = 1-\frac{1}{\phi} \approx 0.382$, $f(2) \approx 0.555$, $f(3) \approx 0.610$, $\ldots$, and with an asymptote of $f(\Delta) \overset{\Delta \to \infty}{\to} \frac{1}{2} + \frac{\sqrt{5}}{10} \approx 0.724$.
\end{proof}

\begin{remark}[Idealized Smallest $C$]
\label{remark_half_Delta_at_best}
The shift-tree technique allows us to demonstrate that roughly $C > \frac{\Delta}{\phi} \approx 0.6\Delta$ is sufficient to apply our technique (Theorem~\ref{theorem_feasible_domain}). However, because the shift-tree expands each vertex at most once, the method is inherently limited to the regime $C > 0.5 \Delta$. 

To see why, consider that reaching an edge $(x,y)$ from multiple directions through $x$---as required in the proof of Theorem~\ref{theorem_large_palette_general_parameterized}---is possible from at most $\deg(x)-1 \le \Delta-1$ directions. Under optimal conditions where the parameter $b=\Delta-1$, we obtain the following constraint:
\[
    \frac{(\Delta-1)(C-1)+1}{\Delta} > \Delta-C+1,
\]
which simplifies approximately to $2C > \Delta$. In fact, achieving $b=\Delta-1$ is unrealistic; as noted following Lemma~\ref{lemma_shift_tree_max_depth}, one can generally only guarantee $b \le C+1$ leaves for \emph{some} vertex.
\end{remark}

By Theorem~\ref{theorem_large_palette_general_parameterized}, a $(\Delta+C)$ edge coloring of a dynamic graph $G$ can be maintained with a recourse that is a function of $C$ and the additional parameter $b$ mentioned in this theorem (which may also depend on $\Delta$). We show that if $(\Delta,C)$ is a feasible pair such that $(C/\Delta)-1/\phi = \Omega(1)$, and $\Delta-C = O(n^{1-\delta})$ for a constant $\delta > 0$, then there is a choice of $b$ such that this recourse is tight. The tightness is with respect to 
the lower bound $\Omega \Big (\log \frac{nC}{(\Delta-C)^2} / \log \frac{\Delta+C}{\Delta-C} \Big )$ by Theorem~\ref{theorem_lower_bound_2018}.

\begin{restatable}[Tight Domain]{theorem}{theoremUpperBoundTightnessDomain}
\label{theorem_tightness_domain}
For every \emph{feasible pair} $(\Delta,C)$, there is a choice of the parameter $b$ in Theorem~\ref{theorem_large_palette_general_parameterized} that yields  a dynamic $(\Delta+C)$ edge coloring algorithm with worst-case $O(\frac{\log n}{\log \frac{\Delta+C}{\Delta-C}} \cdot \frac{1}{(C/\Delta) - 1/\phi})$ recourse per update.
If $\Delta - C = O(n^{1-\delta})$ for a constant $\delta > 0$, and $(C/\Delta)-1/\phi = \Omega(1)$, then this recourse is \emph{tight} and matches the lower bound in Theorem~\ref{theorem_lower_bound_2018} up to a constant factor.
\end{restatable}

\begin{proof}
To show tightness, we match the upper and lower bound on the recourse. First, let us rewrite the lower bound of Theorem~\ref{theorem_lower_bound_2018}. The feasible range satisfies $C > \frac{\Delta}{2}$, therefore $1 < \frac{C}{\Delta-C} < \frac{n}{\Delta-C}$ so the lower bound is:
$$
\Omega \Big (\frac{\log \frac{nC}{(\Delta-C)^2}}{\log \frac{\Delta+C}{\Delta-C}} \Big )
=
\Omega \Big (\frac{\log \frac{n}{\Delta-C}}{\log \frac{\Delta+C}{\Delta-C}} \Big )
\underset{\Delta-C = O(n^{1-\delta})}{=}
\Omega \Big (\frac{\delta \cdot \log n}{\log \frac{\Delta+C}{\Delta-C} } \Big )
\underset{\delta=\Theta(1)}{=}
\Omega \Big (\frac{\log n}{\log \frac{\Delta+C}{\Delta-C} } \Big )
$$
By Theorem~\ref{theorem_large_palette_general_parameterized} the upper bound can be rewritten as follows:
$$
O \Big (\frac{\log n}{\log \frac{C+1}{b}} \Big )
=
O \Big (\frac{\log n}{\log \frac{\Delta+C}{\Delta-C}} \cdot \frac{\log \frac{\Delta+C}{\Delta-C}}{\log \frac{C+1}{b}} \Big )
$$
To complete the proof and get the claimed recourse, we show that there is a feasible choice of $b$ such that $f(\Delta,C,b) \equiv \frac{\ln \frac{\Delta+C}{\Delta-C}}{\ln \frac{C+1}{b}} = O(\frac{1}{(C/\Delta) - 1/\phi})$. (Then, also $f(\Delta,C,b) = O(1)$ for $(C/\Delta) - 1/\phi = \Omega(1)$, which concludes the tightness claim.) By Theorem~\ref{theorem_large_palette_general_parameterized} the feasibility of $b$ requires that (a) $2 \le b \le C$ and (b) $\frac{b(C-1)+1}{\Delta} > \Delta-C + 1$. We divide the proof into two cases, for $C \ge 0.9\Delta$ we choose $b = 5(\Delta-C)$ and for $C < 0.9\Delta$ we choose $b = 1 + \floor{\frac{(\Delta-C+1)\Delta - 1}{C-1}}$.

\textbf{Case 1:} $C \ge 0.9\Delta$ with $b = 5(\Delta-C)$. First we verify that conditions (a) and (b) hold for this choice of $b$. Indeed, (a) Since $C \le \Delta - 2$ then $b > 2$. Also, $C \ge 0.9 \Delta \Rightarrow \frac{10}{9} C \ge \Delta \Rightarrow \frac{1}{9} C \ge (\Delta - C) = \frac{b}{5} \Rightarrow \frac{9}{5} b \le C$, then indeed $b < C$. (b) We have to verify that $\frac{5(\Delta-C)(C-1)+1}{\Delta} > \Delta-C + 1$. We can re-arrange this expression to get:
\begin{equation}
\label{equation_feasbility_case1}
(\Delta-C) \cdot \Big (\frac{5(C-1)}{\Delta} - 1 \Big ) > 1 - \frac{1}{\Delta}
\end{equation}
Equation~(\ref{equation_feasbility_case1}) holds because:
$$
(\Delta-C) \cdot \Big (\frac{5(C-1)}{\Delta} - 1 \Big )
\underset{\Delta-C \ge 2}{\ge}
2 \cdot \Big (\frac{5C-\Delta}{\Delta} - \frac{5}{\Delta} \Big )
\underset{C \ge 0.9\Delta}{\ge}
7 - \frac{10}{\Delta}
\underset{\Delta \ge 2}{\ge}
2 > 1 - \frac{1}{\Delta}
$$
Finally, we analyze $f(\Delta,C,b) = \frac{\ln \frac{\Delta+C}{\Delta-C}}{\ln \frac{(C+1)/5}{\Delta-C}}$.
First note that $\frac{\Delta+C}{\Delta-C}
\le
\frac{\frac{10}{9}C+C}{\Delta-C}
<
\frac{(95/9) \cdot (C+1)/5}{\Delta-C}$,
then $f(\Delta,C,b) < 1 + \frac{\ln (95/9)}{\ln \frac{(C+1)/5}{\Delta-C}}$. Next recall that $\frac{9}{5} < \frac{C+1}{b} = \frac{(C+1)/5}{\Delta-C}$, then $f(\Delta,C,b) < 1 + \frac{\ln (95/9)}{\ln (9/5)} \approx 5.01$. Also, since $C \ge 0.9\Delta$, $\frac{1}{(C/\Delta) - 1/\phi} = \Theta(1)$, thus $f = O(\frac{1}{(C/\Delta) - 1/\phi})$. The first case is complete.

\textbf{Case 2:} $C < 0.9\Delta$ with $b = 1 + \floor{\frac{(\Delta-C+1)\Delta - 1}{C-1}}$. We verify that conditions (a) and (b) hold for this choice of $b$.  (a) $C \le \Delta - 2$ so $\Delta-C+1 \ge 3$, substituting this in the expression for $b$, we get that $b \ge 1 + \floor{\frac{2(\Delta-1)}{C-1}} \ge 3$. The choice of $b$ is the minimal that satisfies Condition (b). Since any $b' \ge b$ also satisfies (b), and $b' = C$ does (by assumption that $(\Delta,C)$ is a feasible pair), we conclude that $b \le C$.  This completes the proof that Condition (a) holds.\footnote{As a more explicit proof that $b \le C$ for $C \ge \frac{\Delta}{\phi} + 4$, notice that $b = 1 + \floor{\frac{\Delta^2 - 1 - (C-1)\Delta}{C-1}} = 1-\Delta + \floor{\frac{\Delta^2 - 1}{C-1}} \le 1-\Delta + \frac{\Delta^2}{C-1} = 1 - \Delta + \frac{\Delta^2}{C} \cdot \frac{C}{C-1} < 1 + \Delta (\phi \cdot \frac{C}{C-1} - 1) = 1 + \Delta (\frac{1}{\phi} + \frac{1}{C-1}) < \frac{\Delta}{\phi} + 1 + \phi \cdot \frac{C}{C-1} < \frac{\Delta}{\phi} + 3.16 < C$.
}

Now proceed to analyze $f(\Delta,C,b)$ (recall: $f(\Delta,C,b) \equiv \frac{\ln \frac{\Delta+C}{\Delta-C}}{\ln \frac{C+1}{b}}$). Since $(C/\Delta) \in (1/\phi,0.9)$ its numerator is $\Theta(1)$ (trivially being increasing in $C$ and attaining values of $\ln \frac{\phi}{2-\phi}$ and $\ln 19$ at the ends of this interval), so it suffices to prove that $g(C,b) \equiv \ln \frac{C+1}{b}$ is $\Omega((C/\Delta) - 1/\phi)$ for the value we chose for $b$. We analyze this function by presenting it as $g(C,b) = \ln (1 + x)$ for $x \equiv \frac{C+1}{b} - 1$ and use $\frac{x}{9} < \ln (1+x) < x$ for $x \in (0,8)$, we first bound $x$:

$$
\frac{(\Delta-C+1)\Delta - 1}{C-1} \le b \le 1 + \frac{(\Delta-C+1)\Delta - 1}{C-1} \Rightarrow
$$
$$
\frac{C^2-1}{\Delta^2-C\Delta + (\Delta - 1)} \ge  \frac{C+1}{b} \ge \frac{C^2-1}{\Delta^2 - C\Delta + (\Delta + C - 2)} \Rightarrow
$$
\begin{equation}
\label{equation_sandwich_b_expression}
X_{LHS} \equiv \frac{C^2 + C\Delta - \Delta^2 - \Delta}{\Delta^2-C\Delta + (\Delta - 1)} \ge  \frac{C+1}{b} - 1 \ge \frac{C^2 + C\Delta - \Delta^2 - (\Delta + C - 1)}{\Delta^2 - C\Delta + (\Delta + C - 2)} \equiv X_{RHS}
\end{equation}

Observe that $X_{LHS}$ is monotone in $C$ (numerator increases and denominator decreases), and therefore bounded by substituting $C = 0.9\Delta$: so $0 \underset{b \le C}{<} x = \frac{C+1}{b} - 1  \le \frac{0.71 \Delta^2 - \Delta}{0.1 \Delta^2 + (\Delta-1)} < 7.1$. It holds in general for $x \in (0,8)$ that $\ln (1+x) > \frac{x}{9}$, therefore $9 \cdot g(C,b) > \frac{C+1}{b} - 1$. To bound $\frac{C+1}{b} - 1$ from below, we proceed to two sub-cases based on the value of $y$ where we denote $C = \frac{\Delta}{\phi} + y$. Then $q \equiv C^2 + C\Delta - \Delta^2 = \Delta^2 (\frac{1}{\phi^2} + \frac{1}{\phi} - 1) + (\frac{2}{\phi} + 1)y \Delta + y^2  = (\frac{2}{\phi} + 1)y \Delta + y^2$. For $y \ge 1$, $q > 2\Delta > (\Delta+C-2)$. Therefore $X_{RHS}$ is a meaningful, positive lower bound, and we can rewrite it as follows, denote $z \equiv \Delta^2 - C\Delta + (\Delta+C-2)$ and note that $z = \Theta(\Delta^2)$:
$$
9 \cdot g(C,b) >
\frac{C+1}{b} - 1 \ge
X_{RHS}
>
\frac{q - 2\Delta}{z}
>
\frac{(\frac{2}{\phi}+1)y\Delta + y^2 - 2\Delta}{z}
>
\frac{(\frac{2}{\phi}-1)y\Delta}{z}
=
\Theta(\frac{y}{\Delta})
$$
Substituting $y = C - \frac{\Delta}{\phi}$ gives $g(C,b) = \Omega((C/\Delta) - 1/\phi)$.

For $y < 1$, by Theorem~\ref{theorem_feasible_domain} we also have $y > 0.609$, then $C/\Delta-1/\phi = \frac{y}{\Delta} = \Theta(\frac{1}{\Delta})$. Then:
$$
\frac{C+1}{b} - 1
\underset{b \le C}{\ge}
\frac{1}{b}
\ge
\frac{1}{1 + \frac{(\Delta-C+1)\Delta-1}{C-1}}
=
\frac{C-1}{\Delta^2 - C\Delta + (\Delta+C-2)} = \Theta \left (\frac{1}{\Delta} \right )
=
\Theta ((C/\Delta) - 1/\phi)
$$
Then again since $9 \cdot g(C,b) > \frac{C+1}{b} - 1$, we get that $g(C,b) = \Omega((C/\Delta) - 1/\phi)$.

\end{proof}

Finally, Theorem~\ref{theorem_tight_worst_case_recourse} follows immediately from Theorem~\ref{theorem_feasible_domain} and Theorem~\ref{theorem_tightness_domain} (recourse) and Theorem~\ref{theorem_large_palette_general_parameterized} (running time). We restate the theorem for convenience.

\theoremTightWorstCaseRecourse*

Theorem~\ref{theorem_tight_worst_case_recourse} does not extract the whole power of the shift-tree technique, and it is possible that other leaf-handler concepts can extend the feasibility domain further. As a concrete example, 
notice that the only values of $C$ that satisfy the requirement that $C \ge \frac{\Delta}{\phi} + 0.724$
in Theorem~\ref{theorem_tight_worst_case_recourse},
for $\Delta \in \{ 4,5,6,7 \}$
are $C \ge \Delta - 1$.\footnote{Technically by Theorem~\ref{theorem_tight_worst_case_recourse}, $\Delta=7$ also seems to require $C > 5.05 \Rightarrow C \ge 6$, but the more accurate condition which we relaxed, given in Theorem~\ref{theorem_feasible_domain}, requires $C > 4.94 \Rightarrow C \ge 5$ for $\Delta=7$.} For such values of $C$ dynamic coloring is of course trivial.
However, with a different leaf handler we can use the shift tree in a different way to give an algorithm with tight recourse for $C=\Delta-2$ and any $\Delta \ge 4$ (in particular for $\Delta =4,5,6,7$).

\begin{algorithm}[t]
    \SetKwFunction{funcLeavesHandler}{LeavesHandler}
    \SetKwProg{Fn}{Function}{:}{}
    
    \Fn{\funcLeavesHandler{caller's context: computed shift-tree $T$ with $b=2$ leaves $x_1$ and $x_2$ that correspond to the same vertex $x \in G$}}{
        \If{a non-leaf copy of $x$ in $T$, $x'$, is an ancestor of $x_i$ but not of $x_{3-i}$}{
        Re-define $x_i$ to be $x'$. // Applies for $i \in \{1,2\}$}

        \looseness=-1
        Let $P_i$ be the path in $T$ from $e$ to $e_i = (parent_T(x_i),x_i)$. Let $P$ be the common prefix of $P_1$ and $P_2$, and let $S_i$ be the unique suffix of $P_i$. Define $P^{12} \equiv P_1 + rev(S_2)$ and $P^{21} \equiv P_2 + rev(S_1)$ where $+$ is concatenation and $rev$ is reversal.

        Find a prefix of either $P^{12}$ or $P^{21}$ such that when shifting colors over it, we get a new uncolored edge $e'$ that can be colored with some available color. Return this prefix as the path $P$.
    }

    \caption{A leaves-handler used with Algorithm~\ref{algorithm_generic_shift_tree}. Corresponds to Theorem~\ref{theorem_DeltaMinus2_colors}.}
    \label{algorithm_handle_leaves_DeltaMinus2}
\end{algorithm}

\begin{restatable}[Simplified $C = \Delta-2$]{theorem}{theoremDeltaMinusTwoSimpler}
\label{theorem_DeltaMinus2_colors}
Let $\Delta \ge 4$. Algorithm~\ref{algorithm_generic_shift_tree} with parameter $b=2$,\footnote{We add the exception that for the uncolored edge $(u,v)$, we do not count a leaf of $v$ if its parent is $u$. So $b=2$ except for possibly $v$ which might require us to see $3$ leaves to trigger the leaves handler.} with Algorithm~\ref{algorithm_handle_leaves_DeltaMinus2}, deterministically maintain a dynamic $(2\Delta -2)$ edge coloring with a tight worst-case recourse of $O(\frac{\log n}{\log \Delta})$ per insertion. Determining the recoloring takes $O(m)$ time.
\end{restatable}

\begin{proof}
Let $e=(u,v)$ be a newly inserted edge. We expand a shift-tree $T$ from $e$ until we determine a useful shiftable path or see $b=2$ leaf copies of some vertex $x$ (or up to $3$ if $x=v$). By Lemma~\ref{lemma_shift_tree_max_depth} we get that $depth(T) = O(\log_\beta n)$ where $\beta = \frac{\Delta-1}{2}$, so  
$depth(T) = O(\frac{\log n}{\log \Delta})$. Note that feasibility requires a choice of $b \in [2,C]$, therefore $C \ge 2$ and thus $\Delta = C+2 \ge 4$. This depth is asymptotically the same as the lower bound of Theorem~\ref{theorem_lower_bound_2018} for $C = \Delta-2$: $\Omega(\log_{(\Delta-1)}{(n (\Delta-2))}) = \Omega(\frac{\log n}{\log \Delta})$. We clarify that even though Lemma~\ref{lemma_shift_tree_max_depth} assumes a fixed $b$, our exception to not count a single possible edge (a leaf $v$ with parent $u$) does not make an asymptotic difference in the depth of the shift-tree.

\looseness=-1
Now we extend the coloring, by recoloring no more than $2 \cdot depth(T)$ edges, which proves the recourse statement. If $T$ contains a useful shiftable path, we simply use it. Otherwise, consider the skeleton $T^x$, which has $2$ leaves. If there is an inner node $x$ that is only an ancestor of one of the leaves, truncate $T^x$ so that the inner $x$ becomes a leaf. Denote the two leaves by $x_1$ and $x_2$, their LCA by $y$, and its parent by $p$ (it could be that $y$ or $p$ is the active copy of $x$). First, shift the colors on the path from $e$ to $(p,y)$ so that $(p,y)$ is now uncolored. If $y=v$ we shifted nothing.

We now have an uncolored edge $(p,y)$, that is near a cycle of edges that visits $y$ and $x$, colored by their original colors when we computed $T$. This is why we ignored $(u,v)$ when counting leaf copies of $v$, to ensure that $(u,v)$ is not in the cycle, because it is the only edge without an original color. Denote by $z_1$ the child of $y$ which is an ancestor of $x_1$, and by $z_2$ the child of $y$ that is an ancestor of $x_2$. Note that $z_1$ could be $x_1$ and $z_2$ could be $x_2$ but not both. In particular $z_1\not= z_2$.
The cycle is either simple or an ``8''-figure in $G$ because of the following. Every non-leaf node in the shift-tree is unique, so the only vertex that can repeat is $x$:
\begin{enumerate}
    \item If $y \ne x$: Then the cycle is simple and contains at least three vertices.
    
    \item If $y=x$ (possible if we expanded $x$ and found the two leaves of $x$ in its own subtree): Then the path from $y=x$ to $x_i$ forms a cycle for $i \in \{1,2\}$. These cycles are disjoint in edges because they are disjoint in vertices except for $x$. Each of these cycles has at least three vertices (including $x$) because $y=x$ cannot be a parent (trivial) or a grandparent (Observation~\ref{observation_parent_child_in_shift_tree}) of $x_i$. Hence, the cycle is indeed an ``8''-figure.
\end{enumerate}

We claim that we can shift the colors and finish as follows: First, shift the color on the path from $(p,y)$ to $(y,z_1)$ or to $(y,z_2)$. Then, we continue shifting the colors around the cycle. If the cycle is not simple ($y=x$), it does not matter how we define the cycle in $G$, though arguably it is more intuitive to think of it as traversing on $T$ from $y$ to $x_1$ and then traversing back up from $x_2$ to $y$ (similar to the case of a simple cycle where $y \ne x$).

We shift until an uncolored edge can be colored by a free color. We argue that this method must succeed for either starting from $(y,z_1)$ or from $(y,z_2)$. See Figure~\ref{figure_tail_and_cycle} to clarify the remainder of the proof.

\begin{figure*}[t]
	\centering
    \begin{subfigure}[t]{.18\textwidth}
        \centering
        \includegraphics[width=\textwidth]{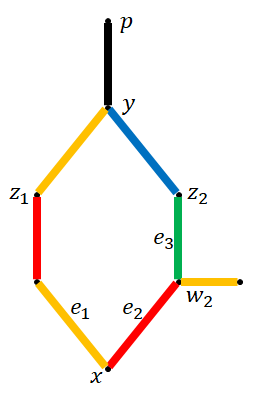}
        \subcaption{initial colors}
        \label{figure_tail_and_cycle_init}
    \end{subfigure}
    \hspace{2mm} 
    \begin{subfigure}[t]{.18\textwidth}
        \centering
        \includegraphics[width=\textwidth]{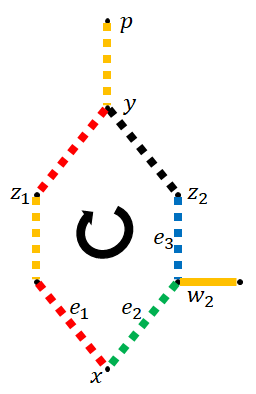}
        \subcaption{full shift}
        \label{figure_tail_and_cycle_nonstop}
    \end{subfigure}
    \hspace{2mm} 
    \begin{subfigure}[t]{.18\textwidth}
        \centering
        \includegraphics[width=\textwidth]{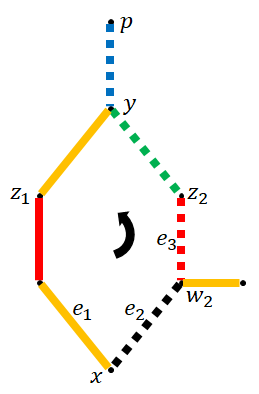}
        \subcaption{stopped shift}
        \label{figure_tail_and_cycle_withstop}
    \end{subfigure}
    \hspace{2mm} 
    \begin{subfigure}[t]{.36\textwidth}
        \centering
        \includegraphics[width=\textwidth]{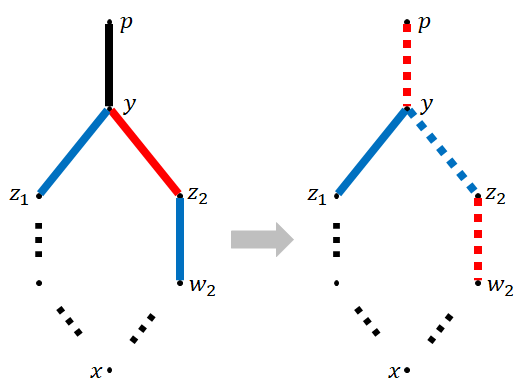}
        \subcaption{Impossible bicolored cycle}
        \label{figure_cycle_shift}
    \end{subfigure}
    
    \caption{\small{Illustrating the proof of Theorem~\ref{theorem_DeltaMinus2_colors}.
    We shifted the colors on the path to $(p,y)$ which is now uncolored (black). (a) Some initial coloring of a cycle. (b) Shifting colors over a full cycle, dashed edges are whose colors were shifted. (c) The shift in the opposite direction in this example stops on $e_2$ because $e_1$ is a bad neighbour with bad color yellow. However, this is actually good for us: given $2\Delta-2$ colors in total, and since only at most $2\Delta-3$ are used around $e_2$ (two neighbours reuse yellow, and $e_2$ has no color), we can simply color $e_2$ and stop. If we could not stop on $e_2$ from either side, then $e_1$ and $e_3$ would have had the same (bad) color. (d) Once we conclude towards a contradiction that we cannot stop in either direction, on any edge of the cycle, we deduce that the cycle is bicolored. But this cannot be by definition of a shift-tree: Shifting the color of $(z_2,w_2)$ up to $(y,z_2)$ will conflict with $(y,z_1)$ that remains unchanged.}}
	\label{figure_tail_and_cycle}
\end{figure*}

Assume for contradiction that no prefix of color shifting around the cycle results in an uncolored edge that has a free color and consider a specific edge $e_2$ in this cycle with neighbours $e_1$ and $e_3$ ($e_2$ is not ``first'' or ``last'' along the cycle, concretely $e_2 \ne (y,z_1)$ and $e_2 \ne (y,z_2)$). If we cannot stop on $e_2$ when we rotate colors in the cycle in either direction, then $color(e_1) = color(e_3)$. Indeed, because $C = \Delta-2$ there can be at most one bad neighbour (color) on each side of $e_2$, and if the shift arrives from a different color, we can pick a free color for $e_2$ and be done. Conversely, the same argument shows that if $color(e_1) \ne color(e_3)$ then we can stop on $e_2$ when shifting in either direction. Applying this argument to every edge on the cycle that is not the first or the last, implies that the cycle is a bicolored (closed) path. This requires the cycle to be even, but even so we get the following contradiction. Because at least one of $z_1$ and $z_2$ is not $x$, assume without loss of generality that $z_1 \ne x$ and therefore it has a child $w_1$ on the path to $x_1$ (possibly $w_1=x_1$). By definition of the shift-tree, $w_1$ can be a child of $z_1$ only if $color(z_1,w_1)$ is free in $y$ (the parent of $z_1$). However, because the cycle is even and bicolored, $color(y,z_2) = color(z_1,w_1)$: a contradiction.
In conclusion, it must be possible to shift the colors partially over the cycle in one direction and be able to finish the coloring on some intermediate edge. The total length of the cycle and the path from the root of $T$ to $(p,y)$ that we shifted at the beginning to reach the cycle is at most twice the depth of $T$.

The running time to extend the coloring is $O(m)$ for Algorithm~\ref{algorithm_generic_shift_tree} by Lemma~\ref{lemma_running_time_without_handler} plus the time it takes Algorithm~\ref{algorithm_handle_leaves_DeltaMinus2} to determine the useful path. We argued in the previous paragraph that determining the availability of \emph{some} color for an edge $e_2$ only requires us to compare the colors of its two neighbours on the cycle, which takes $O(1)$ time, therefore it takes $O(|cycle|)$ time in total, which is $O(|depth(T)|)$. Even though this test per edge does not determine the available color, Algorithm~\ref{algorithm_handle_leaves_DeltaMinus2} only has to determine the path (to be used by Algorithm~\ref{algorithm_generic_shift_tree}). Since $depth(T) \le |T| = O(m)$, the total running time remains $O(m)$.
\end{proof}

The shift-tree technique is also useful for analyzing low arboricity graphs.

\theoremLowArboricity*

\begin{proof}
We compute a shift-tree $T$ on $G$ with respect to the inserted edge. We expand levels $0$ to $i \ge 1$ of $T$, and prove that by depth $O(\log_B n)$ we must encounter a useful shiftable path because the nodes down to depth $i$ correspond to at least $B^i$ unique vertices of $G$ where $B \equiv \minCarboricityB{}$. The recourse is the length of this path, and when substituting $C = \minCarboricity{}$ we get $B = 1 + \frac{\epsilon \alpha}{2 \alpha - 1}$ thus $O(\log_B n) = O(\log_{1+\epsilon} n)$. For $\epsilon < 1$, using $\ln (1+\epsilon) = \Theta(\epsilon)$ we can rewrite the recourse as $O(\frac{1}{\epsilon} \cdot \log n)$. Note that Algorithm~\ref{algorithm_generic_shift_tree} suffices, even without having to specify a ``leaves handler'' subroutine.

To show the exponential number of vertices, consider the expansion of $T$ at some given depth $i \ge 0$ and denote by $N_i$ the number of vertices of $G$ that have corresponding nodes in $T$ at any level down to and including $i$. Note that $N_0 = 1$ (the root). Let $G_i = (V_i,E_i)$ be the subgraph of $G$ induced by these $|V_i| = N_i$ vertices. Let $i \ge 1$, and now bound $|E_i|$ in two ways. First, there are $N_{i-1}$ inner nodes, each with at least $C+1$ children by Lemma~\ref{lemma_basic_shift_children}. This might double-count edges by Observation~\ref{observation_parent_child_in_shift_tree}, but it does not double-count the edges between the new $N_i - N_{i-1}$ nodes to their parents, therefore: $|E_i| \ge \frac{N_{i-1} \cdot (C+1) + (N_i-N_{i-1})}{2}$. On the other hand, the arboricity of $G_i$ is at most $\alpha$ being a subgraph of $G$, therefore $|E_i| \le \alpha (N_i-1)$. Thus:
$$
N_{i-1} \cdot C + N_i \le 2|E_i| \le 2\alpha (N_i-1)
\Rightarrow B < \minCarboricityB{} + \frac{2\alpha}{(2\alpha-1) \cdot N_{i-1}} \le \frac{N_i}{N_{i-1}}
$$
Then the growth rate is indeed (at least) by a factor of $B$ from $N_i$ to $N_{i+1}$. Note that we required that $B>1$. (Relying on the marginal extra term $\frac{2\alpha}{(2\alpha-1) \cdot N_{i-1}}$ does not give an exponential growth rate if $B=1$ ($C = 2\alpha-1$), but only a linear rate.)
\end{proof}

\subsection{Adaptivity and Local-Palettes}
\label{subsection_adaptive_palettes}

In this section, assume that the colors have some natural total order, for example being positive integers. Recall that an algorithm for dynamic edge coloring is said to be adaptive with respect to a parameter if this parameter is not required to be known in advance. For example, an algorithm is $\Delta$-adaptive if it does not need to know an upper bound $\Delta$ in advance, and can adapt the maximum color being used (which implies the size of its palette) naturally to $\Delta_t$ which is the actual maximum degree of the graph at time $t$. Similarly, $\alpha$-adaptiveness applies if the algorithm does not require a bound on the arboricity, etc. Adaptivity is always preferred if one can achieve it, since in practice we may not know a good upper bound, and a conservative upper bound might be detrimental to performance.

One concept relevant for achieving adaptivity is that of \emph{local palettes}, where the coloring takes into account the degree of vertices to only use small colors for edges of vertices with small degree. Intuitively, this locality ensures that it is easier to get rid of high colors when the degree of a vertex gets smaller. Local palettes are used in multiple papers such as \cite{MultistepVizing2023, ArboricitySWAT2024_Christiansen}. Concretely, revisit Table~\ref{table_recourse_results} and note there the expressions with subscripts $t$ or $e$.

Using local palettes, we can adapt the shift-tree algorithm such that whenever we compute the children of a node $x \in T$, we consider for children only the colors available at $p = parent_T(x)$ that are also among the first (smallest) $\deg(p) + C$ colors. We are still guaranteed that $x$ has at least $C+1$ children unless there is a (small) shared available color at $p$ and $x$. This minor adaptation implies that now every shift-tree algorithm in this paper in fact maintains that $color(u,v) \le \max \{ \deg(u),\deg(v)\} + C$ for every edge $(u,v)$ that it recolors when processing an insertion. To make this invariant hold for all the edges at any time, define that a deletion of the edge $(x,y)$ checks if any edge $(x,z)$ or $(y,z')$ now violates this invariant. Since $\deg(x)$ and $\deg(y)$ are the only affected values, and each is decremented by exactly $1$, there can be at most two violating edges (if they were colored by the critically highest allowed color). Then upon deletion, we uncolor each of these edges and recolor the graph as if each of them has been a newly inserted edge. In conclusion, by paying twice the recourse of insertion per deletion, we now maintain the invariant that for every edge, always, $color(u,v) \le \max \{ \deg(u),\deg(v)\} + C$.

The invariant we get is nice, but still not quite $\Delta$-adaptive because $C$ itself is a parameter which depends on $\Delta$. Specifically, Theorem~\ref{theorem_tight_worst_case_recourse} requires $C > \frac{\Delta}{\phi}$. To be purely $\Delta$-adaptive, we may want to reduce $C$ when the maximum degree of the graph shrinks. This is much more complicated, because when the maximum degree decreases due to a single deletion, many edges in the graph might simultaneously violate the invariant due to the reduction in $C$. The same issue arises if $C$ is related to any other dynamic property that could change due to a single deletion, such as arboricity.

That being said, not all is lost. We can upgrade Theorem~\ref{theorem_low_arboricity} to Theorem~\ref{theorem_adaptive_Delta_for_low_arboricity} which still uses a fixed arboricity upper bound, to get a $\Delta$-adaptive (but not $\alpha$-adaptive) edge coloring algorithm with the same asymptotic recourse.


It is worth noting that local palettes are an avenue towards breaking existing worst-case lower bounds (Theorem~\ref{theorem_lower_bound_2018} and Theorem~\ref{theorem_lower_bound_generalized}). These bounds show that the recourse must be high given a specific coloring of a specific graph. However, if an algorithm $ALG$ can maintain, for example, that $color(u,v) \le \rho \cdot \max \{ deg(u),deg(v) \}$ where $\rho = \frac{3}{2}$ for every edge $(u,v)$, then if $\Delta+C \ge 10$ the lower bound of Theorem~\ref{theorem_lower_bound_2018} does not apply to $ALG$ because its coloring never reaches the state assumed by the theorem.\footnote{In the basic lower bound construction, the missing edge $(u,v)$ is adjacent to every possible color from $1$ to $\Delta+C$, even though the degrees of $u$, $v$, and their neighbours are only $1 + \frac{\Delta+C}{2}$.} More generally, any $\rho < 2 - \frac{4}{\Delta+C+2}$ breaks this lower bound.

\section{Lower Bounds: Recourse and Separation (Formal Section)}
\label{appendix_section_lower_bounds_and_separation}

In this section we study two interesting lower bounds. Subsection~\ref{subsection_generalized_lower_bound} is dedicated to generalizing Theorem~\ref{theorem_lower_bound_2018} of \cite{LowerBoundRecourse2018}, to take into account the arboricity $\alpha$ as an additional parameter. Subsection~\ref{appendix_subsection_recourse_separation} shows how to construct a partially colored graph such that extending its coloring requires only $O(1)$ recourse, yet the recourse of shift-based algorithms (Definition~\ref{definition_chain_shift}) is still large. This separation is stated in Theorem~\ref{theorem_gap_shift_recoloring_versus_general_separation}, which generalizes Theorem~\ref{theorem_gap_shift_recoloring_versus_general_separation_simple}.

\subsection{Generalized Recourse Lower Bound}
\label{subsection_generalized_lower_bound}

We restate Theorem~\ref{theorem_lower_bound_generalized} for convenience.

\theoremLowerBoundGeneralized*

\begin{proof}
The construction generalizes that of \cite{LowerBoundRecourse2018}, and also covers odd values of $\Delta+C$ (that were not covered there).

\textbf{Topology and coloring:} We define a graph $H \equiv H(\ell,r,\alpha,b_0,b_1)$ with $\ell$ layers
as follows. In layer $0$, $H$ has a single vertex $r$, which we refer to as its root (even though technically $H$ is undirected and not actually rooted). $r$ has $b_0$ neighbours in layer $1$. In layers $i \ge 1$, we divide the vertices into groups of size $\alpha$, and form a biclique for each group with $b_{i\ mod\ 2}$ new vertices in layer $i+1$. Vertices of layer $i$ that are not grouped (when the total number of vertices at layer $i$ is not an integral multiple of $\alpha$) are grouped first with vertices of layer $i+2$. We require that $\alpha \le b_0,b_1$, which ensures that every layer $j \ge 1$ has at least $\alpha$ vertices, which in turn ensures that every edge connects vertices that are either $1$ or $3$ layers apart. Indeed, any leftover vertices from layer $i$ will belong to the first group of $\alpha$ vertices when grouping layer $i+2$. Edges from layer $i$ to $i+1$ or to $i+3$ are colored from a palette $P_{i\ mod\ 2}$, where $|P_0| = b_0$, $|P_1| = b_1$, and $P_0 \cap P_1 = \emptyset$. Note that we must set $b_0$ and $b_1$ such that $b_0+b_1 \le \Delta+C$. Furthermore, we must satisfy that $\alpha + b_i \le \Delta$ for $i=0,1$, since a vertex may have $\alpha$ neighbours in a lower layer and $b_i$ neighbours in a higher layer. The requirement $\alpha \le b_0,b_1$ ensures that the sizes of the layers are mostly weakly increasing.\footnote{Initially the sizes of layers might not increase because the leftover portion might be significant, but eventually we get monotonicity. For example, with $b_0 = 2\alpha-1$ and $b_1 = \alpha+1$, the number of vertices at layer $i$, $n_i$ is: $n_0=1$, $n_1=2\alpha-1$, $n_2 = \alpha+1$ because almost half of the vertices in layer $1$ were not grouped.} Finally, let $G(\ell) \equiv \{ (u,v) \} \cup H(\ell,u,\alpha,b_0,\Delta+C-b_0) \cup H(\ell,v,\alpha,\Delta+C-b_0,b_0)$, where the edge $e=(u,v)$ is uncolored, and define $G \equiv G(L)$ where $L = argmax \{ \ell \ge 0 : |G(\ell)| \le n \}$. In other words, $L$ is the deepest complete layer we can construct with a budget of $n$ vertices (we can add and ignore singleton vertices to round up to $|G| = n$). See Figure~\ref{figure_biregular_graph} for an example. We will set $b_0$ and $b_1$ such that
 $b_0+b_1 = \Delta+C$ so that we use all the colors, which implies 
that  $\alpha \le \frac{\Delta-C}{2}$ (because $2\alpha + \Delta + C = (\alpha + b_0) + (\alpha + b_1) \le 2\Delta$) as required in the theorem.

\smallskip
\textbf{Outline:} First we show that the arboricity of $G$ is indeed in $[\frac{\alpha+1}{2},\alpha]$. Then we show that the recourse required to color $e$ is $\Omega(L)$,
and that $L$ is asymptotically maximized if we set $b_0 = \Delta-\alpha$ and therefore 
$b_1= \Delta+C-(\Delta - \alpha) = C+\alpha$.  We show that in this case $\Omega(L)$ is the expression stated by the theorem. Note that if $\Delta+C$ is even, the choice $b_0=b_1 = \frac{\Delta+C}{2}$ with $\alpha = \frac{\Delta-C}{2}$ gives the exact construction of \cite{LowerBoundRecourse2018}. In this case, every vertex except for those in layers $0,1,L-2,L-1,L$ has maximum degree: $\alpha + b_i = \Delta$.\footnote{A vertex in layers $0$ or $1$ has only $b_i+1 = \frac{\Delta+C}{2}+1$ neighbours. A vertex in layer $L$ has no ``children'', and so do at most $2(\alpha-1) = \Delta-C-2$ leftover vertices in each of layers $L-1$ and $L-2$.}

\begin{figure*}[ht]
	\centering
    \begin{subfigure}[t]{\textwidth}
        \centering
        \includegraphics[width=0.7\textwidth]{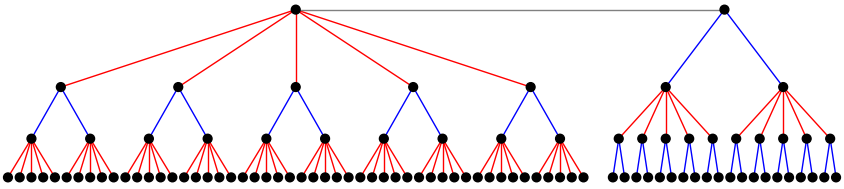}
        \subcaption{$\Delta=6$, $C=1$, $\alpha=1$: then $b_0 = 5$ (red palette), $b_1 = 2$ (blue palette).}
        \label{figure_graph_alpha1}
    \end{subfigure}
    
    \bigskip 
    \begin{subfigure}[t]{\textwidth}
        \centering
        \includegraphics[width=0.7\textwidth]{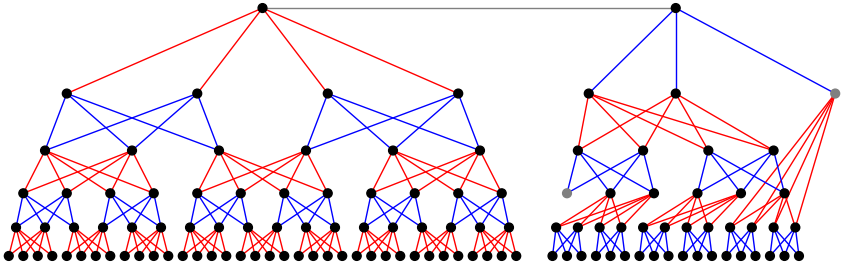}
        \subcaption{$\Delta=6$, $C=1$, $\alpha=2$: then $b_0 = 4$ (red palette), $b_1 = 3$ (blue palette).}
        \label{figure_graph_alpha2}
    \end{subfigure}
    
    \caption{\small{An example of the graph constructed in the proof of Theorem~\ref{theorem_lower_bound_generalized}. (a) In this example $\Delta=6$ and $C=1$ and $\alpha = 1$ (tree), so we split the colors to two palettes: red with $b_0 \equiv \Delta-\alpha = 5$ colors and blue with $b_1 \equiv C+\alpha = 2$ colors. (b) Now $\alpha = 2$, so we get $b_0 = 4$ and $b_1 = 3$. For a third example with $\Delta=5$, $C=1$, $\alpha=2$ ($b_0=b_1=3$), see $\cite{LowerBoundRecourse2018}$. In both figures, the graph $G$ that we construct links two sub-graphs in opposite ``parity'' $p \in \{0,1 \}$, via a single uncolored edge: the top horizontal gray edge. Within each sub-graph the vertices of layers $\ge 1$ (layers start from $0$, top to bottom) are grouped by sets of size $\alpha$, and each set forms a complete bi-clique of size $\alpha \cdot b_{(i+p)\ mod\ 2}$ (depending on the parity) with the next layer, where $b_0 = \Delta-\alpha$ and $b_1 = C+\alpha$. Leftovers from the division by $\alpha$, shown as gray vertices, are grouped with vertices in layer $i+2$. For example, in (b), the rightmost vertex of layer $1$ and the rightmost vertex of layer $3$ form a $2 \times 4$ bi-clique with vertices of layer $4$. The different parity of the two sides is so that every color is being used in the neighbourhood of the uncolored edge, denote it $e$. Regardless of the specific choice of color of each edge from its corresponding palette (and subject to being proper), the coloring cannot be extended to $e$ without recoloring $\Omega(L)$ edges where $L$ is the depth of the deepest full layer in the graph that exists on both sides of $e$. Note that not every vertex is of maximal degree $\Delta$. In (b) the leftover vertex in layer $3$ has no down-edges because we only have $L=5$ layers in this figure.}}
	\label{figure_biregular_graph}
\end{figure*}

\smallskip

\textbf{Bounding the arboricity of $G$, $\alpha_G$.} We show that $\alpha_G \in [\frac{\alpha+1}{2},\alpha]$ regardless of the choice of $b_0$. We bound $\alpha_G$ from below by the arboricity of each of the bi-clique subgraphs of $G$. This arboricity is, by the Nash-Williams theorem \cite{NashWilliamsArboricity1964}, at least $\alpha_G \ge \ceil{\frac{\alpha \cdot b_i}{\alpha + b_i-1}} \ge \ceil{\frac{\alpha \cdot b_i}{2b_i-1}} \ge \frac{\alpha+1}{2}$.\footnote{To clarify how we get $\ge \frac{\alpha+1}{2}$: $y \equiv \frac{\alpha \cdot b_i}{2 b_i - 1} > \frac{\alpha \cdot b_i}{2 b_i} = \frac{\alpha}{2} \equiv x$. So either $x$ is integer and therefore $\ceil{y} \ge \ceil{x} + 1$, or $x$ is half-integer and therefore $\ceil{y} \ge x + \frac{1}{2}$.} To bound $\alpha_G$ from above, note that we can cover $G$ with $\alpha$ forests.
Each such forest is formed as follows.
In each layer $i \ge 1$, each vertex picks a single ``parent'' from a lower layer (typically $i-1$, though potentially $i-3$). This gives a forest which consists of a pair of trees, one rooted at $u$ and the other rooted at $v$. Indeed, by choosing parents from a lower layer, all the vertices eventually connect to one of the roots, and since each vertex except for a root picks a single parent, the number of edges in each component is one less than the number of vertices. The edge $e$ is a bridge, so it does not close a cycle, and we can add it to any (or all) of the forests. We construct $\alpha$ such forests so that
we cover all edges connecting each vertex to a lower layer.

\textbf{$\Omega(L)$ recourse.} To show that the recourse is $\Omega(L)$, we repeat the arguments of \cite{LowerBoundRecourse2018}. Denote the layer of a vertex $x$ by $\ell_x$, and refer to its edge to a neighbour $y$ as \emph{rising} if $\ell_y < \ell_x$, and otherwise as \emph{falling}. Assume for contradiction that no rising edge from layers $L,L-1,L-2$ changes its palette. Consider any vertex $x$ in layer $L-3$, and let $i \equiv (\ell_x\ mod\ 2) = (L-3\ mod\ 2)$. By the definition of $L$, layer $L$ is a full layer, so $x$ has $b_i$ falling edges to vertices in higher layers, either $L-2$ or $L$.  Each falling edge of $x$ is a rising edge from layers $L$ or $L-2$, so by assumption, the palette of each of these $b_i$ falling edges of $x$ is unchanged, and stays $P_i$.  Because $|P_i| = b_i$, any rising edge from $x$ must be colored by a color from the complementary palette $P_{1-i}$, same palette as of its original color. We can continue by induction, and now use the unchanged palette of rising edges from layers $L-3,L-2,L-1$ to prove that rising edges from layer $L-4$ maintain their palette, and so on, until we show that none of the falling edges of $u$ and $v$ change its palette, which means that we cannot color $(u,v)$. So we must recolor (change the palette of) at least one rising edge in every three consecutive layers, so the recourse is at least $\floor{\frac{L}{3}}$.\footnote{The proof of \cite{LowerBoundRecourse2018} is less explicit, stating that we recolor some edge between layers $L-5,\ldots,L$. A rising edge from layer $L-2$ can indeed reach layer $L-5$ if it connects to a leftover vertex from that layer.} 

\textbf{Layers growth rate and lower bounding $L$.} Denote the number of vertices in layer $i$ of $G$ by $n_i$ (recall that layer $i$ in $G$ is the union of layer $i$ in both $H$ sub-graphs of $G$ connected by $e = (u,v)$). For brevity we use $b_0$ and $b_1$, even though we already set $b_1 = \Delta+C - b_0$.

In the special case of $\Delta$ being even, and $\alpha = \frac{\Delta}{2}$ we get: $C=0$ (because $\alpha \le \frac{\Delta-C}{2}$), and also $b_0=b_1=\frac{\Delta}{2}=\alpha$ (because $b_0+b_1=\Delta$ and $\alpha + b_i \le \Delta$). In this case there is no growth: $n_0 = 2$ vertices ($u$ an $v$), and $n_{i} = b_0+b_1 = \Delta$ for $i> 0$. Then $L = \Omega(n/\Delta)$.

If $\alpha < \frac{\Delta}{2}$, still $n_0 = 2$ and $n_1 = b_0 + b_1$. But then $n_2 = b_0 \cdot \floor{\frac{b_1}{\alpha}} + b_1 \cdot \floor{\frac{b_0}{\alpha}}$ and from this point on, the construction yields a recursion that solves for $n_{2i+j} = \Theta(\beta^i \cdot n_j)$ for $j \in \{1,2\}$ (roughly $n_{i+2} \approx \beta \cdot n_i$), where $\beta = \frac{b_0 \cdot b_1}{\alpha^2}$.\footnote{The growth rate is lower bounded by $\frac{n_{i+2}}{n_i} \ge \floor{\frac{b_0}{\alpha}} \cdot \floor{\frac{b_1}{\alpha}}$ but this is a crude underestimation because leftover vertices in each layer are retained for grouping rather than discarded. Moreover, once $n_i$ is sufficiently large, even without retaining leftovers, they amount to a negligible fraction and therefore indeed $n_{i+2} \approx \beta \cdot n_i$.} This holds in each component of $G \setminus \{e\}$ separately and thus together. Our objective is to maximize $L$, so since the growth is exponential in $\beta$, minimizing $\beta$ will increase $L$. Minimizing $\beta$ is simply minimizing a quadratic expression within a bounded domain: Recall that $\beta = \frac{b_0 \cdot b_1}{\alpha^2} = \frac{b_0 \cdot (\Delta+C-b_0)}{\alpha^2}$ and that $ b_0 \le \Delta-\alpha$ and $b_1 = \Delta+C-b_0 \le \Delta-\alpha$ so $b_0 \ge C+\alpha$ (the domain of the quadratic is $C+\alpha \le b_0 \le \Delta-\alpha$). Then $\beta$ is minimized for $b_0 = \Delta-\alpha$ (or symmetrically for $b_0 = C + \alpha$).

Given $\beta$, we proceed to determine $L$. We have that $n_0 = 2$, and for $i \ge 0$: $n_{2i+1} = \Theta( (b_0 + b_1) \cdot \beta^i)$ and $n_{2i+2} = \Theta( \frac{b_0 \cdot b_1}{\alpha} \cdot \beta^i ) = \Theta( \alpha \cdot \beta^{i+1} )$. Then the number of vertices in layers $0$ to $\ell$ is:
$$
\sum_{i=0}^{\ell}{n_i}
\le \sum_{i=0}^{\floor{\ell/2}}{(n_{2i}+n_{2i+1})}
= (2-\alpha) + \sum_{i=0}^{\floor{\ell/2}}{\Theta( (\alpha + b_0 + b_1) \cdot \beta^i)}
$$
$$
= (2-\alpha) + \Theta \Big ( \Delta \cdot \frac{\beta^{\floor{\ell/2} + 1} - 1}{\beta - 1} \Big ) \equiv N(\ell)
$$
Where we used $n_0 = 2 = (2-\alpha) + \Theta(\alpha)$ to conform to the summation of pairs, and $b_0+b_1+\alpha = \Delta+C+\alpha = \Theta(\Delta)$. Next, depth $L+1$ exceeds $n$ vertices, which means that $n \le N(L+1)$, and after re-arranging, including taking-out a logarithm, we get:
$$
\frac{(n + \alpha - 2) \cdot (\beta - 1)}{\Delta} + \Theta(1) \le \Theta(\beta^{\floor{(L+1)/2}+1})
\Rightarrow
\log_{\beta} \Big ( \frac{(n + \alpha - 2) \cdot (\beta - 1)}{\Delta} \Big ) \le L/2 + \Theta(1)
$$
Since $n+\alpha-2 = \Theta(n)$, we get the final expression as claimed by the theorem. The base of the logarithm $\beta = \frac{(\Delta-\alpha)(C+\alpha)}{\alpha^2}$ can be replaced by $\gamma \equiv \frac{\Delta}{\alpha} - 1$ because
$\gamma \le \beta = \gamma \cdot \frac{C+\alpha}{\alpha} \le \gamma^2$. Indeed, $\frac{C+\alpha}{\alpha} \le \gamma$ because $\alpha \le \frac{\Delta-C}{2}$.
\end{proof}

As mentioned, Theorem~\ref{theorem_lower_bound_2018} is a special case of Theorem~\ref{theorem_lower_bound_generalized}. To get Theorem~\ref{theorem_lower_bound_2018}, assume that $\Delta+C$ is even, and let $\alpha = \frac{\Delta-C}{2}$. Then we get: $\beta = \frac{(\Delta+C)^2}{(\Delta-C)^2} \Rightarrow \beta-1 = \frac{4 C \Delta}{(\Delta-C)^2}$, which yields a recourse of $\Omega \Big ( \frac{\log \frac{Cn}{(\Delta-C)^2}}{\log \frac{\Delta+C}{\Delta-C}} \Big )$. If $C \le \frac{\Delta}{3}$, we can simplify to $\Omega(\frac{\Delta}{C} \log \frac{Cn}{\Delta^2})$.

\subsection{Shift-based Recourse Separation}
\label{appendix_subsection_recourse_separation}

In this subsection we show the shift-based recourse separation, and prove the generalization of Theorem~\ref{theorem_gap_shift_recoloring_versus_general_separation_simple}. We start with a definition for a generic recourse lower bound.

\begin{definition}
\label{definition_worst_case_recourse_function}
Let $R(n,\Delta,C)$ be the minimum recourse required to complete the coloring of a worst-case input of a graph 
 $H$ with $n$ vertices, maximum degree $\Delta$, and a $(\Delta+C)$ edge coloring of all edges except for one.
\end{definition}

 Theorem~\ref{theorem_lower_bound_2018} implies that $R(n,\Delta,C) =\ $\separationExpression{}.

\begin{restatable}[Shift Recourse Gap (General)]{theorem}{theoremRecourseGapShiftVersusGeneral}
\label{theorem_gap_shift_recoloring_versus_general_separation}

Let $\Delta \ge 3$ and $0 \le C \le \Delta-3$. For sufficiently large $n$ there exists a graph $G$ with maximum degree $\Delta$, that is $(\Delta+C)$ edge colored except for one edge $e$, and satisfies the following property. We can extend the coloring of $G$ to $e$ with only $O(1)$ recourse, but any shift-based algorithm has recourse of at least $\max_{q \in [1,\Delta-C-2] \cap \mathbb{N}}{R(\floor{n/(q+1)},\Delta-q,C)}$ where $R$ is defined in Definition~\ref{definition_worst_case_recourse_function}. In particular, for $C=0$ and $n \ge 4(\Delta-1)$ the recourse is $\Omega(n/\Delta)$, and by Theorem~\ref{theorem_lower_bound_2018} for $\Delta \ge 4$ and $C \ge 1$ and $n \ge 2(\Delta+C+1)$, we have that the recourse is \separationExpression{}.
\end{restatable}

Note that Theorem~\ref{theorem_gap_shift_recoloring_versus_general_separation} is stated for $C \le \Delta-3$. The case $C = \Delta-2$ is special because then an uncolored edge $e$ with two same-colored adjacent edges has an available color to use and finish coloring. On the other hand, if $e$ has no available color, then choosing its color simply shifts this color from a neighbour edge. Therefore, when focusing on a single update, it only makes sense to differentiate general and shift-based algorithms if $C \le \Delta-3$. However, in an amortized context, it may be relevant to also consider $C = \Delta-2$ since operations such as fully recoloring a graph or parts of it are not shift-based.

\begin{proof}
We divide the vertices of $G$ into $q+1$ equal and disjoint sets, $V = \{v_1,\ldots,v_k \}$ and $U^i = \{ u^i_1,\ldots,u^i_k\}$ for $i=1,\ldots,q$, $k \equiv \floor{\frac{n}{q+1}}$, ignoring up to $q$ leftover vertices. The edges of $G$ are defined as the union of $q+1$ sets $E_0 \cup E_1 \cup \ldots \cup E_q$. The set $E_i$ for $ i \ge 1$ is a matching between $V$ and $U^i$, formally: $E_i = \{ (v_j,u^i_j) \ | \ j \in [k] \}$. Each such matching is colored by a dedicated color. 
The edges of $E_0$ are incident to vertices of $V$, such that $V$ induces a graph $H$,  which is $(\Delta-q + C)$-colored with a single uncolored edge, that yields a recourse of at least $R(\floor{n/(q+1)},\Delta-q,C)$ (Definition \ref{definition_worst_case_recourse_function}). Note that the colors of the matchings $E_i$, $i\ge 1$ are not used in $H$, and $H$ has maximum degree $\Delta' \equiv \Delta-q$ because each $v_i$ already has $q$ neighbours $u^1_i,\ldots,u^q_i$. The condition $q \le \Delta-C-2$ guarantees that (rearranged) $C \le \Delta'-2$. See Figure~\ref{figure_recourse_gap_on_a_path} in Section~\ref{section_overview} for a visualization for $\Delta=3$ (thus $C=0$) with $q=1$.

\looseness=-1
Let $e=(v_{i_1},v_{i_2})$ be the only uncolored edge (for some $i_1,i_2$). We can easily extend the coloring to it by coloring it by any of the dedicated matching colors, say of $E_j$, and recoloring $(v_{i_1},u^j_{i_1})$ and $(v_{i_2},u^j_{i_2})$ with some other color available at $v_{i_1}$ and $v_{i_2}$ respectively ($u^j_{i_1}$ and $u^j_{i_2}$ are leaves and every color is available to them), thus the recourse is $O(1)$. However, when restricted to shifting colors, observe that we cannot use the dedicated matching colors, because by induction the uncolored edge is always between two vertices of $V$, and thus has two adjacent neighbours of each of the dedicated colors. Therefore, the problem reduces to recoloring only edges from $E_0$, and thus the recourse of any shift-based recoloring is at least $R(\floor{n/(q+1)},\Delta-q,C)$.

For $\Delta \ge 4$ and $C \ge 1$ we can choose $q=1$ and use Theorem~\ref{theorem_lower_bound_2018}, which requires $\floor{\frac{n}{q+1}} = \floor{\frac{n}{2}} \ge (\Delta-q + C + 2) = (\Delta+C+1)$ which is satisfied for $n \ge 2(\Delta+C+1)$. For $C=0$, observe that choosing $q=\Delta-2$ leaves us with $2$ colors for the subgraph $H$, which the adversary can maintain as a bicolored path of length $L=\floor{\frac{n}{\Delta-1}}$ over which the recourse is $\Omega(L)$. For a non-zero recourse the path must have at least $3$ edges, or $4$ vertices, each having $q$ additional unique neighbours. Then this construction requires $n \ge 4(q+1) = 4(\Delta-1)$.
\end{proof}

\bibliography{reference}

\end{document}